\documentclass[aps, prx, reprint, amsmath,amssymb,showpacs,floatfix,longbibliography, twocolumn, superscriptaddress,nofootinbib]{revtex4-1}
\PassOptionsToPackage{dvipsnames,table,xcdraw}{xcolor}
\usepackage{cmap}
\usepackage{simplewick}
\usepackage{algorithmic}
\usepackage{multirow} 
\usepackage{mathrsfs}
\usepackage{oplotsymbl}

\usepackage{graphicx}
\graphicspath{ {./figures/} }
\usepackage{float}
\usepackage{caption}
\usepackage{dcolumn}
\usepackage{bm}
\usepackage{amssymb}
\usepackage{tabularx}
\usepackage{mathtools,halloweenmath}
\usepackage{physics}
\usepackage{graphicx}
\usepackage{amsmath}
\usepackage[version=4]{mhchem}
\usepackage{siunitx}
\usepackage{tabularx}

\usepackage{braket}
\usepackage{wick}
\usepackage{comment}
\usepackage{hyperref} 
\hypersetup{
    colorlinks,
    linkcolor={red!50!black},
    citecolor={blue!50!black},
    urlcolor={blue!80!black}
}
\usepackage{multirow}
\usepackage{lipsum}
\usepackage{array}
\usepackage{microtype}
\usepackage{float}
\usepackage{appendix}
\usepackage{amsmath}
\usepackage{comment}
\usepackage{amsthm}
\usepackage{tikz}
\usepackage{tikz-cd}

\usepackage{enumitem}
\usepackage{eczoo}
\usetikzlibrary{quantikz2}

\usetikzlibrary{external}

\usepackage{amscd} 
\usepackage{empheq}
\usepackage{amsthm}
\usepackage{bbm} 
\usepackage{bold-extra}
\usepackage{booktabs} 
\usepackage{circledsteps} 
\usepackage[justification=Justified,singlelinecheck=false]{caption}

\def\prg#1{\paragraph*{{\bf #1}}}

\newcommand{\Z}{{\mathbb{Z}}}
\newcommand{\R}{{\mathbb{R}}}

\newcommand{\T}{{\mathbb{T}}}
\newcommand{\N}{\mathbb{N}}

\usepackage[ruled,lined]{algorithm2e}

\usepackage{mathtools}

\newcommand{\cd}{\textsc{cdisp}}
\newcommand{\st}{\hat{S}}
\global\long\def\id{1}%
\global\long\def\zl{\overline{Z}}%
\global\long\def\xl{\overline{X}}%
\definecolor{greenvva}{RGB}{0,128,0}%
\renewcommand{\dv}{\textsc{dv}}
\newcommand{\cv}{\textsc{cv}}
\newcommand{\dia}[1]{\boldsymbol{#1}}
\global\long\def\ts{T}%

\newtheorem{theorem}{Theorem}
\newtheorem{lemma}[theorem]{Lemma}

\newtheorem{definition}{Definition}

\newtheorem{remark}{Remark}

\newcommand{\vva}[1]{#1}
\newcommand{\sch}[1]{#1}

\DeclareUnicodeCharacter{202F}{\,}
\begin{document}

\title{Hybrid Oscillator-Qudit Quantum Processors:\\ stabilizer states, stabilizer codes, symplectic operations, \vva{and non-commutative geometry}}
\author{Sayan Chakraborty}
\affiliation{Institute for Advancing Intelligence, TCG CREST, Sector V, Salt Lake, Kolkata 700091, India}
\author{Victor V. Albert}
\affiliation{Joint Center for Quantum Information and Computer Science,
NIST/University of Maryland, College Park, Maryland 20742, USA}

\begin{abstract}

We construct stabilizer states and error-correcting codes on combinations of discrete- and continuous-variable systems, generalizing the Gottesman-Kitaev-Preskill (GKP) quantum lattice formalism. Our framework absorbs the discrete phase space of a qudit into a hybrid phase space parameterizable entirely by the continuous variables of a harmonic oscillator. The unit cell of a hybrid quantum lattice grows with the qudit dimension, yielding a way to simultaneously measure an arbitrarily large range of non-commuting position and momentum displacements. Simple hybrid states can be obtained by applying a conditional displacement to a Gottesman-Kitaev-Preskill (GKP) state and a Pauli eigenstate, or by encoding some of the physical qudits of a stabilizer state into a GKP code. The states' oscillator-qudit entanglement cannot be generated using symplectic (i.e., Gaussian-Clifford) operations, distinguishing them as a resource from tensor products of oscillator and qudit stabilizer states. 

\vva{
Simple hybrid codes can be thought of as subsystem GKP codes whose gauge factor is entangled with a qudit.
Our numerical investigations suggest that such codes can sometimes outperform GKP codes against physical noise, and their decoders can be tuned to accommodate either more qudit or more oscillator errors.
We also relate stabilizer codes to non-commutative tori, identifying that a general construction of such tori yields multi-mode multi-qudit extensions of GKP codes.
We explicitly calculate these codes' logical dimension and logical operators by utilizing the Morita equivalence between their stabilizer and logical tori.
} We provide examples using commutation matrices, integer symplectic matrices, and binary codes.

\end{abstract}

\maketitle

\tableofcontents

\section{Introduction \& summary}\label{Sec:Intro}

Quantum effects have proven useful in select applications~\cite{shor1999polynomial,bennett2014quantum,huang2021power,nisqrmp,dalzell2023quantum,yamakawa2024verifiable,jordan2024optimization}, and their stabilization on a scale necessary for such applications is increasingly becoming a reality~\cite{arute2019quantum,bluvstein2024logical,google2025quantum,putterman2025hardware}.

Quantum effects can be realized using discrete or continuous degrees of freedom.
The smallest quantum system --- a qubit --- consists of two states and is a special case of a multi-level system called a qudit.
The label of a qudit's defining states can be thought of as a discrete-variable (\(\dv\)) ``position''.
On the other hand, the position of a continuous-variable (\(\cv\)) quantum system, like a particle on a line or a harmonic oscillator, is continuous.

\begin{figure}[t!]
    \centering
    \includegraphics[width=0.8\columnwidth]{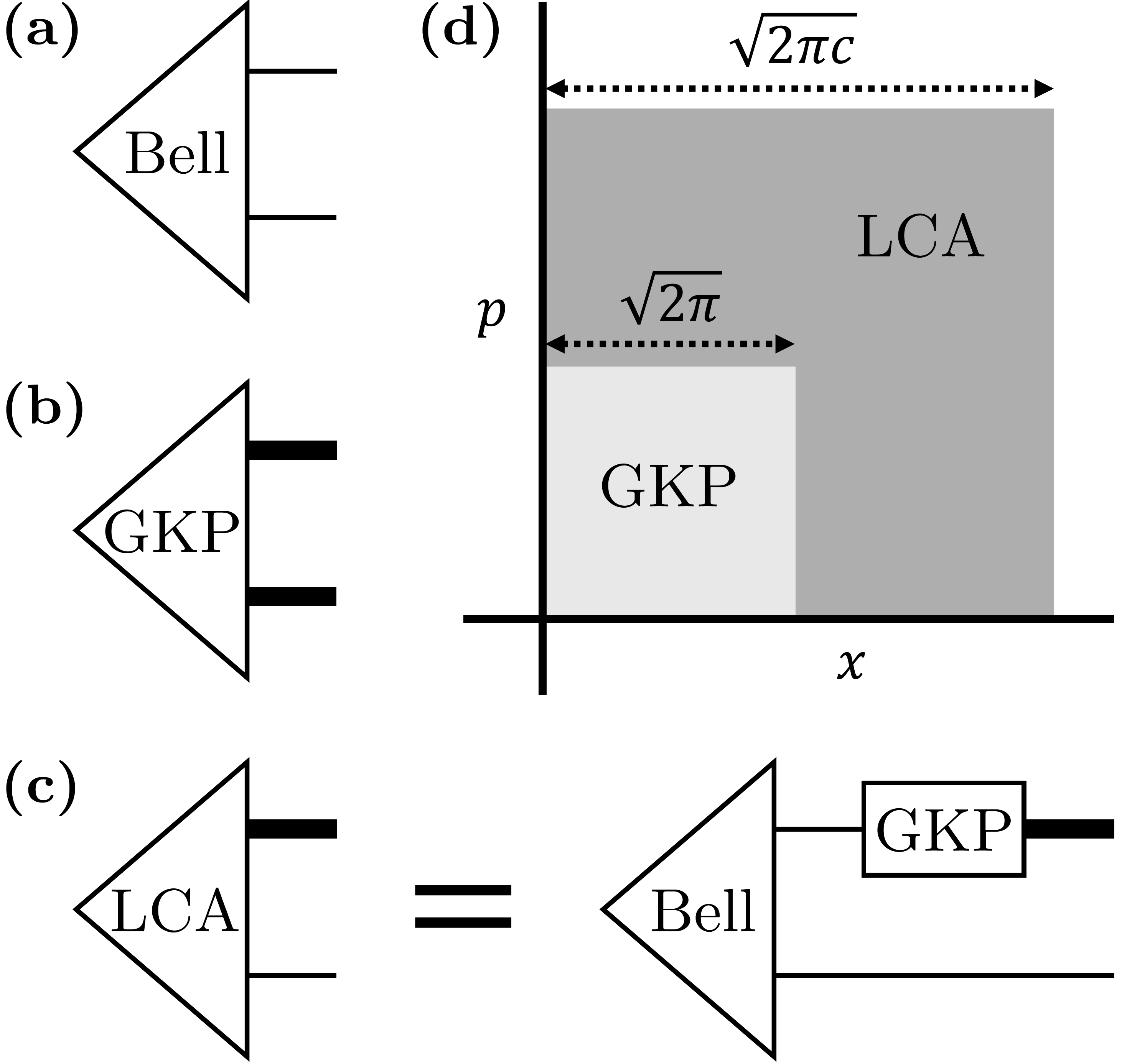}
    \caption{\small
    \textbf{(a)} Ordinary Bell states and \textbf{(b)} encoded GKP Bell states are defined on qudit (thin line) and oscillator (thick line) systems, respectively.
    \textbf{(c)} Simple locally compact Abelian (LCA) states interpolate between these purely \(\dv\) and \(\cv\) states since they can be obtained by encoding one half of a Bell state into a GKP code.
    LCA states cannot be obtained from separable oscillator-qudit states via symplectic (i.e., Gaussian-Clifford) operations, underscoring the non-Gaussianity of the GKP encoder (cf.~\cite{baragiola2019all}).
    \textbf{(d)} The LCA framework absorbs the qudit into the oscillator to form a hybrid phase space that is parameterizable by the oscillator's original variables. 
    The minimal area enclosed by two commuting displacements on this phase space is \(2\pi c\), where \(c\) is the qudit dimension.
    This means position and momentum can be simultaneously measured modulo \(\sqrt{2\pi c}\) for arbitrarily large \(c\).
    }
    \label{fig0_circuits}
\end{figure}

Quantum devices admit both discrete and continuous degrees of freedom.
The orbital, nuclear-spin, electronic-spin, and electronic orbital degrees of freedom of atoms (real  or artificial~\cite{girvin2014circuit}), ions, and molecules are typically thought of as \(\dv\) systems, while their vibrational and, in the case of molecules, rotational modes are best described by continuous variables.
The electromagnetic field has both an amplitude, described by continuous quadratures, and a polarization, described by two directions at a time.

Discrete-variable degrees of freedom are routinely used as ancillas for continuous-variable ones, and vice versa.
But using \(\cv\) and \(\dv\) degrees of freedom on a more equal footing can combine advantages superadditively to yield more powerful quantum devices~\cite{forn2019ultrastrong,wallquist2009hybrid,xiang2013hybrid,andersen2015hybrid,deng2017quantum,blais2021circuit}.
To name a few recent examples, hybrid systems can (in principle) factor integers with surprisingly low overhead~\cite{brenner2024factoring,brenner2025trading}, they can provide
increased error-correcting capability~\cite{park2012quantum,sychev2018entanglement,omkar2020resource,omkar2021highly,lee2024fault}, 
and they have been leveraged to simulate gauge theories~\cite{farrell2023preparations1,farrell2023preparations2,crane2024hybrid,araz2025hybrid,kemper2025hybrid} and vibrational spectra~\cite{huh2015boson,vu2025computational} with real devices~\cite{wang2020efficient,wang2023observation,huie2025three}. 

Most investigations of hybrid systems so far have been rooted in the particular details of their underlying quantum device. 
One exception to this line of thought is the manuscript~\cite{liu2024hybrid}, which distills computational primitives of hybrid systems from a general and platform-agnostic point of view.
However, other canonical primitives, such as uncertainty relations,
stabilizer states, and symplectic operations, remain to be fully fleshed out.

We continue the direction pushed by Ref.~\cite{liu2024hybrid} and collect hybrid analogues of  ``easy'' states and ``easy'' operations on hybrid \(\cv\)-\(\dv\) systems.
These are, respectively, stabilizer states and symplectic (i.e., Gaussian-Clifford) operations.
By \eczoohref[“stabilizer”]{stabilizer}, we mean states that form joint +1-eigenvalue eigenspaces of a commuting group of oscillator displacements and qudit Pauli strings.

\prg{LCA states}
We define a class of
hybrid stabilizer states and error-correcting codes that are a fusion of 
quantum lattice states [a.k.a.\@ \eczoohref[Gottesman-Kitaev-Preskill (GKP)]{quantum_lattice} or Zak basis states]~\cite{gel1950expansion,zak1967finite,aharonov1969modular,10.1103/physreva.64.012310} with \eczoohref[modular-qudit stabilizer]{qudit_stabilizer} states~\cite{gottesman1997stabilizer,gibbons2004discrete,hostens2005stabilizer,nadella2012stabilizer}. 
We examine the properties of a class of experimentally relevant single-mode single-qudit LCA states and codes we call \textit{simple}.
A qudit's phase space is discrete,
but the LCA framework absorbs it into a hybrid phase space parameterizable entirely by the continuous variables of the oscillator. 
This hybrid phase space allows one to measure a larger set of non-commuting oscillator displacements than is possible with a purely oscillator system~\cite{aharonov1969modular}, assuming no qudit errors occur during measurement.
This should yield more powerful protocols for measuring displacements than the recently realized~\cite{valahu2024quantum} GKP-based protocols~\cite{duivenvoorden2017single,zhuang2020distributed,labarca2025quantum}.

We review the structure of symplectic (i.e., stabilizer-preserving) operations, highlighting that they cannot entangle oscillators and qudits~\cite{prasad2008decomposition,bermejo2016normalizer}.
This makes LCA states a distinct resource from tensor products of \(\dv\) and \(\cv\) stabilizer states because each LCA state is entangled with respect to the oscillator-qudit bipartition.
We show that displaced versions of LCA states form an entangled basis~\cite{wallach2000unentangled,parthasarathy2004maximal,parthasarathy2005extremal,walgate2008generic} for an oscillator-qudit system, much like Bell states do for two-qudit systems. 

LCA states can be prepared with readily realizable non-Gaussian resources.
A simple example of an LCA state is a qudit stabilizer state for which a subset of qudits has been encoded into a GKP code (see Fig.~\ref{fig0_circuits}).
LCA states can also be obtained by applying an oscillator-qudit conditional-displacement gate to a GKP state and a Pauli eigenstate.
Such states occur mid-circuit in syndrome extraction protocols based on this gate~\cite{terhal2016encoding}.
In this sense, LCA states have effectively been experimentally realized in the trapped-ion~\cite{de2022error,valahu2024quantum,matsos2024robust,matsos2024universal} and superconducting-circuit~\cite{campagne2020quantum,eickbusch2022fast,sivak2023real,lachance2024autonomous,brock2025quantum} devices implementing such protocols.

\prg{LCA codes}
\vva{
Extending LCA states, our simple LCA codes are constructed from certain GKP codes by first performing a bipartite logical subsystem decomposition on the GKP code and then entangling one of the factors in this decomposition with a qudit.
We show the number-theoretic properties necessary for such decompositions to be possible, yielding simple LCA code families for certain combinations of physical and logical qudit dimensions.

We numerically evaluate the entanglement fidelity of recovery of LCA and GKP codes under the Petz recovery (a.k.a. transpose recovery) for the combination of photon loss on the mode and qudit amplitude damping.
This fidelity is a good proxy for the optimal entanglement fidelity, and our experiments suggest that LCA codes can outperform GKP codes against physical noise for certain parameters (see Sec.~\ref{subsec:numerics}).

Simple LCA codes are special cases of our general multi-mode multi-qudit LCA codes. 
The general construction is obtained by relating stabilizer codes to non-commutative tori, a type of algebra studied in non-commutative geometry and string theory~\cite{rieffel,schwarz1998morita,connes1998noncommutative,rieffel1999morita,seiberg1999string,yoneya2000string,li2004strong,elliott2008strong}.
We utilize Morita equivalence --- a structural relation between algebras --- to obtain logical operators for a given LCA stabilizer group (see Thm.~\ref{thm:main_text}).
We explicitly calculate the logical dimension of the code, expressing it in terms of a Pfaffian.

The parallel between stabilizer error correction and non-commutative geometry --- our main conceptual and technical result --- is summarized in Table~\ref{tab:nc}.
Infinite families of lattice-like codes with guaranteed finite logical dimension are possible over arbitrary qudit structures, and it is likely that exotic and interesting hybrid lattices will lie along the path we lay out.

LCA states and codes combine several interesting and useful features of \(\cv\) and \(\dv\) systems --- lattice structure, Fourier duality, and behavior under symplectic transformations, to name a few.
We anticipate that LCA states may find applications anywhere such features are used.
We hope that our work will spearhead future investigations into fault-tolerant hybrid quantum algorithms and quantum and post-quantum cryptographic protocols~\cite{regev2009lattices,conrad2024good,kuperberg2025hidden}.
}

We overview the key features of simple LCA states in Sec.~\ref{sec:lca-states}.
We then define LCA codes, first for a single-mode single-qudit space in Sec.~\ref{sec:single-mode}, and then using a simple multi-mode extension in Sec.~\ref{sec:standard-form}.
We present a general multi-mode LCA code construction in Sec.~\ref{sec:multimode}.
We discuss symplectic physical and logical gates in Sec.~\ref{sec:logical-gates}.

\begin{table}[t]
\begin{tabular}{cc}
\toprule 
Quantum error correction & Non-commutative geometry\tabularnewline
\midrule
oscillator-qudit phase space & LCA group\tabularnewline
stabilizer algebra & commutative torus\tabularnewline
logical algebra & ~~dual non-commutative torus~~\tabularnewline
stabilizer-to-logical mapping & Morita equivalence\tabularnewline
\bottomrule
\end{tabular}

\caption{
\label{tab:nc}
\vva{
Oscillator-qudit phase spaces are labeled by positions and momenta
of modes, as well as computational and Fourier-dual basis-state labels
of qudits of various dimensions $c_{j}$. These subsystems each contribute
factors of $\mathbb{R}\times\mathbb{R}$ and $\mathbb{Z}_{c_{j}}\times\mathbb{Z}_{c_{j}}$
to the overall phase space, which forms an instance of a
locally compact abelian (LCA) group. The (commuting) stabilizers of
an LCA code are an instance of a commutative torus, while the logical
oscillator-qudit displacements realize a non-commutative torus dual to the stabilizer torus. The duality between them is an instance
of Morita equivalence, a special case of which is the inversion mapping between
the GKP code's lattice $\varTheta$ and its dual, the logical lattice
$\varTheta^{-1}$ (see Sec.~\ref{sec:multimode} and  Thm.~\ref{thm:main_text}).
}
}
\end{table}

\section{Simple LCA states}\label{sec:lca-states}

The key features of LCA states and codes are already evident for the case of a single oscillator and qubit.
We overview this case before concluding this section with the natural qudit extension.

A qubit stabilizer state is a \(+1\)-eigenstate of a group of commuting Pauli operators.
All single-qubit Pauli strings anticommute, so a state whose stabilizer group contains at least two non-identity elements must be defined on two or more qubits.
The smallest example of such a state is the Bell state \(|00\rangle + |11\rangle\), stabilized by \(\hat{\sigma}_{\mathsf{x}} \otimes \hat{\sigma}_{\mathsf{x}}\) and \(\hat{\sigma}_{\mathsf{z}} \otimes \hat{\sigma}_{\mathsf{z}}\) (tensor product will be implied from now on).
Here, \(\hat{\sigma}_{\mathsf{x}}\) and \(\hat{\sigma}_{\mathsf{z}}\) are the standard unitary qubit Paulis.

An oscillator stabilizer state is a \(+1\)-eigenstate of a group of commuting displacement operators.
Let \(\exp(-i a \hat{p})\) and \(\exp(i b \hat{x})\) be displacement operators for a single mode, defined using momentum and position operators, \(\hat p\) and \(\hat x\), respectively, and real displacements \(a,b\).
The simplest example, often called the \textit{qunaught} GKP state~\cite{walshe2020continuous}, is the unique \(+1\)-eigenstate of the two commuting single-mode displacements for which \(a = b = \sqrt{2\pi}\).
Moving along a square using these two displacements traces out an area of \(2\pi\), yielding a group commutator (a.k.a. Berry phase) of \(\exp(i 2\pi) = 1\).
The corresponding qunaught GKP state is a superposition, or comb, of oscillator positions
\begin{equation}\label{eq:gkp}
    |\varnothing\rangle=\sum_{\ell\in\mathbb{Z}}|\ell\sqrt{2\pi}\rangle_{\hat{x}}\quad\quad\text{(qunaught state)},
\end{equation}
where \(|x\rangle_{\hat x}\) are (non-normalizable) oscillator position vectors~\cite{albert2025bosonic}.

A hybrid oscillator-qubit stabilizer state is a \(+1\)-eigenstate of a group of commuting tensor products of displacement and Pauli operators.
A trivial example is a tensor product of a Bell state with the qunaught GKP state, whose hybrid stabilizer group is the tensor product of the corresponding oscillator and qubit groups.
This state is separable since the qubit and oscillator factors of its group elements are independent.
In other words, commutation is ensured since 
the displacements and Pauli elements commute individually and independently of each other.

An LCA stabilizer state admits a stabilizer group whose oscillator and qubit factors \textit{do not} commute individually.
The simplest such group is generated by the pair
\begin{equation}\label{eq:lca-qubit}
    \st_{X}=e^{-i\sqrt{\pi}\hat{p}}\hat{\sigma}_{\mathsf{x}}\quad\quad\text{and}\quad\quad \st_{Z}=e^{i\sqrt{\pi}\hat{x}}\hat{\sigma}_{\mathsf{z}}~.
\end{equation}
The oscillator parts of the above stabilizers anticommute since the phase-space area of the square traced out by them is \(\pi\).
This phase is exactly compensated by the anticommutation of the two Paulis.

\begin{figure}[t]
    \centering
    \includegraphics[width=0.95\columnwidth]{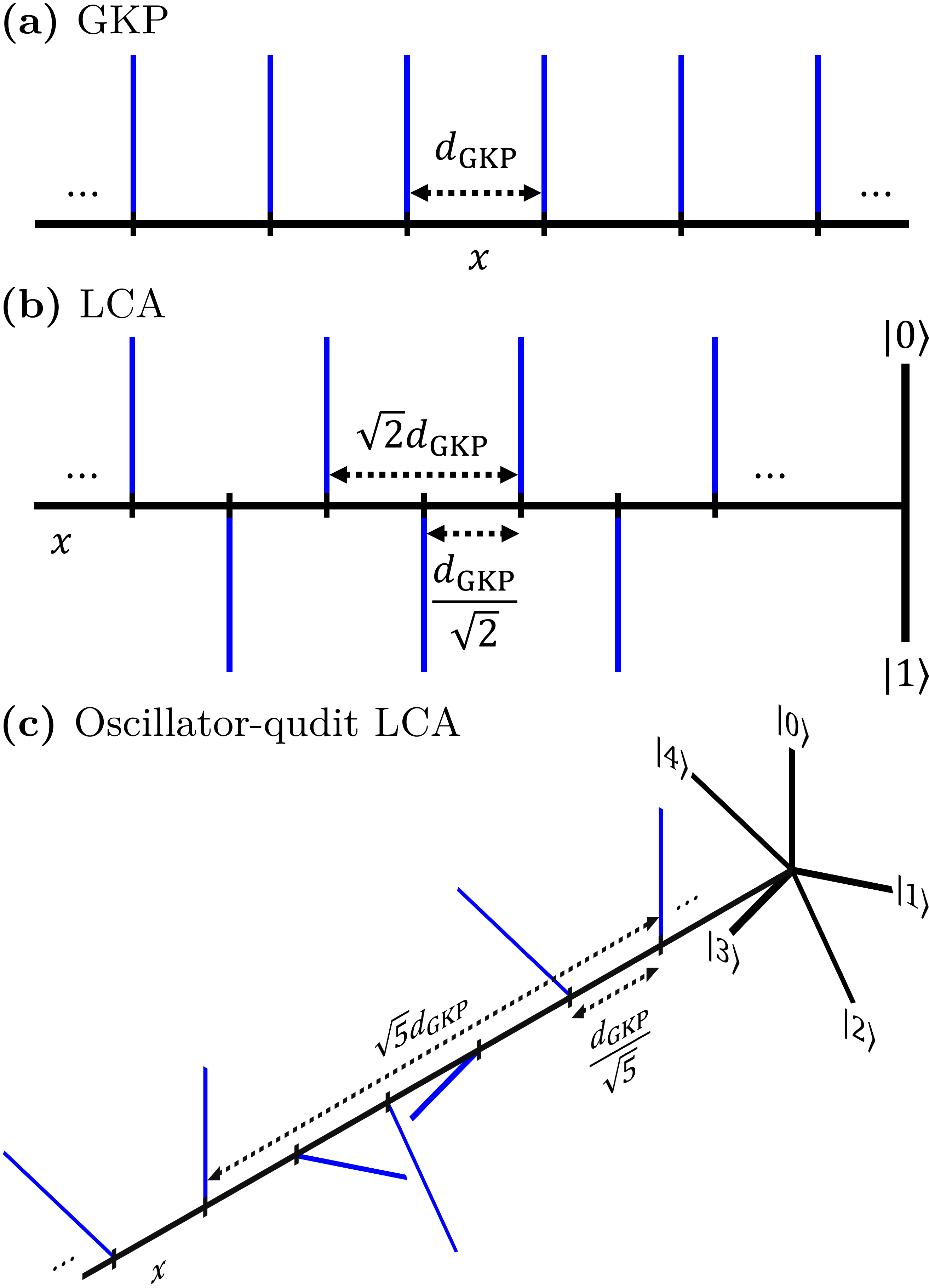}
    \caption{\small
    \textbf{(a)} A GKP qunaught state \(|\varnothing\rangle\)~\eqref{eq:gkp} is a superposition, or \textit{comb}, of oscillator position vectors spaced uniformly at a distance of \(d_{\text{GKP}} = \sqrt{2\pi}\). 
    \textbf{(b)} A simple oscillator-qubit LCA state~\eqref{eq:lca-qubit-state} is a superposition of two such combs, each tensored with the qubit state \(|0\rangle\) or \(|1\rangle\).
    The spacing within each comb is a factor of \(\sqrt{2}\) larger than that of the qunaught state, while the spacing between combs is smaller by the same factor.
    \textbf{(c)} Simple oscillator-qudit LCA states~\eqref{eq:qudit-lca-state}, depicted here for qudit dimension \(c=5\), are a superposition of \(c\) combs tensored with the \(c\) qudit basis states.
    The intra-comb spacing is a factor \(\sqrt{c}\) larger than that of the qunaught state.
    }
    \label{fig1_lca}
\end{figure}

The unique state stabilized by the above group is
\begin{subequations}\label{eq:lca-qubit-state}
\begin{align}
|\text{LCA}\rangle&=\sum_{\ell\in\mathbb{Z}}|\ell\sqrt{\pi}\rangle_{\hat{x}}\left|\ell\text{ mod }2\right\rangle \\&=\sum_{s\in\mathbb{Z}}{|{(2s)\sqrt{\pi}}\rangle_{\hat{x}}}\left|0\right\rangle +{|{(2s+1)\sqrt{\pi}}\rangle_{\hat{x}}}\left|1\right\rangle ~,
\end{align}
\end{subequations}
where unlabeled kets are qubit canonical basis states \(|0\rangle\) and \(|1\rangle\).
The first line shows that the teeth of the oscillator position-vector comb are sorted among the qubit states according to parity via the mapping,
\begin{equation}\label{eq:mapping-lattice}
\ell \to (\ell\sqrt{\pi},\ell\text{ mod }2)~,    
\end{equation}
for the \(\ell\)th tooth.
This is an embedding of the integer lattice \(\mathbb{Z}\) into the group \(\mathbb{R}\times\mathbb{Z}_2\) labeling the oscillator-qubit configuration space (cf.~\cite{albert2020robust}).
This embedding yields the two 0-state and 1-state subcombs in the second equation, consisting of even and odd multiples of \(\sqrt{\pi}\), respectively.

The above state corresponds to an embedding of a square lattice since the displacement parts of the stabilizers~\eqref{eq:lca-qubit} are both by the same amount.
One can, of course, shape the lattice to be rectangular by applying a squeezing operation.

The LCA state is not normalizable since each position vector is not normalizable, and since each comb contains an infinite number of such vectors.
LCA states can be normalized in the same way as GKP states, namely, by replacing each position vector with a finitely squeezed state, and by modulating each comb of position vectors with a Gaussian envelope~\cite{10.1103/physreva.64.012310,menicucci2014fault,albert_performance_nodate} (see Sec.~\ref{sec:approximate}).

\subsection{Larger detectable displacement}

The mapping~\eqref{eq:mapping-lattice} reveals that the spacing of LCA states must be different than that of GKP states --- a single \vva{integer shift} in \(\ell\) does not map the lattice back onto itself.
Instead, two \vva{shifts} are required to match the second (modular) coordinate.

Observing Eq.~\eqref{eq:lca-qubit-state}, the spacing between the teeth of each subcomb in the LCA state is \(2\sqrt{\pi}\).
This means that a displacement of \(2\sqrt{\pi}\) is required to shift a tooth of each comb to the location of its neighbor on the same comb.
This is  a factor of \(\sqrt{2}\) larger than that of the qunaught state~\eqref{eq:gkp}.
The displacement by \(2\sqrt{\pi}\) is a stabilizer obtained by squaring \(\st_{X}\)~\eqref{eq:lca-qubit}.

Shorter displacements \(\exp(-ia\hat{p})\) with \(|a|<2\sqrt{\pi}\) shift the combs away from their original positions, with the exception of \(|a|=\sqrt{\pi}\).
Such a shift exchanges the two subcombs, but the resulting state,
\begin{align}\label{eq:qubit-error}
e^{-i\sqrt{\pi}\hat{p}}|\text{LCA}\rangle&=\sum_{s\in\mathbb{Z}}{|{(2s+1)\sqrt{\pi}}\rangle_{\hat{x}}}\left|0\right\rangle +{|{(2s)\sqrt{\pi}}\rangle_{\hat{x}}}\left|1\right\rangle  \nonumber\\
&=\hat{\sigma}_{\mathsf{x}}|\text{LCA}\rangle~,
\end{align}
is orthogonal to the LCA state because the two subcombs are paired with the wrong qubit states.
Therefore, all shifts up to (and excluding) \(2\sqrt{\pi}\) yield states orthogonal to the original LCA state.
The shifted states at \(a=\pm\sqrt{\pi}\) are the same state, and the continuum of distinct shifted states can be labeled by either \(0\leq a < 2\sqrt{\pi}\) or \(-\sqrt{\pi} \leq a < \sqrt{\pi}\).

The same effect is observed for momentum shifts by writing the LCA state in the qubit \(|\pm\rangle\) and oscillator momentum bases \(|p\rangle_{\hat p}\),
\begin{equation}
    |\text{LCA}\rangle = 2\sqrt{\pi} \sum_{s\in\mathbb{Z}}|(2s)\sqrt{\pi}\rangle_{\hat{p}}\left|+\right\rangle +|(2s+1)\sqrt{\pi}\rangle_{\hat{p}}\left|-\right\rangle ~.
\end{equation}
The above can be obtained by inserting a resolution of the identity in terms of the two bases in Eq.~\eqref{eq:lca-qubit-state}.
Momentum shifts \(0 < |b| < \sqrt{\pi}\) and \(\sqrt{\pi} < |b| < 2\sqrt{\pi}\) yield orthogonal subcombs, the shift by \(|b|=\sqrt{\pi}\) exchanges the subcombs but still yields an orthogonal state, and the shift by \(|b|=2\sqrt{\pi}\) is the square of the stabilizer~\(\st_{Z}\)~\eqref{eq:lca-qubit} or its adjoint.

\begin{table*}[t]
\begin{tabular}{cccccc}
\toprule 
~~System~~~ & $|\psi\rangle$ & $U$ & \(U\) symplectic? & $U|\psi\rangle$ & \(U|\psi\rangle\) stabilizer?\tabularnewline
\midrule
$\Z_{c}\times\Z_{c}$ & $|0\rangle_{\phantom{\text{GKP}}}|+\rangle_{\phantom{\text{GKP}}}$ & $\textsc{cnot}$ & {\color{greenvva}\checkmark} & qudit Bell & {\color{greenvva}\checkmark}\tabularnewline
$\R_{\phantom{c}}\times\R_{\phantom{c}}$ & $|\overline{0}\rangle_{\text{GKP}}|\overline{+}\rangle_{\text{GKP}}$ & $\exp(i\hat{p}_{1}\hat{x}_{2})$ & {\color{greenvva}\checkmark} & GKP Bell & {\color{greenvva}\checkmark}\tabularnewline
$\R_{\phantom{c}}\times\Z_{c}$ & $|\overline{0}\rangle_{\text{GKP}}|+\rangle_{\phantom{\text{GKP}}}$ & $\cd_{\sqrt{2\pi/c}}$~\eqref{eq:qudit-cond-disp} & $\mathbin{\color{red}\times}$ & LCA & {\color{greenvva}\checkmark}\tabularnewline
\bottomrule
\end{tabular}

\caption{
Table comparing two-qudit, two-oscillator, and oscillator-qudit stabilizer states, each created by an entangling unitary acting on a separable state.
The first state is a two-qudit Bell state, created by the Clifford \(\textsc{cnot}\) gate acting on a tensor product of qudit \(Z\) and \(X\) eigenstates.
The second is a GKP logical Bell state, created by a Gaussian conditional displacement acting on encoded GKP eigenstates.
The third is an LCA state, created by the non-symplectic oscillator-qudit conditional displacement~\eqref{eq:qudit-cond-disp} acting on a GKP state and a qudit \(|+\rangle\) state.
While all three are stabilizer states --- unique eigenstates of a commuting group of Pauli/displacement operators --- the third one cannot be created from a separable state by a Gaussian-Clifford operation.
}
\end{table*}

\subsection{Measuring the displacement}

We show that the Dirac-delta orthogonal set of displaced states,
\begin{equation}\label{eq:displaced-lca}
e^{-i\xi_{1}\hat{p}}e^{i\xi_{2}\hat{x}}|\text{LCA}\rangle\quad\text{for}\quad 0\leq \sch{\xi_{1},\xi_{2}}< 2\sqrt{\pi},
\end{equation}
is a complete ``basis'' for the oscillator-qubit system.
It may be surprising that applying qudit Paulis is not necessary to span the entire oscillator-qudit state space.
This is because displacements by \(\sqrt{\pi}\) perform the same action on the state, as shown in Eq.~\eqref{eq:qubit-error} for the case of the momentum and Pauli-\(X\) shifts.
In this way, the qudit is ``absorbed'' into the \(\cv\) phase space of the oscillator.

The displaced states are complete because their displacements exactly cover the required phase-space volume.
To show this, we pretend the physical qubit of the LCA state is in its own GKP mode and infer the required volume by freezing out any displacements that take the state out of the encoded subspace.

After encoding the physical qubit, the LCA state becomes the two-mode GKP Bell state as shown in Fig.~\ref{fig0_circuits}(b),
\begin{subequations}
\begin{align}
|\text{LCA}\rangle&\to|\overline{0}\rangle_{\text{GKP}}|\overline{0}\rangle_{\text{GKP}}+|\overline{1}\rangle_{\text{GKP}}|\overline{1}\rangle_{\text{GKP}}\\&\propto\,U_{BS}|\varnothing\rangle|\varnothing\rangle~,
\end{align}
\end{subequations}
where the overline kets are logical codewords of the GKP code stabilized by \(\exp(-i\sqrt{4\pi}\hat{p})\) and \(\exp(i\sqrt{4\pi}\hat{x})\).
This state can be obtained from two qunaught GKP states via a beam-splitter \(U_{BS}\)~\cite{walshe2020continuous}, as shown in the second line (up to normalization).

The beamsplitter preserves phase-space volume, so the phase-space volume required to create a complete basis by displacing the states on either side of the above equation is the same.
This is quantified by the following,
\begin{equation}\label{eq:volume}
A_{\text{LCA}} A_{\text{GKP}} = A^2_{\varnothing}~.
\end{equation}
The volume on the right-hand side is the product of the areas formed by distinct displacements of the qunaught GKP state.
Its stabilizers are displacements by \(\sqrt{2\pi}\), so all smaller displacements span an area of \(A_{\varnothing} = 2\pi\) for each mode~\cite{bacry1975proof}.

The volume on the left-hand side of Eq.~\eqref{eq:volume} is a product of the area \(A_{\text{LCA}}\) of displaced LCA states that we wish to determine, times the area \(A_{\text{GKP}}\) of detectable displacements of the GKP code of the second mode.
The latter area quantifies all displacements that would be frozen out if the qubit is assumed to be rigidly fixed in its codespace, which is the same as assuming the qubit was not encoded in the GKP code in the first place.
The logical displacements of the GKP code are \(\exp(-i\sqrt{\pi}\hat{p})\) and \(\exp(i\sqrt{\pi}\hat{x})\), and all detectable displacements span an area of \(A_{\text{GKP}} = \pi\).

Solving the above yields \(A_{\text{LCA}}=4\pi \), matching the area spanned by the displacements in Eq.~\eqref{eq:displaced-lca} and  proving that the displaced LCA states form a complete set.

Measuring in the set of displaced LCA states can distinguish displacements in either direction and along either position or momentum as long as the absolute value of each is less than \(\sqrt{\pi}\).
On the other hand, a measurement in the set of displaced GKP states allows one to distinguish any displacements whose absolute value is less than \(\sqrt{\pi/2}\).
This means there is a \(\sqrt{2}\) amplification in the range of distinguishable displacements of LCA states relative to GKP states.
This amplitification is possible \textit{as long as} there is no physical-qubit error during the displacement.
Such an error is equivalent to an additional displacement of \(\sqrt{\pi}\) in one or more directions, as evident by Eq.~\eqref{eq:qubit-error}, and would yield an incorrect measurement.

\subsection{Non-symplectically generated entanglement}
\label{sec:symplectic-entanglement}

The LCA state~\eqref{eq:lca-qubit-state} couples each GKP comb with its own qubit basis state, meaning it is entangled across the oscillator-qubit partition.
This entanglement is maximal w.r.t. the qubit, so tracing out the oscillator yields identity on the qubit.
All LCA states are entangled in this sense, and the set of such states forms a completely entangled subspace~\cite{wallach2000unentangled,parthasarathy2004maximal,parthasarathy2005extremal,walgate2008generic} of the oscillator-qubit space.
Displaced versions of each state, as in Eq.~\eqref{eq:displaced-lca}, form a completely entangled basis for the space.

The LCA state is also a stabilizer state, namely, it is an eigenstate of a commuting subgroup of the group of oscillator displacements and qubit Paulis.
However, it cannot be obtained from a separable oscillator-qubit state under symplectic transformations, i.e., under any unitaries that preserve this group.
This is in contrast to qubit stabilizer states and qunaught GKP states, which are all interconvertible via Clifford and Gaussian operations  \sch{within themselves,} respectively.
In this sense, LCA states form a new class of resource states whose entanglement cannot be created or destroyed by Gaussian-Clifford operations.

The simple reason for this feature is that oscillator-qubit symplectic transformations are always tensor products of a Gaussian operation on the oscillator and a Clifford gate on the qubit.
To see this, we can express the qubit Pauli matrices as exponentials of qubit position and momentum variables, \(\hat{\sigma}_{\mathsf{x}} = (-1)^{\hat m}\) and \(\hat{\sigma}_{\mathsf{z}} = (-1)^{\hat{\jmath}}\), where \(\hat{\jmath} = (\id-\hat{\sigma}_{\mathsf{z}})/2\) and \(\hat m = (\id-\hat{\sigma}_{\mathsf{x}})/2\), respectively~\cite{albert2017general}.
Any entangling symplectic gate must shift the oscillator position or momentum by some multiple of \(\hat m\) or \(\hat \jmath\).
This multiple is real-valued,
while the qubit position is binary. 
Therefore, there can never be an entangling gate of such a form since it would not preserve the binary phase space of the qubit.

Consider, for example, the oscillator-qubit conditional displacement by a real \(b\),
\begin{subequations}\label{eq:qubit-cond-disp}
\begin{align}
 \cd_{b}&=\id\otimes|0\rangle\langle0|+e^{ib\hat{p}}\otimes|1\rangle\langle1|\\&=\exp\left(ib\hat{p}\otimes\hat{\jmath}\right)~.
\end{align}
\end{subequations}
This gate does not preserve the displacement-Pauli group, i.e., it is not Gaussian-Clifford, because it shifts the oscillator position, 
\(\hat{x} \to \hat{x} + b \hat{\jmath}\),    
by a non-binary shift of the qubit.
When conjugating the corresponding displacement, this translates to
\begin{equation}
e^{ia\hat{x}}\to e^{ia(\hat{x}+b\hat{\jmath})}=e^{ia\hat{x}}\hat{\sigma}_{\mathsf{x}}^{ab/\pi}
\end{equation}
for any real \(a\).
The power of \(\hat{\sigma}_{\mathsf{x}}\) needs to be an integer, but that is not the case for generic \(a,b\).
\vva{
In fact, appending Clifford-Gaussian transformations with conditional displacements yields a universal gate set~\cite{brenner2025trading}.
}

One can show there are no entangling symplectic operations more formally by noticing that such operations correspond to homomorphisms between the configuration spaces of the oscillator and qubit, \(\mathbb R\) and \(\mathbb Z_2\), respectively, and recalling that there are no such homomorphisms.
Roughly speaking, one can never map a finite group into the reals because there is no way to break the periodicity of the group's elements without violating its multiplication rules (see Appx.~\ref{app:single-mode-symplectic} and Refs.~\cite{prasad2008decomposition,bermejo2016normalizer}).

\subsection{Initialization}

The LCA state is obtained by applying an oscillator-qubit conditional displacement~\eqref{eq:qubit-cond-disp} by \(b = \sqrt{\pi}\) to a tensor product of the GKP code's zero codeword and the qubit \(|+\rangle\) state.
The comb of the codeword consists of teeth separated by \(\sqrt{4\pi}=2\sqrt{\pi}\), and the conditional displacement shifts this comb by an extra \(\sqrt{\pi}\) whenever the qubit is in the \(|1\rangle\) state.
This yields the two combs making up the LCA state~\eqref{eq:lca-qubit-state},
\begin{equation}
    \cd_{\sqrt{\pi}} |\overline{0}\rangle_{\text{GKP}}|+\rangle \propto|\text{LCA}\rangle~.
\end{equation}

Another way to see that non-Gaussian resources are required to make the LCA state is to observe that the state can be obtained by initializing two qubits in a Bell state and encoding one of them into the GKP code.
The qubit's Paulis are mapped to the anti-commuting logical operators of the code, \(\hat{\sigma}_{\mathsf{x}}\to\exp(-i\sqrt{\pi}\hat{p})\) and \(\hat{\sigma}_{\mathsf{z}}\to\exp(i\sqrt{\pi}\hat{x})\), and the Bell state's stabilizer group becomes the LCA stabilizer group generated by the operators in Eq.~\eqref{eq:lca-qubit}.

\subsection{Oscillator-qudit LCA states}\label{sec:lca-states-qudit}

The qubit case so far discussed extends naturally to oscillator-qudit LCA states of a mode and a qudit of dimension \(c \geq 2\) and reduces to the GKP qunaught state~\eqref{eq:gkp} at \(c=1\).
The stabilizer generators extend to
\begin{subequations}
    \begin{align}
   \st_{X}&=e^{-i\sqrt{\frac{2\pi}{c}}\hat{p}}\hat{X}^{\dagger}\\\st_{Z}&=e^{i\sqrt{\frac{2\pi}{c}}\hat{X}}\hat{Z}~,
\end{align}
\end{subequations}
where \(\hat{X},\hat{Z}\) are the usual modular-qudit Pauli matrices satisfying \(\hat{X}\hat{Z}\hat{X}^{\dagger}\hat{Z}^{\dagger}=\exp(\ensuremath{-i2\pi/c})\).
The commutation of the qudit factors of each stabilizer compensates the phase arising from the non-commuting displacements.

The two stabilizers admit a unique joint eigenstate with eigenvalue \(+1\),
\begin{subequations}\label{eq:qudit-lca-state}
    \begin{align}
|\text{LCA}\rangle&=\sum_{\ell\in\mathbb{Z}}\left|\ell\sqrt{{\textstyle \frac{2\pi}{c}}}\right\rangle _{\hat{x}}\left|-\ell\right\rangle \\&=\sum_{j\in\mathbb{Z}_{c}}\left(\sum_{s\in\mathbb{Z}}\left|\left(sc+j\right)\sqrt{{\textstyle \frac{2\pi}{c}}}\right\rangle _{\hat{x}}\right)\left|-j\right\rangle~,
    \end{align}
\end{subequations}
where unmarked kets, \(\left|-\ell\right\rangle\) and \(\left|-j\right\rangle\), are the \(c\) distinct qudit \(\hat{Z}\)-eigenstates, defined for argument modulo \(c\).

The second form of the state shows that each qudit eigenstate, \(\left|-j\right\rangle\), is paired with its own GKP state, and that the \(c\) different GKP states are spaced apart by \(\sqrt{2\pi/c}\).
This means that a displacement of \(c\sqrt{2\pi/c} = \sqrt{2\pi c}\) is required to map the state back to itself.
This is a factor of \(\sqrt{c}\) longer than the stabilizing displacement of the qunaught GKP state, obtained from the above at \(c = 1\).

The momentum-space form of the state can be obtained by inserting a resolution of the identity in the modal momentum and qudit-\(\hat{X}\) bases, yielding a similar comb form,
\begin{equation}\label{eq:dual-comb}
|\text{LCA}\rangle=\sqrt{2\pi c}\sum_{j\in\mathbb{Z}_{c}}\left(\sum_{s\in\mathbb{Z}}\left|\left(sc+j\right)\sqrt{{\textstyle \frac{2\pi}{c}}}\right\rangle _{\hat{p}}\right)|j\rangle_{\hat{X}}~,
\end{equation}
where \(|j\rangle_{\hat{X}}\) is the \(j\)th \(\hat{X}\)-eigenstate.
The same argument holds for momentum shifts, and the shortest pure-momentum displacement stabilizer is a momentum shift of \(\sqrt{2\pi c}\).

The square unit cell enclosed by \(\sqrt{2\pi c}\)-displacements in each direction has area \(A_{\text{LCA}} = 2\pi c\), notably proportional to the qudit dimension.
The set of LCA states displaced by displacements whose values lie in the unit cell is a complete set because of a qudit version of the volume argument we made for the oscillator-qubit case.

To state the argument, we embed the qudit into its own qudit GKP code, after which the oscillator qudit LCA state becomes a GKP-qudit Bell state,
\begin{equation}
|\text{LCA}\rangle\to\sum_{j\in\mathbb{Z}_{c}}|\overline{j}\rangle_{\text{GKP}}|\overline{-j}\rangle_{\text{GKP}}\propto\,e^{i\hat{p}_{1}\hat{x}_{2}}|\overline{0}\rangle_{\text{GKP}}|\overline{+}\rangle_{\text{GKP}}~.
\end{equation}
The proportionality shows that, up to normalization, this Bell state can be created from a separable logical state using the above conditional oscillator-oscillator displacement.
This conditional displacement \textit{is} a Gaussian-Clifford operation, unlike its oscillator-qudit cousin, meaning that it preserves phase-space volume.

Both the qudit GKP-zero and qudit GKP-plus states are Gaussian-deformed versions of the qunaught GKP state, so the joint volume of their unit cell is the same as the product of two qunaught states.
The volume of detectable displacements of the GKP qudit is \(2\pi/c\), yielding the correct LCA-state volume,
\begin{equation}
A_{\text{LCA}}=\frac{A_{\varnothing}^{2}}{A_{\text{GKP}}}=\frac{(2\pi)^{2}}{2\pi/c}=2\pi c\,.
\end{equation}

Using similar arguments as for the oscillator-qubit case, we conclude that oscillator-qudit symplectic transformations are always tensor products of Gaussian operations and modular-qudit Clifford gates (see Appx.~\ref{app:single-mode-symplectic}).
As such, oscillator-qudit LCA states cannot be obtained by acting with such transformations on a separable state.
The two ways to initialize an LCA state --- encoding one half of a qudit Bell state into a GKP qudit, or applying a conditional oscillator-qudit displacement --- also carry over to the qudit case.
The oscillator-qudit conditional displacement required for the latter is
\begin{equation}\label{eq:qudit-cond-disp}
    \cd_{\sqrt{2\pi/c}}=e^{i\sqrt{\frac{2\pi}{c}}\hat{p}\otimes\hat{\jmath}}=\sum_{j\in\Z_{c}}e^{ij\sqrt{\frac{2\pi}{c}}\hat{p}}\otimes|j\rangle\langle j|~.
\end{equation}

\subsection{Hybrid Zak bases}

A type of dual basis that exists for GKP codes --- the discrete Zak basis~\cite{gel1950expansion,zak1967finite} --- can also be extended to the LCA case~\cite{kaniuth1998zeros,enstad2020balian}.

The set of distinct displaced LCA states,
\(|\xi_{1},\xi_{2}\rangle=e^{-i\xi_{1}\hat{p}}e^{i\xi_{2}\hat{x}}|\text{LCA}\rangle\),
is sometimes called the continuous or periodic Zak basis~\cite{englert2006periodic,ketterer2016quantum,fabre2020wigner,pantaleoni2023zak}.
It is parameterized by two angles, \(0 \leq \xi_{1,2} < \sqrt{2\pi c}\), and the fact that it is complete implies that one can parameterize functions on the real line with those on a two-dimensional torus.

Taking the two-dimensional Fourier transform on this torus yields an integer-valued hybrid analogue of the  Zak basis,
\begin{equation}
  |n_{1},n_{2}\rangle=\int_{0}^{\sqrt{2\pi c}}d\xi_{1}d\xi_{2}e^{i(n_{1}\xi_{2}-n_{2}\xi_{2})}|\xi_{1},\xi_{2}\rangle\,,
\end{equation}
for integers \(n_{1,2}\).
Since the unit cell area of the displacements behind displaced LCA states is \(2\pi c\), the integers parameterizing the discrete Zak basis can be thought to form a grid of spacing \(\sqrt{2\pi c}\) --- a factor of \(\sqrt{c}\) larger than the grid of the original discrete Zak basis.

\begin{figure}[t]
    \centering
    \includegraphics[width=0.95\columnwidth]{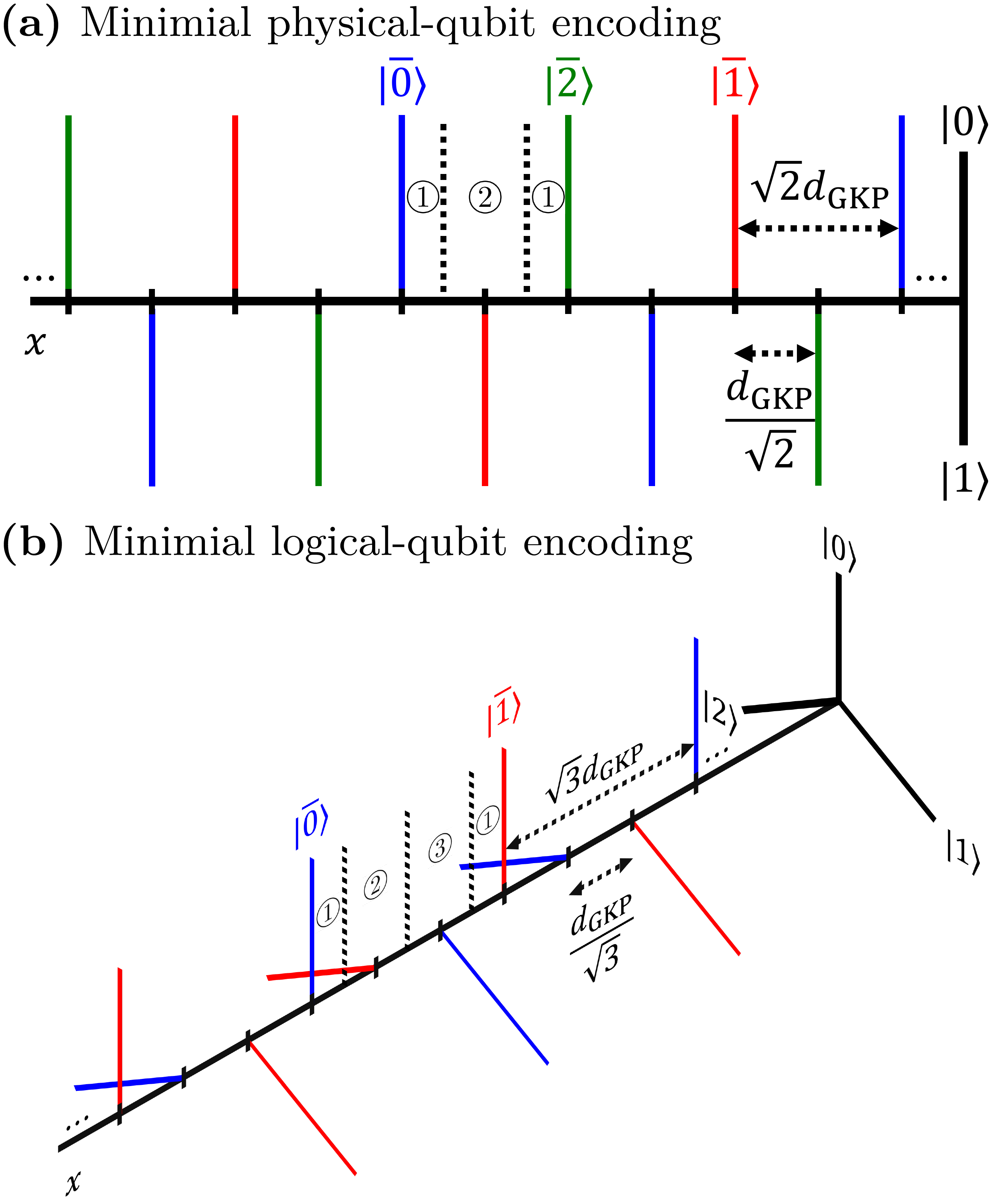}
    \caption{\small
    \label{fig2_encodings}
    Depictions of logical codewords of LCA codes encoding \textbf{(a)} a logical qutrit (\(K=3\)) into a mode and a physical qubit (\(c=2\)) [see Eq.~\eqref{eq:qutrit}], and \textbf{(b)} a logical qubit (\(K=2\)) into a mode and a physical qutrit (\(c=3\)) [see Eq.~\eqref{eq:qubit}]. 
    These are the smallest encodings of a logical qutrit and a logical qubit, respectively.
    Both utilize the same six oscillator combs, paired up in two different ways to the physical qubit and qutrit basis states.
    The shortest undetectable displacement is \(\sqrt{c}d_{\text{GKP}}\), 
    where \(d_{\text{GKP}} = \sqrt{2\pi/K}\)~\eqref{eq:code-distance}.
    A decoder correcting against all single-qudit Pauli-\(\hat{X}\) errors along with small position displacements can be constructed by splitting up the syndrome space into segments \(\Circled{1},\Circled{2},\cdots,\Circled{c}\), as shown (see Sec.~\ref{sec:qec}).
    The same picture holds against Pauli-\(\hat{Z}\) errors and momentum displacements in the Fourier domain. 
    }
\end{figure}

\section{Simple LCA codes}
\label{sec:single-mode}

We define LCA codes for a single oscillator and a single mode.
These span the joint \(+1\)-eigenvalue eigenspace of a commuting set of displacement-Pauli operators.

\subsection{Stabilizers}

LCA codes are defined by a non-negative integer \(c\) and integer \(d\) that satisfy \begin{equation}\label{eq:gcd}
   \gcd(c,d)=1\quad\leftrightarrow\quad d\in\Z_c^{\times}\quad\quad\text{(GCD condition).}
\end{equation}
In other words, \(d\) is an element of the multiplicative group \(\Z_c^{\times}\) of integers modulo \(c\) --- the subset of \(\Z_c\) for which multiplication modulo \(c\) is well-defined.
The value of \(c\) is the qudit dimension, and the possible logical dimensions \(K\) (not all are possible for a given \(c\)!) are determined by \(d\).

The stabilizer group is generated by the operators
\begin{subequations}\label{eq:stabilizers}
    \begin{align}
\st_{X}&=e^{-i\sqrt{\frac{2\pi K}{c}}\hat{p}}\hat{X}^{\dagger d}\\\st_{Z}&=e^{i\sqrt{\frac{2\pi K}{c}}\hat{X}}\hat{Z}~.
    \end{align}
\end{subequations}
Since \(\hat{X}^c=\hat{Z}^c=\id\), stabilizers with the same \(K\) and different values \(d,d^{\prime}\) are equivalent whenever \(d \equiv d^{\prime}\) modulo \(c\).
There are thus at most \(c-1\) different ``classes'' of qudit factors, corresponding to the \(c-1\) distinct nontrivial powers of \(\hat{X}^{\dagger}\).
This is saturated for prime \(c\), in which case there are exactly \(c-1\) code classes.

For a given \(d \mod c\), the possible values of \(K\) differ from each other by multiples of \(c\) and can be modulated by a free integer parameter \(\theta\),
\begin{equation}
    K = c\theta+d \quad\quad\text{(logical dimension).}
\end{equation}
All integer values of \(d,\theta\) are allowed in principle, and an absolute value should be taken on the right-hand side to accommodate this (see Appx.~\ref{app:general-single-mode-integer-case} for such a treatment).
Instead, we fix \(\theta\geq 0\) and \(0 \leq d \leq c-1\), which guarantees a non-negative \(K\).
This is done without loss of generality since a code with any other \(d\) is equivalent to one with a \(d\) in our defined range via a compensating change in \(\theta\), and since a code with \(c\theta+d<0\) is equivalent to a code with the signs of \(\theta,d \) flipped.

The commutation relation between the two stabilizers is satisfied due to the qudit and oscillator phases cancelling each other, and due to \(\theta\) being an integer,
\begin{equation}\label{eq:lattice-comm}
\st_{X}\st_{Z}\st_{X}^{\dagger}\st_{Z}^{\dagger}=e^{-i\frac{2\pi}{c}(c\theta+d-d)}=e^{-i2\pi\theta}=1~.
\end{equation}

The \(K=d=1\) case for general \(c\) reduces to the simple LCA states from Sec.~\ref{sec:lca-states-qudit}.
The other classes have \(d > 1\) and thus host logical codespaces.
We call these cases \((c,d)\)-\textit{LCA codes} and refer to the \(\theta = 0\) special cases as \textit{simple \((c,d)\)-LCA codes}  --- a family that includes the qunaught GKP state (\(c=d=1\)), simple LCA states (\(c>d=1\)), as well as all LCA codes with \(K=d > 1\).
Any single-mode single-qudit LCA code can be obtained from a simple LCA code by tuning \(d\) and \(\theta\).

\subsection{Codewords}
\label{subsec:exact-codewords}

\vva{
Codewords of simple LCA codes are constructed by starting with a \(Kc\)-dimensional GKP code, splitting it up via a logical subsystem decomposition into \(K\)- and \(c\)-dimensional factors, and entangling the latter factor with the \(c\)-dimensional qudit.
}

Given \(\mu\in\mathbb{Z}_K\), a basis of codewords is defined by
\begin{subequations}
    \begin{align}
\!\!|\overline{\mu}\rangle&=\sum_{\ell\in\mathbb{Z}}\left|\left(\ell K+\mu\right){\textstyle \sqrt{\frac{2\pi}{Kc}}}\right\rangle _{\hat{x}}\left|-(\ell d+\mu)\right\rangle \\&=\sum_{j\in\mathbb{Z}_{c}}\!\!\left(\sum_{s\in\mathbb{Z}}\left|\left(jK+\mu\right){\textstyle \sqrt{{\textstyle \frac{2\pi}{Kc}}}}+\sqrt{2\pi Kc}s\right\rangle _{\hat{x}}\right)\!\!\left|-(jd+\mu)\right\rangle\nonumber\\&=\sum_{j\in\Z_{c}}\left|\text{GKP}_{Kj+\mu}^{Kc}\right\rangle \left|-(jd+\mu)\right\rangle~,\label{eq:gkp-comp}
\end{align}
\end{subequations}
where \(|~\rangle_{\hat x}\) is a position vector, and where the unmarked ket is a canonical qudit \(\hat{Z}\) eigenstate whose argument is evaluated modulo \(c\). 
The second equation splits the hybrid comb into a superposition of \(c\) different combs, each multiplied by a unique qudit state. 
The qudit state is unique because \(-(jd+\mu)\) modulo \(c\) evaluates to a distinct value for each \(j\) due to the GCD condition~\eqref{eq:gcd}.

\vva{
Equation~\eqref{eq:gkp-comp} expresses the LCA state in terms of its underlying GKP code, where we define 
\begin{equation}\label{eq:gkp-codewords}
    \left|\text{GKP}_{\mu}^{K}\right\rangle =\sum_{s\in\Z}\left|\left(Ks+\mu\right){\textstyle \sqrt{\frac{2\pi}{K}}}\right\rangle_{\hat x} 
\end{equation}
to be the \(\mu\)th logical state of a GKP qudit of dimension \(K\).
Each LCA codeword contains \(c\) distinct GKP codewords, and a total of \(Kc\) codewords participate in the LCA code.
The GKP code's logical states are split up according to the logical\footnote{This decomposition is different from a previously studied decomposition, see e.g.~\cite{shaw2024stabilizer}, which occurs on the level of the error syndromes.} subsystem decomposition
\begin{equation}\label{eq:subsystem}
    \Z_{Kc} = \Z_K \times \Z_c\quad\quad\text{(logical subsystem decomp.)}.
\end{equation}
The \(c\)-dimensional factor entangles with the \(c\) states of the physical qudit, while the \(K\)-dimensional factor stores the logical information.
This decomposition occurs by the Chinese remainder theorem because \(K\) and \(c\) are defined to be coprime. 
}

We cover a few notable specific cases, two of which are depicted in Fig.~\ref{fig2_encodings}.
Other possible combinations of \(c\) and \(d\) are those whose corresponding entry in Table~\ref{tab:bezout} is filled in.

\prg{GKP (\(c=1\))}
This case yields the square-lattice GKP code with logical dimension \(K = \theta+1\) and a separation of \(\sqrt{2\pi/K}\) between the \(K\) codeword combs.
The simple (\(\theta=0\)) case reduces to the qunaught GKP state~\eqref{eq:gkp}.

\begin{widetext}

\prg{Physical qubit (\(c=2\))}
The GCD condition for this case requires \(d\) to be odd, and the fact that \(\hat{\sigma}_{\mathsf{x}}^d = \hat{\sigma}_{\mathsf{x}}^{d\text{ mod }2}\) means that qubit factors of the stabilizers~\eqref{eq:stabilizers} for any admissible \(d\) are the same as those at \(d=1\).
The simple element of this class is the state at \(\theta=0\), and other values of \(\theta\) are LCA codes with odd logical dimension \(K = 2\theta +1\).
Notably, a logical qubit is excluded from the possible encodings.
This holds more generally: one cannot encode a \(c\)-dimensional logical space in a mode and a physical qudit of the same dimension.

The case of \(\theta = 0\) reduces to the LCA state in Eq.~\eqref{eq:lca-qubit-state}.
The smallest possible encoding is a logical qutrit (\(\theta=1\)), depicted in Fig.~\ref{fig2_encodings}(a). 
Its explicit codewords are
\begin{subequations}\label{eq:qutrit}
    \begin{align}
        |\overline{0}\rangle&=\sum_{s\in\sqrt{12\pi}\mathbb{Z}}\left|{\color{blue}0{\textstyle \sqrt{\frac{\pi}{3}}}+s}\right\rangle _{\hat{x}}\left|0\right\rangle +\left|{\color{blue}3{\textstyle \sqrt{\frac{\pi}{3}}}+s}\right\rangle _{\hat{x}}\left|1\right\rangle \\|\overline{1}\rangle&=\sum_{s\in\sqrt{12\pi}\mathbb{Z}}\left|{\color{red}1{\textstyle \sqrt{\frac{\pi}{3}}}+s}\right\rangle _{\hat{x}}\left|1\right\rangle +\left|{\color{red}4{\textstyle \sqrt{\frac{\pi}{3}}}+s}\right\rangle _{\hat{x}}\left|0\right\rangle \\|\overline{2}\rangle&=\sum_{s\in\sqrt{12\pi}\mathbb{Z}}\left|{\color{greenvva}2{\textstyle \sqrt{\frac{\pi}{3}}}+s}\right\rangle _{\hat{x}}\left|0\right\rangle +\left|{\color{greenvva}5{\textstyle \sqrt{\frac{\pi}{3}}}+s}\right\rangle _{\hat{x}}\left|1\right\rangle  ~.
    \end{align}
\end{subequations}
\vva{
The corresponding GKP logical subsystem decomposition~\eqref{eq:subsystem} is \(\Z_6 = \Z_{3}\times \Z_{2}\), where the first factor is the logical factor, and the second is entangled with the physical qubit.
}

\prg{Physical qutrit (\(c=3\))}
The GCD condition for this case requires \(d\) to not be a multiple of 3, yielding two ``chiral'' classes of codes with \(K = 3\theta +1\) and \(K = 3\theta + 2\), respectively.
Since a physical qutrit is required to encode a logical qubit, this case includes the smallest possible encoding of a logical qubit --- the simple \((3,2)\)-LCA code.
The codewords for that encoding, depicted in Fig.~\ref{fig2_encodings}(b), are
\begin{subequations}\label{eq:qubit}
    \begin{align}
     |\overline{0}\rangle&=\sum_{s\in\sqrt{12\pi}\mathbb{Z}}\left|{\color{blue}0{\textstyle \sqrt{\frac{\pi}{3}}}+s}\right\rangle _{\hat{x}}\left|0\right\rangle +\left|{\color{blue}2{\textstyle \sqrt{\frac{\pi}{3}}}+s}\right\rangle _{\hat{x}}\left|1\right\rangle +\left|{\color{blue}4{\textstyle \sqrt{\frac{\pi}{3}}}+s}\right\rangle _{\hat{x}}\left|2\right\rangle \\|\overline{1}\rangle&=\sum_{s\in\sqrt{12\pi}\mathbb{Z}}\left|{\color{red}1{\textstyle \sqrt{\frac{\pi}{3}}}+s}\right\rangle _{\hat{x}}\left|2\right\rangle +\left|{\color{red}3{\textstyle \sqrt{\frac{\pi}{3}}}+s}\right\rangle _{\hat{x}}\left|0\right\rangle +\left|{\color{red}5{\textstyle \sqrt{\frac{\pi}{3}}}+s}\right\rangle _{\hat{x}}\left|1\right\rangle  ~.
    \end{align}
\end{subequations}
The same six comb states are utilized in this code as in the previous \(c=2\) logical-qutrit encoding. 
\vva{
The GKP logical subsystem decomposition~\eqref{eq:subsystem} is also the same, \(\Z_6 = \Z_{2}\times \Z_{3}\), but now the \(\Z_2\) factor is used for the logical information.
}
    
\end{widetext}

\prg{Physical qusix (\(c=6\))}

The GCD condition is more restrictive for composite values of \(c\), with only two possible values of \(d\), one and five, allowed in this case.
\vva{
The former yields a simple LCA state while the latter yields a five-dimensional logical space with subsystem decomposition \(\Z_{5}\times \Z_6 = \Z_{30}\).
}
The simple codes yield derivative codes with dimensions \(K = 6\theta + 1\) and \(K = 6\theta + 5\), respectively.

~~~~~~~~~~~~~~~~~~

~~~~~~~~~~~~~~~~~~~~~~~
\subsection{Logicals \& dual basis}

Since \(d\) is an element of the multiplicative group \(\Z_c^{\times}\), it has a multiplicative inverse modulo \(c\).
This inverse, defined here as the integer \(-a\), satisfies Bézout's identity,
\begin{equation}\label{eq:bezout}
    bc-ad=1\quad\quad\text{(Bezout)},
\end{equation}
for some integer \(b\).
The parameter \(a\) determines a pair of logical operators, while \(b\) is present in their commutation.

Logical operators are
\begin{subequations}\label{eq:logicals}
    \begin{align}
       \xl&=e^{-i\sqrt{\frac{2\pi}{Kc}}\hat{p}}\hat{X}^{\dagger}\\\zl&=e^{-i\sqrt{\frac{2\pi}{Kc}}\hat{x}}\hat{Z}^{a}~.
    \end{align}
\end{subequations}
These commute to
\begin{subequations}\label{eq:logical-comm}
\begin{align}
\xl\zl\xl^{\dagger}\zl^{\dagger}&=e^{i\frac{2\pi}{Kc}}e^{i\frac{2\pi}{c}a}&\\&=e^{i\frac{2\pi}{c\theta+d}\left(a\theta+\frac{ad+1}{c}\right)}&(K=c\theta+d)\\&=e^{i\frac{2\pi}{K}(a\theta+b)}&(ad+1=bc),
\end{align}
\end{subequations}
which is a primitive \(K\)th root of unity because \(a\theta+b\) can be shown to be coprime to \(K\).
Therefore, the above operators encode a logical qudit of dimension \(K\).

Dual-basis codewords can be obtained by applying powers \(\mu\) of the logical \(\hat{Z}\) to the equal superposition of all \(\hat{Z}\)-type codewords, denoted by \(|\overline + \rangle\) and proportional to the state in Eq.~\eqref{eq:dual-comb} with \(\sqrt{\frac{2\pi}{c}}\to \sqrt{\frac{2\pi}{Kc}}\),
\begin{equation}
    \overline{Z}^{\mu}|\overline{+}\rangle=\sum_{\ell\in\mathbb{Z}}\left|\left(\ell K-\mu\right){\textstyle \sqrt{\frac{2\pi}{Kc}}}\right\rangle _{\hat{p}}|\ell+a\mu\rangle_{\hat{X}}\,,
\end{equation}
where \(|~\rangle_{\hat{X}}\) is a qudit Pauli-\(\hat{X}\) eigenstate.

\begin{table}
\centering %

\begin{tabular}{c|c|c|c|c|c|c|c|c|c}
{$c\backslash d$} & {$1$} & {$2$} & {$3$} & {$4$} & {$5$} & {$6$} & {$7$} & {$8$} & {$9$}\tabularnewline
\hline 
{$\overset{\phantom{.}}{2}$} & {$\overline{1},0$} & {} & {} & {} & {} & {} & {} & {} & {}\tabularnewline
\hline 
{$\overset{\phantom{.}}{3}$} & {$\overline{1},0$} & $1,1$ & {} & {} & {} & {} & {} & {} & {}\tabularnewline
\hline 
{$\overset{\phantom{.}}{4}$} & {$\overline{1},0$} & {} & $1,1$ & {} & {} & {} & {} & {} & {}\tabularnewline
\hline 
{$\overset{\phantom{.}}{5}$} & {$\overline{1},0$} & $2,1$ & $\overline{2},\overline{1}$ & $1,1$ & {} & {} & {} & {} & {}\tabularnewline
\hline 
{$\overset{\phantom{.}}{6}$} & {$\overline{1},0$} & {} & {} & {} & $1,1$ & {} & {} & {} & {}\tabularnewline
\hline 
{$\overset{\phantom{.}}{7}$} & {$\overline{1},0$} & $3,1$ & $2,1$ & $\overline{2},\overline{1}$ & $\overline{3},\overline{2}$ & $1,1$ & {} & {} & {}\tabularnewline
\hline 
{$\overset{\phantom{.}}{8}$} & {$\overline{1},0$} & {} & $\overline{3},\overline{1}$ & {} & $3,2$ & {} & $1,1$ & {} & {}\tabularnewline
\hline 
{$\overset{\phantom{.}}{9}$} & {$\overline{1},0$} & $4,1$ & {} & $2,1$ & $\overline{2},\overline{1}$ & {} & $\overline{4},\overline{3}$ & $1,1$ & {}\tabularnewline
\hline 
{$\overset{\phantom{.}}{10}$} & {$\overline{1},0$} & {} & $3,1$ & {} & {} & {} & $\overline{3},\overline{2}$ & {} & $1,1$\tabularnewline
\end{tabular}

\caption{
\small
Integer pairs ``\(a,b\)'' satisfying Bézout's identity~\eqref{eq:bezout}, \(bc-ad=1\), given integers \(c\) and \(d<c\). 
An overline means the integer is negative, e.g., \(\overline 2 \to -2\). 
Each pair gives rise to the logical operators of a simple \((c,d)\)-LCA encoding with logical space of dimension \(K<c\) into a mode and a \(c\)-dimensional qudit.
Empty entries remind us that the LCA construction requires that \(\gcd(c,d)=1\).
The logical dimension can then be extended by tuning an additional free parameter \(\theta\) to yield an LCA code with dimension \(K=c\theta+d\).
\label{tab:bezout}
}
\end{table}

\section{Error detection \& correction}\label{sec:qec}

We present two decoding strategies for arbitrary \(c,d,K\) that use the syndromes of the code in ways.
The first protects against pure displacements, with the largest uncorrectable displacement being  a factor of \(\sqrt{c}\) larger than that of the GKP code.
The second protects against arbitrary qudit errors at the expense of a smaller set of correctable displacements relative to those of the GKP code.
\vva{
More balanced decoders are possible, and we highlight one for the case \(c=2\).
}

\subsection{Pure-displacement decoder}

The set of displaced states,
\begin{equation}\label{eq:displaced-states}
    |\xi_1,\xi_2,\mu\rangle=e^{-i\xi_{1}\hat{p}}e^{i\xi_{2}\hat{x}}|\overline{\mu}\rangle~,
\end{equation}
for logical index \(\mu\in \mathbb{Z}_K\) and displacements satisfying
\begin{equation}
    0\leq\xi_{1},\xi_{2}<\sqrt{2\pi c/K}~,
\end{equation}
forms a complete set of error states for the oscillator-qudit system.
The corresponding phase-space volume of the unit cell of the hybrid lattice is \(2\pi c/K\), and any displacement that keeps the code in this unit cell is detectable.

The shortest codespace-preserving displacement in each direction is the \(c\)th power of a logical operator.
For example,
\begin{equation}
    \xl^{c}=e^{-ic\sqrt{\frac{2\pi}{Kc}}\hat{p}}\hat{X}^{\dagger c}=e^{-id_{\text{LCA}}\hat{p}}~.
\end{equation}
This is a factor of \(\sqrt{c}\) larger than the same displacement for a \(K\)-dimensional square-lattice GKP code.
By analogy with \(\sqrt{2\pi/K}=d_{\text{GKP}}\) being the distance of the GKP code,
we can define
\begin{equation}\label{eq:code-distance}
d_{\text{LCA}}=\sqrt{\frac{2\pi c}{K}}=\sqrt{c}~d_{\text{GKP}}\quad\quad\text{(code distance)}
\end{equation}
to be the pure-displacement distance of the LCA code.

Displacements in either direction, \(\epsilon_1\) in position and \(\epsilon_2\) in momentum, can be corrected as long as the absolute value of either displacement is less than \(d_{\text{LCA}}/2\),
\begin{equation}
   0\leq|\epsilon_{1}|,|\epsilon_{2}|<\frac{d_{\text{LCA}}}{2}\quad\quad\text{(correctable pure disp's)}.
\end{equation}
\vva{
The volume of correctable displacements is therefore \(2\pi c/K\), the volume of the lattice's square unit cell centered at the origin.
}
This is larger than that of a GKP code by the factor of \(c\).

, but correction need not hold when qudit errors are considered.

\subsection{Qudit-error decoder}

An alternative decoding strategy sacrifices some protection against displacements but protects against \textit{arbitrary} qudit errors.

Each error space, labeled by a syndrome pair \((\xi_1,\xi_2)\), can be related to the codespace via either a pure displacement or a combination of qudit error and smaller displacement.
Here, we assume the latter, yielding a way to protect against all qudit errors and the following set of shifts,
\begin{equation}
    0\leq|\epsilon_{1}|,|\epsilon_{2}|<\frac{d_{\text{LCA}}}{2c}=\delta\quad\text{(disp's + qudit errors)}.
\end{equation}
This is a smaller range of correctable displacements than that of a GKP code.

We describe this decoding procedure first for position shifts and Pauli-\(\hat{X}\) errors.
The syndrome range \([0, d_{\text{LCA}}) \ni \xi_1\) is partitioned into \(c\) regions,
\begin{subequations}\label{eq:split}
\begin{align}
\label{eq:split1}
\Circled{1}&=\left[0,\delta\right)\cup\left[(2c-1)\delta,2c\delta\right)\\\Circled{2}&=\left[1\delta,3\delta\right)\\\Circled{3}&=\left[3\delta,5\delta\right)\\&\vdots\nonumber\\\Circled{r}&=\left[\left(2r-3\right)\delta,\left(2r-1\right)\delta\right)\\&\vdots\nonumber
\end{align}
\end{subequations}
for \(2\leq r\leq c\).
This partitioning and the region labels are shown for two examples in Fig.~\ref{fig2_encodings}.
We associate one of \(c\) qudit errors with each region, while a sufficiently small displacement modulates the value of \(\xi_1\) within that region.
The cases at the boundaries between two regions can be assigned to either, and we assign each boundary syndrome to the region whose label is the smaller of the two.

Syndromes in region \(\Circled{1}\) handle pure shifts.
As shown in Eq.~\eqref{eq:split1}, this range consists of two segments.
Syndromes in the first segment, \(\epsilon_1\in \left[0,\delta\right)\), correspond to shifts by \(\epsilon_1\), while shifts by the same amount in the other direction correspond to syndromes \(d_{\text{LCA}} - \epsilon_1\in\left[(2c-1)\delta,2c\delta\right)\).

We assign syndromes in segment \(\Circled{2}\) to handle errors of the form \(\hat{X} \exp{-i\epsilon_1 \hat{p}}\) for \(|\epsilon_1|< \delta\).
Such shifts correspond to syndromes \(\xi_1 = 2\delta + \epsilon_{1}\).
The extra factor of \(2\delta\) accounts for the 
qudit error \(\hat{X}\),
which brings the codespace to the same error space as a shift by that amount,
\begin{equation}
\hat{X}=e^{-i2\delta\hat{p}}\,\xl^{\dagger}~,
\end{equation}
up to a logical Pauli.
The presence of the logical Pauli means that the qudit error and the displacement by \(2\delta\) cannot both be corrected.

Generalizing this, we can assign a qudit error \(\hat{X}^{r-1}\) for \(2 \leq r \leq c\) to each segment \(\Circled{r}\).
The different segments handle the different possible qudit errors, and syndromes within each segment are used to correct displacements satisfying \(|\epsilon_1| < \delta\) according to the assignment
\begin{equation}
    \xi_{1}=(2r-2)\delta+\epsilon_{1}~.
\end{equation}

The same partitioning can be done with momentum shifts.
We split up the momentum syndromes \([0, d_{\text{LCA}}) \ni \xi_2\) into \(c\) segments as in Eq.~\eqref{eq:split}.
The first segment is used to protect against pure momentum shifts.
Syndromes in each segment \(\Circled{r}\) of the remaining \(c-1\) segments are used to protect against qudit Pauli-\(\hat{Z}\) shifts \(\hat{Z}^{a(r-1)}\).
The extra power of \(a\) is due to the states  ``twisting'' differently in momentum space --- the qudit error \(\hat{Z}^a\) is related to the shortest momentum displacement,
\begin{equation}
\hat{Z}^{a}=e^{i2\delta\hat{x}}\,\zl\,,
\end{equation}
as derived from Eq.~\eqref{eq:logicals}.

\subsection{Balanced decoder}

\vva{
We can tune the displacement error capability of the above all-qubit-error decoder by moving the boundaries between regions \(\Circled{1}\) and \(\Circled{2}\) for LCA codes constructed with a physical qubit (\(c=2\)).
Those boundaries are depicted by the dotted lines in Fig.~\ref{fig2_encodings}(a). 
One can similarly move boundaries between various regions for codes with \(c>2\), but this results in decoders that are less symmetrical than the qubit example we mention.

Moving the two lines closer grows region \(\Circled{1}\) while shrinking region \(\Circled{2}\), which reserves a larger set of syndromes for pure displacements at the expense of a narrower set of correctable displacements occurring with a qubit error.
This regime should be useful when the oscillator is noisier than the qubit.

On the other hand, expanding region \(\Circled{2}\) allows the code to correct a larger set displacements that occur with a qubit error, at the expense of a smaller correctable set of pure displacements.
This tuning is appropriate if the physical qubit is noisier than the oscillator.
}

\section{Normalizable LCA codewords}
\label{sec:approximate}

\vva{
Ideal LCA states, like GKP states, are not normalizable.
The expression~\eqref{eq:gkp-comp} of the codewords of a \(K\)-dimensional LCA code in terms of codewords of a corresponding \(Kc\)-dimensional GKP code allows for simple approximate versions of these states that are normalizable.

An approximate GKP codeword \(\mu\) of a code with dimension \(K\) is~\cite[Eq.~(24)]{matsuura2020equivalence}
\begin{equation}\label{eq:gkp-approximate}
\left|\text{GKP}_{\mu}^{K}(\Delta)\right\rangle \propto\sum_{s\in\Z}e^{-\frac{1}{2}\Delta^{2}f(s)^{2}}e^{-if(s)\hat{p}}\hat{S}_{-\ln\Delta}|\text{vac}\rangle~,
\end{equation}
where \(|\text{vac}\rangle\) is the zero Fock state of the oscillator, where \(\hat S\) is the operator that squeezes the vacuum along the vertical (momentum axis) in phase space, and where   
\begin{equation}
f(s)=\textstyle{\sqrt{\frac{2\pi}{K}}}(Ks+\mu)~.    
\end{equation}
Taking \(\Delta \to 0\) yields ideal GKP states from Eq.~\eqref{eq:gkp-codewords}.
In that limit, the overlap between two normalized codewords \(\mu\neq \nu\) is asymptotically~\cite[Eq.~(77)]{matsuura2020equivalence}
\begin{subequations}
\label{eq:gkp_overlap}
\begin{align}
\left\langle \text{GKP}_{\mu}^{K}(\Delta)|\text{GKP}_{\nu}^{K}(\Delta)\right\rangle & \sim\exp\left[-\frac{\pi}{2\Delta^{2}}\frac{(\mu-\nu)^{2}}{K}\right]\\
&\leq\exp\left(-\frac{\pi}{2\Delta^{2}}\frac{1}{K}\right)~.
\end{align}
\end{subequations}
In words, the overlap is suppressed exponentially with \(K^{-1}\), the inverse of the GKP code's dimension.
For any nonzero value of \(\Delta\), this overlap leads to an intrinsic logical error in the code.

Normalizable versions of LCA states are then constructed by substituting the above normalizable GKP states into Eq.~\eqref{eq:gkp-comp} to yield
\begin{equation}\label{eq:normalizable_lca_states}
    |\overline{\mu_{\Delta}}\rangle={\textstyle \frac{1}{\sqrt{c}}}\sum_{j\in\Z_{c}}\left|\text{GKP}_{Kj+\mu}^{Kc}(\Delta)\right\rangle \left|-(jd+\mu)\right\rangle ~.
\end{equation}

\subsection{Codeword overlap}

Each LCA codeword \(\mu\) is constructed from \(c\) normalizable GKP codewords, and the entire LCA code is constructed from a GKP code of dimension \(Kc\) (see Sec.~\ref{subsec:exact-codewords}).
Per Eq.~\eqref{eq:gkp_overlap}, we would expect that the asymptotic overlap between such LCA states is suppressed exponentially in \((Kc)^{-1}\) as \(\Delta \to 0\).
However, the oscillator-qudit entanglement of the \(\Z_c\) factor of the GKP code's logical subsystem decomposition~\eqref{eq:subsystem} leads to a suppression in \(c^2(Kc)^{-1} = c/K\).

We state the answer first before proving it.
In the \(\Delta \to 0\) limit, the overlap between LCA codewords \(\mu\neq\nu\) is asymptotically upper-bounded as [cf. Eq.~\eqref{eq:gkp_overlap}]
\begin{equation}\label{eq:lca-overlap}
    \langle\overline{\mu_{\Delta}}|\overline{\nu_{\Delta}}\rangle\lesssim\exp\left(-\frac{\pi}{2\Delta^{2}}\frac{c}{K}\right)~.
\end{equation}
This bound is the same as the bound between two codewords \(\mu,\nu\) of the underlying \(Kc\)-dimensional GKP code that are at least \(c\) ``clicks'' apart, i.e., for which \(|\mu-\nu| \geq c\).
In other words, the LCA code's entanglement with the physical qudit has ensured that the underlying GKP codewords \(\mu,\nu\) for which \(|\mu-\nu|<c\) never overlap with each other and thus do not contribute to the LCA codeword overlap.

To prove this, we first write out the overlap using Eq.~\eqref{eq:normalizable_lca_states}, evaluate the inner products between the physical qudit states, and plug in the asymptotic expression~\eqref{eq:gkp_overlap},  yielding
\begin{equation}
    \langle\overline{\mu_{\Delta}}|\overline{\nu_{\Delta}}\rangle\sim\frac{1}{c}\sum_{j^{\prime},j\in\Z_{c}}e^{-\frac{\pi}{2\Delta^{2}}\frac{\left(K[j^{\prime}-j]+\mu-\nu\right)^{2}}{Kc}}\delta_{(j^{\prime}-j)d,\nu-\mu}^{\Z_{c}}~,
\end{equation}
where \(\delta_{x,y}^{\Z_{c}} = 1\) if \(x = y - Mc\) for some integer \(M\), and zero otherwise.

We now simplify the delta function. First, we get rid of the factor of \(d\) in the delta function by recalling that it has an inverse modulo \(c\) via the Bezout identity~\eqref{eq:bezout}, i.e., \(-ad \equiv 1\) modulo \(c\).
Multiplying both sides by \(-a\) simplifies the delta-function term to \(\delta_{j^{\prime}-j,a(\mu-\nu)}^{\Z_{c}}\).
Solving the delta function yields
\begin{equation}
    j^{\prime}-j=a(\mu-\nu)-M_{j}c
\end{equation}
for \(j\)-dependent integers \(M_j\).

Plugging this in and simplifying using the Bezout identity~\eqref{eq:bezout}, this time over the integers, yields
\begin{equation}
    \langle\overline{\mu_{\Delta}}|\overline{\nu_{\Delta}}\rangle\sim\frac{1}{c}\sum_{j\in\Z_{c}}e^{-\frac{\pi}{2\Delta^{2}}\frac{c}{K}\left(\left[a\theta+b\right]\left[\mu-\nu\right]-KM_{j}\right)^{2}}~.
\end{equation}
The difference term in parentheses turns out to be \(\geq 1\) when \(\mu\neq\nu\) because \(a\theta+b\) is nonzero and is coprime to \(K = c\theta+d\).
}

\subsection{Error word overlap}

\vva{
A more fine-grained understanding~\cite{glancy2006error} of the intrinsic memory error caused by the \(\Delta\)-regularization of the ideal LCA states can be done by looking at the overlap of codewords with their displaced counterparts, which make up the error words~\eqref{eq:displaced-states} of the code.

We examine the \(\mu=0\) codeword for simplicity, with the other cases yielding similar scaling with code parameters.
Evaluating the qudit inner products simplifies the overlap to one that is purely between various GKP error words. We obtain
\begin{equation}
    \langle\xi_{1},\xi_{2},0|\overline{0_{\Delta}}\rangle\propto\sum_{j\in\Z_{c}}\langle\text{GKP}_{Kj}^{Kc}|e^{-i\xi_{2}\hat{x}}e^{i\xi_{1}\hat{p}}\left|\text{GKP}_{Kj}^{Kc}(\Delta)\right\rangle ~,
\end{equation}
where the bra term is an ideal GKP state, while the ket term is a normalized GKP state.

We can make the inner product \(j\)-independent by remembering that the ideal \(\mu\)th GKP codeword can be obtained from the \(\mu=0\) codeword by displacing the latter in position by \(\mu\sqrt{\frac{2\pi}{Kc}}\).
This no longer holds exactly for normalizable GKP codewords since different codewords come with different envelope terms, \(e^{-\frac{1}{2}\Delta^{2}f(s)^{2}}\), in Eq.~\eqref{eq:gkp-approximate}.
As such, we work in the \(\Delta\to 0\) limit and make this approximation for the normalizable state (ket term above),
\begin{equation}
    \left|\text{GKP}_{\mu}^{Kc}(\Delta)\right\rangle \sim e^{-i\mu\sqrt{\frac{2\pi}{Kc}}\hat{p}}\left|\text{GKP}_{0}^{Kc}(\Delta)\right\rangle \,,
\end{equation}
while also plugging in the exact expression for the bra term.

Permuting and simplifying the displacements yields a \(j\)-dependent phase, but disentanges the GKP overlap term.
Squaring the overlap and simplifying the sum over \(j\) yields
\begin{equation}\label{eq:LCA-errorword-overlap}
    \!\!\left|\langle\xi_{1},\xi_{2},0|\overline{0_{\Delta}}\rangle\right|^{2}\propto g(\xi_{2})\left|\langle\text{GKP}_{0}^{Kc}|e^{-i\xi_{2}\hat{x}}e^{i\xi_{1}\hat{p}}\left|\text{GKP}_{0}^{Kc}(\Delta)\right\rangle \right|^{2},
\end{equation}
which is the error-word overlap of a \(Kc\)-dimensional GKP code modulated by
\begin{equation}
    g(\xi_{2})=\frac{\sin^{2}(\xi_{2}\sqrt{\pi Kc/2})}{\sin^{2}(\xi_{2}\sqrt{\frac{\pi K}{2c}})}~.
\end{equation}
This function disappears in the GKP case, at which \(c=1\).

As \(\Delta\to 0\), the GKP overlap term in Eq.~\eqref{eq:LCA-errorword-overlap} is known~\cite{wang2019quantum,tzitrin2020progress} to roughly be a Gaussian centered at \(\xi_1=\xi_2=0\).
The modulating \(g\)-function is also peaked at \(\xi_2=0\), and no other peaks occur in the relevant interval \(\xi_2 \in [-\sqrt{\frac{\pi c}{2K}},\sqrt{\frac{\pi c}{2K}} )\).
The LCA error-word overlap thus behaves similarly to the GKP one.

When properly normalized, the error-word overlap can be integrated over one third of the correctable pure displacement range of each quadrature to yield a probability of intrinsic memory error during the Steane fault-tolerant syndrome extraction protocol~\cite{glancy2006error}. 

}

\subsection{Average occupation number \& effective dimension}

\vva{
To quantify the effective Hilbert space of the code, we define the following energy operator on the oscillator and qudit,
\begin{equation}
    \hat{N}=\hat{n}+\sum_{\ell=0}^{c-1}\ell|\ell\rangle\langle\ell|~,
\end{equation}
where \(\hat n\) is the usual occupation number on the mode, and where the second term is an analogous occupation number for the qudit, expressed in terms of the qudit's \(\hat{Z}\) eigenstates.

Evaluating the expectation value of the above energy operator in the codeword~\eqref{eq:normalizable_lca_states} yields, asymptotically as \(\Delta \to 0\),
\begin{equation}\label{eq:energy0}
    \langle\overline{\mu_{\Delta}}|\hat{N}|\overline{\mu_{\Delta}}\rangle\sim\frac{1}{2\Delta^{2}}+\frac{c}{2}-1\,.
\end{equation}
This is a sum of the usual GKP codeword energy and the qudit piece, \(\sum_{\ell=0}^{c-1} j/c = (c-1)/2\).
The energy diverges in the same way with decreasing \(\Delta\) as the GKP case, while diverging linearly with increasing \(c\).

The energy of a code can similarly be evaluated as
\begin{equation}\label{eq:energy}
    \overline n = \frac{1}{K} \text{Tr} (P \hat N)
\end{equation}
for \(K\)-dimensional code projector \(P\).
This, in turn, can be used to define the \textit{effective dimension}~\cite{iosue_continuous-variable_2024} of the bounded-energy oscillator-qudit state subspace that contains the code.
This quantity,
\begin{equation}
    q_{\text{eff}} = \lceil 2\overline n + 1 \rceil~, 
\end{equation}
can be thought of as the physical space dimension of oscillator-qudit states whose energies are less than \(\overline n\).

When the oscillator is unoccupied and all qudit states are used equally, \(\overline n = (c-1)/2\), the effective dimension correctly reduces to the qudit dimension \(c\).
On the other hand, when the qudit isn't used, the effective dimension quantifies the oscillator's subspace of bounded energy states~\cite[Eq.~(49)]{iosue_continuous-variable_2024}.

The effective dimension can be used to estimate the energy required to house a GKP or LCA code of logical dimension \(K\).
Plugging in the asymptotic \(\Delta \to 0\) expression~\eqref{eq:energy0} for the energy of a GKP codeword implies that \(\Delta^{-2}+c-1 \gg K\) in order for the codewords of such a code to be sufficiently resolved.
}

\subsection{Numerical results}
\label{subsec:numerics}

\vva{
\begin{figure}[t!]
    \centering
    \includegraphics[width=1.0\columnwidth]{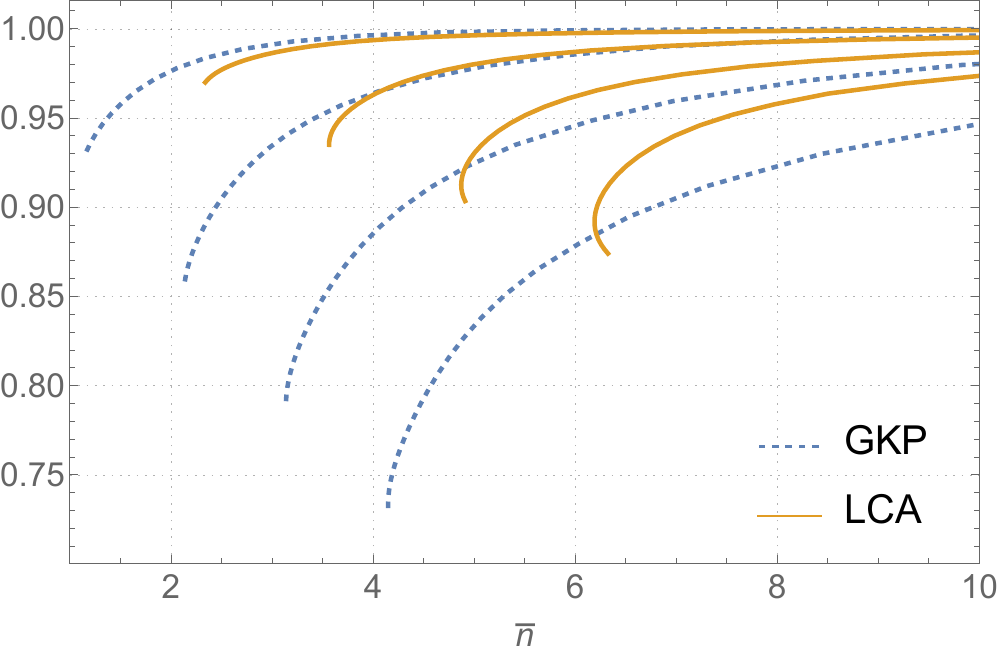}
    \caption{\small
    \vva{
    Entanglement fidelity versus energy~\eqref{eq:energy} under the transpose recovery of GKP and oscillator-qubit LCA codes against photon loss and qubit amplitude damping at the same noise rate \(\gamma = 0.05\) and various logical dimensions \(K\). 
    The four dotted curves correspond to fidelities of GKP codes of logical dimension \(K=3\), \(5\), \(7\), and \(9\), shown from left to right.
    Similarly, the four solid curves correspond to oscillator-qubit LCA codes with the same logical dimensions.
    The \(K=3,5\) encodings perform similarly, but LCA code fidelities at lower energy are higher than GKP fidelities for \(K=7,9\).
    To make this plot, we have implemented a cutoff of the oscillator occupation-number, and have numerically made an orthonormal basis out of the normalized LCA states. 
    We attach a \textsc{Mathematica} notebook with our numerical derivations to the arXiv submission of this manuscript.
    }
    }
    \label{fig_numerics}
\end{figure}

Despite no clear winner among the decoders discussed in Sec.~\ref{sec:qec}, LCA codes can perform similarly to, and sometimes outperform, GKP codes under physical noise channels.
We show this by evaluating \(c=2\) LCA code performance under the combination of photon loss on the mode and amplitude damping on the qubit.

We numerically evaluate the entanglement fidelity of LCA and GKP codes under the Petz recovery map \cite{petz1986sufficient}, which has an analytic expression \cite{zheng2024near} that allows for relatively efficient numerical evaluation.
This fidelity is at most within a factor of two away from the optimal entanglement fidelity~\cite{barnum2002reversing}, so while we cannot prove that LCA codes outperform GKP codes, a stronger performance under the Petz recovery is good evidence that this is the case.

For the noise channel, we consider the combination of the photon loss channel \cite{albert_performance_nodate} and qubit amplitude damping.
Since the latter channel can be obtained by restricting the former to the first two Fock states, we treat the physical qubit states as said Fock states and use the same noise parameter \(\gamma\) for both channels.

Figure~\ref{fig_numerics} plots the entanglement fidelity after Petz recovery for GKP and \((c=2)\)-LCA codes of logical dimension \(K\in \{3,5,7,9\}\) as a function of energy.
All curves increase monotonically with the energy, indicating that approximate LCA codes improve in performance as they approach ideal LCA codes.
The same effect was observed numerically for GKP codes~\cite{albert_performance_nodate}. 
We can also see that, as the logical dimension increases, LCA codes begin outperform GKP codes at low energy.
}

\section{Standard form of\\ simple LCA codes}
\label{sec:standard-form}

Simple \((c,d)\)-LCA codes can be analogously defined for multiple modes.
Such codes can be studied in a lattice formalism by collecting powers of each stabilizer generator into a lattice generator matrix \(T\). 

\subsection{Single-mode case}
We first build a lattice generator matrix out of the simple \((c,d)\)-LCA code from the previous section, with \(\theta = 0\) and nonzero \(K=d\in\mathbb{Z}_c\).
Each column of the matrix 
\begin{equation}
T=\begin{pmatrix}\sqrt{d/c} & 0\\
0 & \sqrt{d/c}\\
-d & 0\\
0 & 1
\end{pmatrix}\equiv\begin{pmatrix}T_{\cv}\\
T_{\dv}
\end{pmatrix}
\end{equation}
describes one of the two stabilizer generators of the single-mode single-qudit LCA code. The first two entries in each column, from top to bottom, are the amounts the generator displaces the mode's position and momentum, divided by \(\sqrt{2\pi}\).
The second two entries are the powers of \(\hat{X}\) and \(\hat{Z}\), in that order.

The second equality above expresses the matrix as a block matrix of two blocks, the \(2\times 2\) continuous-variable (\(\cv\)) block defining the mode's operators making up each generator, and the similarly sized discrete-variable (\(\dv\)) block defining the qudit's operators.

Encoding a set of commuting operators, the matrix \(T\) satisfies a symplectic condition 
\begin{equation}\label{eq:symplectic-condition-zerotheta}
T^{\intercal} JT=0~,\quad\quad\text{where}\quad\quad J = J_{\cv}\oplus J_{\dv},
\end{equation}
where \(T^{\intercal}\) is the transpose of \(T\), and where the hybrid symplectic form \(J\) is a block-diagonal matrix consisting of \(\cv\) and \(\dv\) pieces,
\begin{equation}
J_{\dv}=\begin{pmatrix}0 & 1/c\\
-1/c & 0
\end{pmatrix}\quad\text{and}\quad J_{\cv}=\begin{pmatrix}0 & 1\\
-1 & 0
\end{pmatrix}\,.
\end{equation}
Since the oscillator and qudit pieces do not commute independently, the \(\cv\) and \(\dv\) pieces do not individually satisfy their own symplectic conditions.
Instead, the two pieces contribute the same matrix \(Z\), but with opposite signs, 
\begin{subequations}
\begin{align}
T_{\dv}^{\intercal} J_{\dv}T_{\dv}&=+Z=\frac{d}{c}\begin{pmatrix}0 & -1\\
1 & 0
\end{pmatrix}\\T_{\cv}^{\intercal} J_{\cv}T_{\cv}&=-Z~.
\end{align}
\end{subequations}
The contributions of these two pieces cancel, thereby satisfying Eq.~\eqref{eq:symplectic-condition-zerotheta}.

\subsection{Multi-mode normal form}

The single-mode code yields a natural generalization to \(p\geq 1\) modes and qudits.
We define diagonal matrices,
\begin{subequations}
\begin{align}
\dia d&=\text{diag}(d_{1},d_{2},\cdots,d_{p})\\\dia c&=\text{diag}(c_{1},c_{2},\cdots,c_{p})\,,
\end{align}
\end{subequations}
listing the dimensions and code parameters for each \((c_j,d_j)\)-LCA qudit code.
From now on, \(\dia{x}\) denotes a diagonal matrix whose entries are \(x_j\) for any \(x_j > 0\).

The multi-mode standard form is
\begin{subequations}\label{eq:canonical-standard-form}
\begin{align}
    \ts&=\begin{pmatrix}\ts_{\cv}\\
\ts_{\dv}
\end{pmatrix}\\\ts_{\cv}&=\dia{1}_{2}\otimes\sqrt{\dia{d/c}}\\\ts_{\dv}&=(\dia{-d})\oplus \dia{1}_{p}~,
\end{align}
\end{subequations}
where \(\dia{1}_p\) is the \(p\)-dimensional identity.
The \(2p\)-dimensional symplectic forms are now
\begin{subequations}\label{eq:symplectic-forms}
\begin{align}
J_{\dv}&=\begin{pmatrix}0 & 1\\
-1 & 0
\end{pmatrix}\otimes\dia c^{-1}\\J_{\cv}&=\begin{pmatrix}0 & 1\\
-1 & 0
\end{pmatrix}\otimes \dia{1}_{p}\,.
\end{align}
\end{subequations}
with the full symplectic form still
\(J = J_{\cv} \oplus J_{\dv}\).

The degree of non-commutation for the \(\cv\) and \(\dv\) parts is the same magnitude but of opposite sign, so the two pieces cancel each other to yield a stabilizer code,
\begin{equation}
    \ts^{\intercal} JT=\ts_{\dv}^{\intercal} J_{\dv}\ts_{\dv} - \ts_{\cv}^{\intercal} J_{\cv}\ts_{\cv}=Z-Z=0~,
\end{equation}
where the matrix \(Z\) is
\begin{equation}
Z=\begin{pmatrix}0 & -1\\
1 & 0
\end{pmatrix}\otimes\dia{d/c}~.
\end{equation}

The logical dimension \(K\) is just the product of each code's dimensions \(d_j\). 
We can verify this using the general theory defined in the next section,
\begin{align}
K=|\text{Pf}\,\dia{d/c}|\det\dia c=\det\dia c\prod_{j=1}^{p}\frac{d_{j}}{c_{j}}=\prod_{j=1}^{p}d_{j}~.
\end{align}
Above, we use the formula for the Pfaffian of an anti-symmetric matrix.

The matrix \(Z\) encodes the degree of non-commutation of the \(\cv\) and \(\dv\) pieces of the code.
On the \(\dv\) side, it encodes an algebra of \(j\) pairs of Pauli matrices that commute up to the phase \(\exp(2\pi i d_j/c_j)\) and is an example of a \textit{commutation matrix}~\cite{englbrecht2022transformations,gunderman2023transforming,sarkar2024qudit} (cf. the frustration graph matrix~\cite{makuta2025frustration,mann2025graph}).

\subsection{Relation to commutation matrices}
\label{sec:smith}

We can obtain the LCA standard form~\eqref{eq:canonical-standard-form} by starting with a general commutation matrix of a Pauli algebra, determining the minimal qudit dimensions needed to realize that algebra, and appending an oscillator displacement operator to each Pauli to compensate their non-commutativity.

We consider a Pauli matrix set without specifying the dimension of the qudits that this set acts on.
Such a set is defined by its group commutators, which, for any two elements \(P_j, P_k\) of the set, are of the form \(P_i P_j P_i^{\dagger} P_j^{\dagger} = \exp(-i 2\pi v/u)\) for some integers \(v,u\).
We convert this information into a \(2p \times 2p\) anti-symmetric rational commutation matrix \(A\)~\cite{englbrecht2022transformations,gunderman2023transforming,sarkar2024qudit} whose \((j,k)\)th entry is \(v/u\), and whose \((k,j)\)th entry is \(-v/u\).

Any rational matrix \(A\) can be made into an integer matrix \(m A\) by multiplying  by \(m\) --- the lowest common denominator (lcd) of all denominators.
To determine the minimal dimensions \(c_j\) of the \(p\) qudits required to realize the Pauli algebra defined by \(A\), we first bring the commutation matrix \(mA\) into \textit{alternating Smith normal form} (a.k.a. skew-normal form, normal form with respect to congruence, or symplectic normal form)~\cite[Thm. 18]{kuperberg2001kasteleyn}
\cite{SageSymplecticBasis,sarkar2024qudit,conrad2024fabulous}
\cite[p. 598, Exercise XV.17]{lang_algebra}.
This can be done via an invertible integer matrix \(R\), which, by definition, has to have determinant \(\pm 1\) (i.e., is unimodular).
Such matrices form the group \(\mathrm{GL}(2p,\mathbb{Z})\), and one of them brings \(mA\) into normal form that satisfies
\begin{equation}\label{eq:normal_form}
R^{\intercal} AR=\begin{pmatrix}0 & -1\\
1 & 0
\end{pmatrix}\otimes\frac{\boldsymbol{h}}{m}\quad\quad\text{(alt.\@ Smith form)},
\end{equation}
where \(\dia{h}\) is a diagonal matrix of non-negative integers \(h_j\) that are ordered so that they divide each other, i.e., \(h_1|h_2|\cdots|h_p\).
In this section, we assume none of the \(h_j\)'s are zero for simplicity.

We can equate each diagonal element \(h_j/m\) to a ratio \(K_j/c_j\) with \(\gcd(K_j,c_j)=1\).
Moreover, we can split up \(K_j = c_j \theta_j+d_j\) into some integer multiple \(\theta_j\) of \(c_j\) and a nonzero remainder \(d_j \in\mathbb{Z}_{c_j}\).  
These are defined to match the notation of the LCA codes.
Altogether, we have
\begin{equation}
    \frac{\boldsymbol{h}}{m}=\dia{\theta}+\boldsymbol{d/c}~,
\end{equation}
where we have used the notation \(\dia{x} = (x_1,x_2,\cdots)\) for any bold symbol.

We can absorb the \(\left(\begin{smallmatrix}0 & -1\\
1 & 0
\end{smallmatrix}\right)\) matrix in Eq.~\eqref{eq:normal_form} and define
\begin{equation}
    R^{\intercal} AR=\varTheta-Z~,
\end{equation}
where the matrices
\begin{subequations}
\begin{align}
\varTheta&=\begin{pmatrix}0 & 1\\
-1 & 0
\end{pmatrix}\otimes\dia{\theta}\\Z&=\begin{pmatrix}0 & -1\\
1 & 0
\end{pmatrix}\otimes\dia{d/c}~.
\end{align}
\end{subequations}
Setting \(\varTheta = 0\) yields the same \(Z\) matrix as that arising from a tensor product of simple LCA codes in standard form~\eqref{eq:canonical-standard-form}.

In summary, the above procedure splits a general anti-symmetric matrix \(A\) into a piece \(Z\) that can be used to define a simple multi-mode LCA code, and another piece \(\varTheta\) that can toggle the code's logical dimension.
A slight modification of the oscillator encoding map accommodates this toggling,
\begin{subequations}
\begin{align}
    \ts_{\cv}&=\dia{1}_{2}\otimes\sqrt{\dia{\theta+d/c}}\\\ts_{\dv}&=(\dia{-d})\oplus \dia{1}_{p}~,
\end{align}
\end{subequations}
yielding a more general tensor-product LCA code.
This code's encoder satisfies
\begin{equation}\label{eq:encoder}
    T^{\intercal}JT=\varTheta-Z+Z = \varTheta~,
\end{equation}
the defining condition for an integral symplectic lattice.
The GKP case is recovered by setting all \(c_j=d_j=1\).

The partitioning \(K_j = c_j \theta_j+d_j\) into a multiple of \(c_j\) and a remainder is not necessary, and an alternative formulation can relax both \(\theta_j\) and \(d_j\) to be arbitrary integers (see Appx.~\ref{app:general-single-mode-integer-case} for such a treatment).
That way, a \(\theta_j\) can be absorbed into \(d_j\), implying that different pairs \(\theta_j,d_j\) can be used to construct the same LCA code.

These constructions extend further to those that consider \(\varTheta\) and \(Z\) as \textit{separate} degrees of freedom.
In other words, the two matrices cannot put into alternating Smith normal form using the same unimodular transformation.

\section{General LCA codes}\label{sec:multimode}

We borrow~\cite{rieffel,schwarz1998morita,connes1998noncommutative,rieffel1999morita,seiberg1999string,yoneya2000string,li2004strong,elliott2008strong} from a decidedly general construction of embeddings of the form~\eqref{eq:encoder} using a given anti-symmetric rational commutation matrix \(Z\) and \textit{arbitrary} \(\varTheta\) as long as \(\varTheta - Z\) is invertible.
We continue to keep \(\varTheta\) an integer matrix here, which specializes the construction to embeddings of integral symplectic lattices into physical oscillators and qudits.
Each \(Z\) defines the number and dimensions of the physical qudits and corresponds to a particular class of LCA codes, yielding infinite extensions of GKP codes and their underlying lattice theory \vva{in which the Gram matrices \(\varTheta\) need not commute with \(Z\).}

\subsection{From lattices to non-commutative tori}

The defining equation~\eqref{eq:encoder} for an LCA code, \(T^{\intercal}JT=\varTheta\), can be interpreted as an embedding of a lattice into a system of oscillators and qudits defined by the hybrid symplectic form \(J = J_{\cv}\oplus J_{\dv}\)~\eqref{eq:symplectic-forms}.
Columns of the generator matrix \(T\) define stabilizers that leave the lattice invariant, while the anti-symmetric integer Gram matrix \(\varTheta\)
defines how these stabilizers commute.

More precisely, two stabilizers, \(U_i\) and \(U_j\), commute as
\begin{equation}
U_{i}U_{j}=e^{-2\pi i\varTheta_{ij}}U_{j}U_{i}\,.
\end{equation}
This stabilizer algebra is an example of a \textit{non-commutative torus}.
In this case, the torus is commutative because \(\varTheta\) is integer, but rational matrices, such as commutation matrices (see Sec.~\ref{sec:smith}), or irrational matrices define non-commutative variants.

The stabilizer algebra is called a torus because it is in one-to-one correspondence with the unit cell of the lattice defined by the stabilizer generators.
For example, the unit cell of the square-lattice GKP code, whose phase-space lattice \(\mathscr{D}=\mathbb{Z}^2\), is the two-dimensional torus \(\R^2/\mathscr{D} = \mathbb{T}^2\), where \(\mathbb{T}\cong U(1)\) is the circle group.
Since ``tori'' and ``lattices'' are dual to each other in this way, we use the two words interchangeably to describe \(\varTheta\).

The coefficients of LCA logical displacements, \(\xl\) and \(\zl\)~\eqref{eq:logicals}, can also be collected into columns of their own generator matrix \(S\).
This matrix defines an embedding of the code's dual lattice and its associated dual torus \(\varTheta^{\perp}\), satisfying 
\begin{equation}\label{eq:dual-torus}
    S^{\intercal}JS=-\varTheta^{\perp}~.
\end{equation}
Here, \(\varTheta^{\perp}\) is rational since the logicals do not commute, making this dual torus non-commutative.

In the case of GKP codes~\cite{10.1103/physreva.64.012310}, the dual torus is related to the original as \(\varTheta^\perp = \varTheta^{-1}\).
The purpose of the non-commutative torus formalism is to determine the dual torus and the corresponding logical operators, defined by \(S\), for general LCA codes.

\subsection{Morita equivalence and \(\textrm{SO}(2p,2p~|~\Z)\)}

The torus \(\varTheta\) and its dual \(\varTheta^\perp\) are related by a general duality transformation --- Morita equivalence --- that is valued in a particular group.

Consider first the single-mode \((c,d)\)-LCA code from Sec.~\ref{sec:single-mode}.
The code's two-dimensional torus \(\varTheta=\theta\left(\begin{smallmatrix}0 & 1\\
-1 & 0
\end{smallmatrix}\right)\), is defined solely by the stabilizer commutator phase \(\theta\)~\eqref{eq:lattice-comm}.
The code's dual torus is determined by the commutator phase of the logical operators~\eqref{eq:logical-comm}.
This phase is obtained from the stabilizer commutator phase via the mapping 
\begin{equation}
\theta \to \frac{a\theta+b}{c\theta+d}~,
\end{equation}
where \((c,d)\) are the defining parameters of the code, and where \((a,b)\) are obtained from these parameters via the Bezout identity~\eqref{eq:bezout}, \(bc-ad=1\).
The four integers can be arranged into the matrix \(\left(\begin{smallmatrix}-a & -b\\
c & d
\end{smallmatrix}\right)\) whose determinant is one, revealing that the mapping between the tori is a special case of a Möbius transformation. 

The mapping between phases can be upgraded to a matrix-valued mapping directly between \(\varTheta\).
For an LCA code on \(p\) modes, define the \(4p\)-dimensional block matrix
\begin{equation}
    g=\begin{pmatrix}A & B\\
C & D
\end{pmatrix}~,
\end{equation}
which acts on the matrix defining the torus as
\begin{equation}\label{eq:morita}
    g\cdot\varTheta=(A\varTheta+B)(C\varTheta+D)^{-1},
\end{equation}
and whose \(2p\)-dimensional blocks \(A,B,C,D\) satisfy
\begin{subequations}
    \label{eq:son-conditions}
\begin{align}
    A^{\intercal}C+C^{\intercal}A=B^{\intercal}D+D^{\intercal}B&=0\\A^{\intercal}D+C^{\intercal}B&=\dia{1}_{2p}\,.
\end{align}
\end{subequations}

The conditions on the blocks imply that \(g\) lies in the group \(\textrm{SO}(2p,2p~|~\Z)\), the group of orientation-preserving (i.e., determinant-one) integer-valued rotations in \(4p\) dimensions that preserve the quadratic form $x_1 x_{2p+1}+x_2 x_{2p+2}+\cdots+x_{2p} x_{4p}$.
Any two tori related by such a \(g\) via Eq.~\eqref{eq:morita} can be shown to be \textit{Morita equivalent}~\cite{rieffel1999morita,li2004strong,li2004strong,elliott2008strong}.
This yields a classification of different tori into Morita-equivalent classes.
Reference~\cite{dereli2021bloch} discusses Morita equivalence in the case of GKP codes.

This classification is exceedingly coarse-grained.
Not only is a commutative torus \(\varTheta\) Morita equivalent to its non-commutative dual, \textit{all} rational tori are equivalent to each other~\cite[Corr. 4.2]{elliott2008strong} because \(\text{SO}(2p,2p~|~\Z)\) transformations allow one \sch{to add, rotate, and invert (whenever possible) by arbitrary integers.}
Nonetheless, one can work out the precise \(g\) necessary to obtain the dual torus of any given torus.
This yields logical operators for general LCA codes.

\subsection{General construction}

We state the LCA construction in the language of quantum coding theory and provide a parallel mathematical discussion in Appxs.~\ref{app:lca-intro}-\ref{app:general-lca}.

General LCA codes are defined as embeddings of integer anti-symmetric \(\varTheta\) into \(p\) modes and \(k \leq p\) qudits.
The number and local dimensions of the qudits are determined by a given commutation matrix \(Z\) via its alternating Smith normal form (see Sec.~\ref{sec:smith}). 
The number of qudits \(k < p\) whenever this  normal form is not full rank \sch{over the integers}.
Any stabilizer algebra defined by \(\varTheta\) is embeddable in the way defined by \(Z\) as long as \(\varTheta - Z\) is invertible.
We summarize this in the following theorem.

\begin{theorem}[Informal.]\label{thm:main_text}
Let $Z$ be a rational anti-symmetric matrix of dimension \(2p\). Then there exists a rotation $g \in \mathrm{SO}(2p, 2p \mid \mathbb{Z})$ 
such that, for any $\varTheta$ with $\varTheta-Z$ invertible, 
there exist generator matrices $T$ and \(S\) defining the stabilizers and logical operators, respectively, of a \(p\)-mode \((k\leq p)\)-qudit LCA code with torus \(\varTheta\)~\eqref{eq:encoder} and dual torus \(\varTheta^\perp = g\cdot\varTheta\)~\eqref{eq:dual-torus}.
\end{theorem}

To construct the generator matrices, we first bring the integer matrices \(mZ\) and \(m(\varTheta - Z)\) into alternating Smith normal form, where \(m\) is the lowest common denominator of any denominators in \(Z\) (we do not need to bring in \(\varTheta\) here because it is integer).
This yields
\begin{subequations}\label{eq:smith-theta-z}
\begin{align}
\varTheta-Z&=Q^{\intercal}\begin{pmatrix}0 & \dia t/m\\
-\dia t/m & 0
\end{pmatrix}Q\\Z&=R^{\intercal}\begin{pmatrix}0 & \dia h/m & 0\\
-\dia h/m & 0 & 0\\
0 & 0 & 0
\end{pmatrix}R~,
\end{align}
\end{subequations}
where the transformations \(Q,R \in \text{GL}(2p,\Z)\).
We are guaranteed that \(\varTheta - Z\) is invertible, but \(Z\) may have rank \(2k < 2p\). 
Hence, the diagonal matrix \(\dia t \equiv \text{diag}(t_1,t_2,\cdots)\) is \(p\)-dimensional while \(\dia h\) is \(k\)-dimensional.
This step yields the number of qudits \(k\) necessary to construct the LCA code defined by \(Z\).

Next, we re-express each element \(h_j/m\) in terms of the pair \((c_j,d_j)\) satisfying \(\gcd(c_j ,d_j)=1\) (see Sec.~\ref{sec:smith}), yielding \(\dia h/m=\dia{d/c}\).
The \(c_j\) give us the qudit dimensions, and \(Z\) can be converted into
\begin{equation}\label{eq:gcd-stuff}
    Z=R^{\intercal}\begin{pmatrix}0 & \dia{d/c} & 0\\
-\dia{d/c} & 0 & 0\\
0 & 0 & 0
\end{pmatrix}R~.
\end{equation}
In contrast to the case of simple LCA codes, here we do not assume that \(d < c\).

The stabilizer-generator and logical-Pauli matrices are of the form
\begin{equation}
    \ts=\begin{pmatrix}\ts_{\cv}\\
\ts_{\dv}
\end{pmatrix}\quad\quad\text{and}\quad\quad S=\begin{pmatrix}S_{\cv}\\
S_{\dv}
\end{pmatrix}\,.
\end{equation}
Their columns define the displacement coefficients of each generating operator in the order ``\(2\pi\)-position shift, \(2\pi\)-momentum shift, Pauli-\(\hat{X}\) power, Pauli-\(\hat{Z}\) power''.

The former has respective \(\cv\) and \(\dv\) pieces
\begin{subequations}\label{eq:encoder-general}
\begin{align}
   T_{\cv}&=\begin{pmatrix}0 & \dia{\sqrt{t}}\\
-\dia{\sqrt{t}} & 0
\end{pmatrix}Q\\T_{\dv}&=\begin{pmatrix}\dia d & 0 & 0\\
0 & \dia{1}_{k} & 0
\end{pmatrix}R~,
\end{align}
\end{subequations}
where the matrix with \(\dia d\) in it is \(2k\)-by-\(2p\) dimensional.
These can be checked to satisfy Eq.~\eqref{eq:encoder}, realizing the torus \(\varTheta\).

The latter has pieces
\begin{subequations}
\begin{align}
   S_{\cv}&=J_{\cv}T_{\cv}^{-\intercal}R^{\intercal}\left(\begin{array}{ccc}
\dia c^{-1} & 0 & 0\\
0 & \dia c^{-1} & 0\\
0 & 0 & -\dia{1}_{2(p-k)}
\end{array}\right)\\S_{\dv}&=\left(\begin{array}{ccc}
0 & -\dia{1}_{k} & 0\\
\dia a & 0 & 0
\end{array}\right)~,
\end{align}
\end{subequations}
where \(T^{-\intercal}=(T^{-1})^{\intercal}\), and where each \(a_j\) satisfies its own Bezout identity, \(a_{j}d_{j}+b_{j}c_{j}=1\).
A tedious calculation can be done to check that these satisfy Eq.~\eqref{eq:dual-torus} with \(\varTheta^\perp = g\cdot \varTheta\) for a particular \(g\).

\prg{Logical dimension}
Calculating the logical dimension \(K\) of LCA codes starts out in the same way as with GKP codes.
Namely, \(K^2\) is the ratio of the size of the torus \(\varTheta\) to its dual \(\varTheta^{\perp}\), i.e., the ratio of the unit-cell volumes of the stabilizer phase-space lattice \(\mathscr D\) to its dual the logical lattice \(\mathscr D^\perp\),
\begin{equation}
    K^{2}=\frac{|(\R^{2p} \times \Z_{\dia c}^2)/\mathscr{D}|}{|(\R^{2p} \times \Z_{\dia c}^2)/\mathscr{D}^{\perp}|}~.
\end{equation}
Here, we use the shorthand \(\Z_{\dia c} = \Z_{c_1}\times \Z_{c_2}\times\cdots\times\Z_{c_k}\).

An interesting simplification (see Appx.~\ref{app:logical-dim}) re-expresses this in terms of a Pfaffian to yield
\begin{equation}
K=\left|\text{Pf}(\varTheta-Z)\right|\det\dia c=\frac{\det\dia{ct}}{m^{p}}\,.
\end{equation}
It is not obvious this is an integer, but it is so because the product of the qudit dimensions is designed to cancel out any denominators in the Pfaffian.

\section{Logical Clifford gates}
\label{sec:logical-gates}

GKP codes admit logical Clifford gates via Gaussian transformations.
We anticipate that LCA codes similarly admit logical Clifford gates via Gaussian-Clifford transformations.
We define such transformations and prove this for several families of LCA codes.

\subsection{Gaussian-Clifford transformations}

Continuous and modular degrees of freedom cannot be mixed in a way that preserves the hybrid displacement-Pauli group (see Sec.~\ref{sec:symplectic-entanglement} and Appx.~\ref{app:single-mode-symplectic}).
The hybrid symplectic group on \(p\) modes and \(k\) qudits therefore splits into the Gaussian and hybrid Clifford groups,
\begin{equation}
   \text{Sp}(2p,\R)\times\text{Sp}(2k,\Z_{\dia c})\,.
\end{equation}
The latter factor is shorthand for the group of Pauli-group preserving operations on \(k\) qudits of dimensions \(c_1\) through \(c_k\).
This group includes the Clifford group of each qudit \(j\), \(\text{Sp}(2,\Z_{c_j})\), as well as any entangling gates between qudits.

A Gaussian-Clifford operation can be represented by the block matrix
\begin{equation}\label{eq:gaussian-clifford}
    V = V_{\cv} \oplus V_{\dv}~,
\end{equation}
where the first piece is the symplectic representation of a Gaussian gate, and the second piece is the integer-matrix representation of a hybrid Clifford gate,
\begin{subequations}
\begin{align}
V_{\cv}&\in\text{Sp}(2p,\R)\\V_{\dv}&\in\text{Sp}(2k,\Z_{\dia c})~.\label{eq:qudit-symplectic}
\end{align}
\end{subequations}
We define these to satisfy the following conditions with respect to their corresponding symplectic forms~\eqref{eq:symplectic-forms},
\begin{subequations}
\label{eq:clifford}
\begin{align}
V_{\cv}^{\intercal}J_{\cv}V_{\cv}&=J_{\cv}
\label{eq:clifford-cv}\\
V_{\dv}^{\intercal}J_{\dv}V_{\dv}&=J_{\dv}~.\label{eq:clifford-dv}
\end{align}
\end{subequations}
For fixed qudit dimension \(c_j = c\), the \(\dv\) symplectic condition~\eqref{eq:clifford-dv} is typically defined modulo \(c\), but we define it over the integers to handle cases with different qudit dimensions. 
As such, the actual qudit symplectic operation corresponding to \(V_{\dv}\) is one whose matrix elements are reduced modulo \(c_j\), with \(j\) depending on the qudit the matrix element is acting on.

The symplectic action of the \(\cv\) and \(\dv\) pieces of a given Gaussian-Clifford is realized by a unitary acting on the oscillator-qudit system.
The group of such unitaries is called the oscillator~\cite{KAIBLINGER2009233,howe2012non,fiber-bundle-ft}\cite[Thm. 128]{de2011symplectic} and qudit~\cite{burton2024genons,gurevich2012weil,KAIBLINGER2009233,10.1093/qmath/ham023,hayashi2017group} metaplectic group, \(\textrm{Mp}(2p,\R)\) and \(\textrm{Mp}(2k,\Z_{\dia c})\), respectively . 
Each group is a cover of its corresponding oscillator or qudit symplectic group, meaning that multiple metaplectic elements can reduce to the same symplectic element.
This reduction is done by ignoring any global phases and displacement/Pauli actions.
As such, when promoting a symplectic operation to a metaplectic one, one may need to apply a displacement/Pauli operator and multiply by a particular phase to achieve a well-defined action (see Appx.~\ref{app:gates}).

Prepending a Gaussian-Clifford operation to an LCA code yields two encoders, \(T\)~\eqref{eq:encoder-general} and \(VT\), of the same lattice \(\varTheta\) for any \(Z\).
Our general code construction can be explicitly augmented by Clifford transformations (see Appx.~\ref{app:general-lca-qudit-clifford}), and certain LCA codes with the same standard form can be related via such operations (see Appx.~\ref{app:state-prep}).
Here, we show that a particular subgroup of such transformations yields logical gates.

\subsection{Logical gates from lattice symmetries}

Each torus \(\varTheta\) comes with its own abstract symmetries (see~\cite{elpw2010structure, jeong2015finite,chakluef2019metaplectic,chaktracing2023tracing,chaksymmet2025symmetrized}), and we conjecture that such symmetries can be used to implement logical Clifford gates for general LCA codes. 
We review this result for the GKP case, extending it to single-mode \((c,d)\)-LCA codes and general LCA codes with \(\dia{d}=\dia{1}\). 

The symmetry group, or \textit{automorphism group}, of a \(2p\)-dimensional torus \(\varTheta\) is the set of all unimodular transformations that leave the torus alone,
\begin{equation}
\operatorname{Aut}(\varTheta)=\left\{ W\in\mathrm{GL}(2p,\mathbb{Z})\mid W^{\intercal}\varTheta W=\varTheta\right\}.
\end{equation}
These depend on the structure of \(\varTheta\), which is a general anti-symmetric integer matrix.
But they simplify to
\begin{equation}
    \operatorname{Aut}(\varTheta)=\operatorname{Sp}(2,\Z)\quad\quad\quad\text{(\ensuremath{p=1} case)}
\end{equation}
for the two-dimensional case because all \(W\)'s are integer, and because the torus is proportional to the oscillator symplectic form, \(\varTheta = \theta J_{\cv}\) for integer \(\theta\).

\prg{GKP case}
Embedding a two-dimensional lattice into an oscillator yields a GKP code.
It is known that each automorphism \(W\) can be implemented via some Gaussian operation \(V_{\cv}\), i.e.,
\begin{equation}
    T_{\cv}W=V_{\cv}T_{\cv}\quad\quad\leftrightarrow\quad\quad V_{\cv}=T_{\cv}WT_{\cv}^{-1}~.
\end{equation}
Here, we invert \(T_{\cv}\) to solve for \(V_{\cv}\), and it is indeed a Gaussian operation since it satisfies the symplectic condition:
\begin{subequations}\label{eq:clifford-derivation}
\begin{align}
    V_{\cv}^{\intercal}J_{\cv}V_{\cv}&=T_{\cv}^{-\intercal}W^{\intercal}T_{\cv}^{\intercal}J_{\cv}T_{\cv}WT_{\cv}^{-1}\\&=T_{\cv}^{-\intercal}W^{\intercal}\varTheta WT_{\cv}^{-1}\\&=T_{\cv}^{-\intercal}\varTheta T_{\cv}^{-1}\\&=J_{\cv}\,.
\end{align}
\end{subequations}
Acting with this transformation on the GKP code performs a logical Clifford operation~\cite{10.1103/physreva.64.012310}.

An example \(V_{\cv}\) is the Fourier transform, \(\hat p \to -\hat x\) and \(\hat x \to \hat p\), which acts on the GKP stabilizers [Eq.~\eqref{eq:stabilizers} for \(c=1\)] as
\begin{subequations}
\label{eq:transform}
\begin{equation}
    \st_{X} \to \st_{Z} \quad\quad\text{and}\quad\quad \st_{Z} \to \st_{X}^\dagger~.
\end{equation}
By observing its action on the logical operators [Eq.~\eqref{eq:logicals} for \(c=1\)], we see that it implements a logical Hadamard transformation, 
\begin{equation}
\xl\to\zl^{\dagger}\quad\quad\text{and}\quad\quad\zl\to\xl~,
\end{equation}
\end{subequations}
on the GKP code.
The Fourier transform is a rigid rotation on the oscillator's phase space (a.k.a. a passive linear optical transformation), so it does not amplify displacement noise and is thus fault-tolerant with respect to such noise~\cite{10.1103/physreva.64.012310,10.48550/arxiv.2407.03270,fiber-bundle-ft}.

\prg{\((c,d)\)-LCA case}
We can extend the above logic to the case of a general single-mode single-qudit LCA code, defined by parameters \((c,d)\) satisfying the GCD condition \(\gcd(c,d)=1\).
In this case, the stabilizers~\eqref{eq:stabilizers} come with qudit pieces, and any logical gate that permutes oscillator parts of different stabilizers via \(V_{\cv}\) has to be compensated accordingly on the qudit factor by a Clifford gate \(V_{\dv}\).

We can try to solve for the Clifford piece,
\begin{equation}
    T_{\dv}W=V_{\dv}T_{\dv}\quad\quad\leftrightarrow\quad\quad V_{\dv}=T_{\dv}WT_{\dv}^{-1}~.\label{eq:divisibility-clifford}
\end{equation}
This transformation satisfies the version of Eq.~\eqref{eq:clifford-derivation} with \(\cv\to\dv\) and \(\varTheta \to Z\) since \(Z\) is proportional to \(\varTheta\) for two-dimensional tori.
However, there is no guarantee that \(V_{\dv}\) is an integer matrix due to the presence of \(T_{\dv}^{-1}\).

To remedy this, we restrict to only those \(W\)'s which yield integer \(V_{\dv}\) gates.
That way, each \(V_{\dv}\) yields a valid \(c\)-dimensional qudit Clifford transformation whose symplectic matrix is the mod-\(c\) reduction of \(V_{\dv}\).
Equation~\eqref{eq:divisibility-clifford} implies that such realizable symmetries come from a particular \textit{congruence group}~\cite{ZerbesModularForms2022},
\begin{equation}
\Gamma_{0}(d)=\left\{ W=\begin{pmatrix}\star & \star\\
r & \star
\end{pmatrix}\in\operatorname{Sp}(2,\mathbb{Z})\;\middle|\;d\text{ divides } r\right\} .
\end{equation}
This is a proper subgroup of the automorphism group, but reducing this group modulo \(c\) still yields all possible qudit Clifford gates~\cite[Lemma 3.2.1]{ZerbesModularForms2022}.

We extend the GKP Fourier gate to the LCA case~\cite{walters2000chern}. 
We recall the other two code parameters \((a,b)\) that arise from the Bezout identity~\eqref{eq:bezout}, \(bc-ad=1\).
Performing the aforementioned Fourier transform on the oscillator while mapping the qudit Paulis as~\cite{gottesman1998fault}
\begin{equation}
    \hat{X}\to \hat{Z}^{a}\quad\quad\text{and}\quad\quad Z\to \hat{X}^{d}
\end{equation}
implements a logical Hadamard on the LCA code, satisfying Eqs.~\eqref{eq:transform} for general \((c,d)\).

\prg{\(\dia d = \dia 1\) multi-mode case}
In this case, we set \(\dia d\) to be the all-ones vector.
That way, the pseudoinverse \(T_{\dv}^{-1}\) is an integer matrix, and all symmetries \(W\) can be realized as qudit Clifford transformations represented by \(T_{\dv} W T_{\dv}^{-1}\)~\eqref{eq:divisibility-clifford}.
Such gates are available for any lattice since there are no additional assumptions on \(\varTheta\).

\section{Examples}

We discuss ways to concatenate LCA codes, mention logical-mode encodings, and provide further multi-mode multi-qudit examples using integer symplectic matrices and linear binary codes.

\prg{Hybrid concatenation}

Any qudit stabilizer code can be converted into an LCA code by encoding one or more of its physical qudits into a GKP qudit.
The non-Gaussian nature of the GKP encoder makes it impossible to encode information into this code using only Gaussian-Clifford operations.

One can also take any qudit stabilizer code and encode each of its physical \(d\)-dimensional qudits into a \((c,d)\)-LCA qudit.
This reduces to a qudit stabilizer code concatenated with a GKP code for \(c=1\).
Including an additional discrete degree of freedom for each physical system may lower code thresholds further than \eczoohref[concatenating with a GKP code]{gkp_concatenated}~\cite{fukui2017analog,fukui2018high,vuillot2019toricGKP,noh2020fault}.

\prg{Logical-mode encodings}

Minimial analogues of oscillator-into-oscillator GKP codes~\cite{noh2020encoding} can be constructed  on a two-mode and single-qudit system as follows.
Initialize mode 1 in an arbitrary state, initialize mode 2 and the qudit in an LCA state, and then apply a two-mode Gaussian transformation on the two modes.
This encoding extends the two-mode-squeezing code, which underlies the standard form of oscillator-into-oscillator GKP codes~\cite{wu2023optimal}, to the LCA case.
The use of the qudit may reduce the logical noise floor of the resulting logical-mode code against displacement noise.

\prg{Integer symplectic matrices}

Consider LCA codes with \(p=k\) and identical qudit dimension, \(c_j=c>1\).
Integer symplectic matrices, \(M\in\text{Sp}(2p,\Z)\), form a special class of operations that can be used \textit{simultaneously} for the \(\cv\) and \(\dv\) parts of a Gaussian-Clifford transformation \(V\)~\eqref{eq:gaussian-clifford} for such codes.

Integer symplectic matrices form the integer-valued matrix subgroup of the \(\cv\) symplectic matrices, \(\text{Sp}(2p,\Z)\subset \text{Sp}(2p,\R)\).
They also satisfy the \(\dv\) constraint~\eqref{eq:clifford-dv} by definition.
This automatically makes them valid symplectic transformations modulo \(c\), satisfying \((M \text{ mod }c)\in \text{Sp}(2p,\Z_c)\). 
In fact, the reverse also holds~\cite[Thm. VII.21]{newman1972integral}: every qudit Clifford transformation $M$ can be promoted to an integer symplectic transformation $N$ with $M = N(\bmod ~c)$.

Augmenting a given LCA code with \(V = M\oplus M\) yields a new code with encoders \(MT_{\cv}\) and \(MT_{\dv}\) for any qudit dimension \(c\).
Such matrices include those of block form
\begin{equation}
M=\begin{pmatrix}A & 0\\
0 & A^{-\intercal}
\end{pmatrix}\quad\text{or}\quad M=\begin{pmatrix}\dia 1_{p} & B\\
0 & \dia 1_{p}
\end{pmatrix}\,,
\end{equation}
where \(A\) is unimodular and \(B\) is symmetric.

Interesting integer symplectic matrices include the symplectic generator matrices of the \(E_8\) and Leech lattices~\cite[Appx. 2]{buser1994period}.
For example, 
\begin{equation}
M_{E_8} = \begin{pmatrix}
2 & \phantom{-}1 & \phantom{-}0 & \phantom{-}1 & \phantom{-}1 & \phantom{-}0 & \phantom{-}0 & \phantom{-}0 \\
1 & \phantom{-}2 & \phantom{-}1 & \phantom{-}0 & \phantom{-}0 & \phantom{-}1 & \phantom{-}0 & \phantom{-}0 \\
 0 & \phantom{-}1 & \phantom{-}2 & -1 & \phantom{-}0 & \phantom{-}0 & \phantom{-}1 & \phantom{-}0 \\
 1 & \phantom{-}0 & -1 & \phantom{-}2 & \phantom{-}0 & \phantom{-}0 & \phantom{-}0 & \phantom{-}1 \\
 1 & \phantom{-}0 & \phantom{-}0 & \phantom{-}0 & \phantom{-}2 & -1 & \phantom{-}0 & -1 \\
 0 & \phantom{-}1 & \phantom{-}0 & \phantom{-}0 & -1 & \phantom{-}2 & -1 & \phantom{-}0 \\
0 & \phantom{-}0 & \phantom{-}1 & \phantom{-}0 & \phantom{-}0 & -1 & \phantom{-}2 & \phantom{-}1 \\
 0 & \phantom{-}0 & \phantom{-}0 & \phantom{-}1 & -1 & \phantom{-}0 & \phantom{-}1 & \phantom{-}2
\end{pmatrix}~,
\end{equation}
represents both a Gaussian circuit and, when taken modulo \(c=2\), the qubit circuit from Fig.~\ref{fig3_circuit}.
These matrices are also unimodular and can therefore be used for the transformations \(Q,R\) in Eqs.~\eqref{eq:smith-theta-z} to yield other LCA codes.

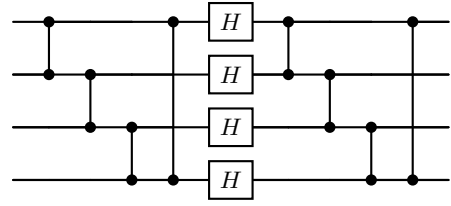
\begin{figure}
    \centering
\begin{quantikz}[row sep=0.2cm, column sep=0.4cm]
 & \ctrl{1} & \qw      & \qw      & \ctrl{3} & \gate{H} & \ctrl{1} & \qw      & \qw      & \ctrl{3} & \qw \\
 & \control{} \qw & \ctrl{1} & \qw      & \qw      & \gate{H} & \control{} \qw & \ctrl{1} & \qw      & \qw      & \qw \\
 & \qw      & \control{} \qw & \ctrl{1} & \qw      & \gate{H} & \qw      & \control{} \qw & \ctrl{1} & \qw      & \qw \\
 & \qw      & \qw      & \control{} \qw & \control{} \qw & \gate{H} & \qw      & \qw      & \control{} \qw & \control{} \qw & \qw
\end{quantikz}
    \caption{\small\label{fig3_circuit}
    Four-qubit circuit obtained from interpreting the \(E_8\) symplectic generator matrix modulo 2 as a binary symplectic matrix. 
    This circuit is invariant under cyclic permutations and maps \(\hat{\sigma}_{\mathsf{x}}^{j}\to \hat{\sigma}_{\mathsf{x}}^{j-1}\hat{\sigma}_{\mathsf{z}}^{j}\hat{\sigma}_{\mathsf{x}}^{j+1}\) and \(\hat{\sigma}_{\mathsf{z}}^{j}\to \hat{\sigma}_{\mathsf{z}}^{j-1}\hat{\sigma}_{\mathsf{x}}^{j}\hat{\sigma}_{\mathsf{z}}^{j+1}\), where superscripts are evaluated modulo 4.
    The two sets of CZ gates are examples of circuits creating cluster states~\cite{son2011quantum,son2012topological}.
    }
\end{figure}

\prg{Linear binary codes}

We construct an oscillator-qubit LCA code whose qubit encoder $T_{\dv}$ is a generator matrix of a linear binary code.

Let $G$ be a generator matrix of dimension $2k \times 2p$ of a linear binary code, where $k < p$~\cite{macwilliams1977theory}.
Assume that the rows of $G$ are linearly independent over $\mathbb{Z}_2$.
Hence, they remain linearly independent when we think of them as integer vectors.
We then pad the remaining \(2(p-k)\) rows with more linearly independent vectors to yield a unimodular matrix \(R\).

We fix the qudit dimensions to be equal, $c_j = 2$,
and set $\mathbf{d} = \dia{1}_k$. 
Plugging in our $R$ into Eq.~\eqref{eq:encoder-general},
we obtain the required $T_{\dv} = G$.
This corresponds to
\begin{equation}
Z = R^{\intercal} \begin{pmatrix}
0 & \dia{1}_k/2 & 0 \\
-\dia{1}_k/2 & 0 & 0 \\
0 & 0 & 0
\end{pmatrix} R.
\end{equation}

Such a \(Z\) is compatible with any \(\varTheta\) as long as $\varTheta - Z$ is invertible.
One simple choice of such a $\varTheta$ is:
\begin{equation}
\varTheta = R^{\intercal} \begin{pmatrix}
0 & L_1 & 0 & 0 \\
-L_1 & 0 & 0 & 0 \\
0 & 0 & 0 & L_2 \\
0 & 0 & -L_2 & 0
\end{pmatrix} R,
\end{equation}
where $L_1$ is the block corresponding to $\dia{1}_k/2$ in $Z$, and $L_1$, $L_2$ are diagonal unimodular matrices of dimension \(k\) and \(p-k\), respectively.

\section{Conclusion \& Future work}

We develop the theory of oscillator-qudit stabilizer states and error-correcting codes.
By ``stabilizer'', we mean that the states or codes are joint \(+1\)-eigenvalue eigenspaces of a commuting group of oscillator displacements and qudit Pauli strings.

Our framework is based on mathematical embeddings of integral symplectic lattices into locally compact Abelian (LCA) groups, which define the hybrid configuration space of oscillator-qudit systems.
As such, we refer to our constructions as \textit{LCA states} and \textit{LCA codes}.

LCA states are entangled with respect to the oscillator-qudit partition.
They are superpositions of tensor-product states, each of which consisting of a Gottesman-Kitaev-Preskill (GKP) quantum lattice state~\cite{gel1950expansion,zak1967finite,aharonov1969modular,10.1103/physreva.64.012310} and a qudit basis state.
They cannot be obtained from a tensor product of oscillator and qudit stabilizer states via Gaussian-Clifford transformations because such transformations are non-entangling.
This is in contrast to GKP and qudit stabilizer states, which can all be obtained from each other via Gaussian and Clifford operations, respectively.
Since an LCA state can be obtained from a qudit  stabilizer state by encoding a part of the state into a GKP code, the distinction between state classes highlights the usefulness of the GKP encoding map as a distinct resource (cf.~\cite{baragiola2019all}). 

\prg{LCA states}
LCA states are characterized by a rational anti-symmetric matrix \(Z\) that is called, in other contexts, a  commutation matrix~\cite{englbrecht2022transformations,gunderman2023transforming,sarkar2024qudit} (cf. the frustration graph matrix~\cite{makuta2025frustration,mann2025graph}).
The matrices \(\mp Z\) quantifies the non-commutativity among the parts of each LCA stabilizer generator that are supported on the oscillator and qudit factors, respectively.
Their opposite signs guarantee that the non-commutativity between the two parts cancels, ensuring that all stabilizers commute.

The shortest undetectable displacement of an oscillator-qudit LCA state is larger than that of the GKP state by a square-root of the physical qudit dimension.
As such, LCA states can be used to measure a larger set of non-commuting displacements than GKP states, making them useful for improving various metrological protocols~\cite{valahu2024quantum,duivenvoorden2017single,zhuang2020distributed,labarca2025quantum}.
\vva{
Taking the physical qudit dimension to infinity yields a phase space that is essentially classical. Since LCA states are oscillator-qudit entangled, this may be an interesting example of how a classical limit can arise from quantum entanglement~\cite{richens2017entanglement}.
}

\vva{
Generalizing GKP states, LCA states correspond to symplectic lattices in phase space and lie in eigenspaces of noncompact but \textit{discrete} stabilizer groups.
On the other hand, another class of oscillator stabilizer states --- analog/Gaussian stabilizer states~\cite{lloyd1998analog,braunstein1998error,barnes2004stabilizer,kwon2025most} --- correspond to hyperplanes and are stabilized by \textit{continuous} stabilizer groups, i.e., abelian Lie groups.
Studying stabilizer codes defined by continuous groups (cf.~\cite{jakobsen2016density}) is an interesting topic for future investigation, especially given the usefulness~\cite{zhang2006continuous,menicucci2006universal,gu2009quantum} and generically entangled nature~\cite{zhang2009quantum,kwon2025most} of their purely oscillator counterparts.
In particular, it could be useful to define hybrid CV-DV cluster states on oscillator-qudit systems.
}

The LCA framework combines real-space and binary/\(q\)-ary lattices into one hybrid structure.
Since both types of lattices yield computationally hard problems,
the LCA framework provides an intriguing playground for the development of lattice-based quantum algorithms and post-quantum encryption schemes~\cite{regev2009lattices,conrad2024good,kuperberg2025hidden}.
It will also be interesting to investigate connections of LCA states to Landau-Level physics~\cite{mumford1983tata,haldane1985periodic,10.1103/physreva.64.012310,dereli2021bloch,rymarz2021hardware,fan2023quantum}.

\prg{LCA codes}
We construct LCA codes using the theory of non-commutative tori~\cite{rieffel,schwarz1998morita,connes1998noncommutative,rieffel1999morita,seiberg1999string,yoneya2000string,li2004strong,elliott2008strong}, an important example in non-commutative geometry.
Integral symplectic lattices correspond to commutative tori, and our construction extends their embeddings into real space --- GKP codes --- to those into products of real space and discrete qudit factors.
The main reason for deploying the intricate mathematical theory is to determine the non-commutative torus of logical operators that is dual to the commutative torus defining an LCA code's stabilizers.

We develop two separate error-correcting strategies, one against pure displacement noise, and another against a smaller set of displacements and all qudit errors.
We also note that mixed strategies are possible, and it would be interesting to further develop them.

GKP codes can increase the error-correcting power of qubit encodings by concatenating each physical qubit into a GKP code~\cite{fukui2017analog,fukui2018high,vuillot2019toricGKP,noh2020fault}. 
GKP mode-into-mode encodings can even help with storing an infinite amount of quantum information~\cite{noh2020encoding,wu2023optimal}.
We anticipate that our various outlined LCA concatenation strategies should further boost  code performance and error-correcting thresholds.

\prg{Non-commutative geometry}
\vva{Our relation between stabilizer codes and non-commutative tori is likely to bear more fruit.}
The non-commutative torus framework is mathematically mature, but more work needs to be done to draw it closer to applications.
While it yields an exceedingly general class of codes, \vva{it is not clear whether this is the only way to construct hybrid codes,} or when two given general LCA codes are equivalent via symplectic operations.

Logical gates on LCA codes correspond to automorphisms of non-commutative tori~\cite{elpw2010structure, jeong2015finite, chakluef2019metaplectic, chaktracing2023tracing, chaksymmet2025symmetrized}, and studying their structure and degree of fault tolerance is of primary interest (cf.~\cite{10.1103/physreva.64.012310,10.48550/arxiv.2407.03270,fiber-bundle-ft} for the GKP case).
\vva{
In particular, fault-tolerant gates on GKP codes coupled to qudits have been studied in the contemporary Ref.~\cite{brenner2025trading}.
We believe this work is effectively studying finite-energy LCA codes since their initial states are GKP states, and since conditional oscillator-qudit displacement is part of their gate set.
It will be interesting to relate the two approaches in the future.
}

The generalized Möbius transformation for obtaining the logical operators may be related to similar transformations in other contexts~\cite[Eq.~(3)]{van2004graphical}\cite[Eq.~(2.24)]{menicucci2011graphical}\cite[Thm.~15]{leverrier2018p}.

General embeddings of non-commutative tori can be done into LCA groups that contain integer factors \(\Z\), which define linear rotors~\cite{albert2017general}.
But when the torus is commutative, the rotor factors decouple from the oscillator and qudit factors.
While the LCA construction does not extend to rotor-oscillator  and rotor-qudit codes, one can instead use rotor-oscillator entangling symplectic operations (generated by \(\hat x \otimes \hat L\) and \(\hat p \otimes \hat L\) for rotor momentum \(\hat L\)) and qudit-oscillator entangling symplectic operations (generated by \(\hat m \otimes \hat L\) and \(\hat \jmath \otimes \hat L\) in the notation of our paper) to make entangled states and codes.
We leave development of such stabilizer encodings as a topic for future investigation.

\vspace{.2cm}

\begin{acknowledgments}
V.V.A.\@ acknowledges 
Anthony J.\@ Brady,
\vva{Ansgar Burchards}, 
Jonathan Conrad,
Yvonne Y.\@ Gao,
Shawn Geller,
Yi-Kai Liu,
Carl Miller,
and
Mohammad Nobakht
for useful comments and discussions.
S.C.\@ acknowledges funding from the Anusandhan National Research Foundation (ANRF), Government of India, under the 
Prime Minister Early Career Research Grant (File no. ANRF/ECRG/2024/006557/PMS). 
V.V.A.\@ acknowledges NSF grants OMA2120757 (QLCI) and CCF2104489.
\vva{
Certain equipment, instruments, software, or materials are identified in this paper in order to specify the experimental procedure adequately. Such identification is not intended to imply recommendation or endorsement of any product or service by NIST, nor is it intended to imply that the materials or equipment identified are necessarily the best available for the purpose.
}
V.V.A.\@ thanks Ryhor Kandratsenia for providing daycare support throughout this work.
\end{acknowledgments}

\bibliography{references}

\newpage
\appendix
\onecolumngrid

\pagebreak
\section*{Appendices}

\section{LCA groups}\label{sec:displacement}

Let $M = \mathbb{R}^p \times \Z_{\dia c}$, where $p \in \mathbb{Z}_{\geq 0}$, and $\Z_{\dia c}$ is a finite abelian group of the form
\begin{equation}
\Z_{\dia c} = \mathbb{Z}_{c_1} \times \cdots \times \mathbb{Z}_{c_k}, \quad \text{with } c_1, \ldots, c_k \in \mathbb{N}.
\end{equation}
Put $n = 2p$. We now aim to define \textit{displacement operators} acting on the Hilbert space $L^2(M)$.

Let $\widehat{M}$ denote the Pontryagin dual of $M$, and define the group
\begin{subequations}
\begin{align}
G&=M\times\widehat{M}\\&=\mathbb{R}^{p}\times\mathbb{R}^{p}\times\Z_{\dia c}\times\Z_{\dia c}
\end{align}
\end{subequations}
where in the last two steps we use the fact that $\widehat{M}$ is isomorphic to $M.$ 
However, we will retain the hat notation for clarity.

Each component pair in this product has an associated symplectic form. For the continuous (oscillator) part, we define the standard symplectic form:
\begin{equation}
J_{\cv} = \begin{pmatrix}
0 & \dia{1}_p \\
- \dia{1}_p & 0
\end{pmatrix}.
\end{equation}
For the discrete (qudit) part, we introduce the diagonal matrix of inverse qudit dimensions and its corresponding symplectic form:
\begin{equation}
\dia{c} = \operatorname{diag}\left({c_1}, \ldots, {c_k}\right), \quad
J_{\dv } = \begin{pmatrix}
0 & \dia{c}^{-1} \\
-\dia{c}^{-1}  & 0
\end{pmatrix}.
\end{equation}
The combined symplectic form on $G$ is then given by
\begin{equation}
\label{eq:symplectic-form}
J = \begin{pmatrix}
J_{\cv} & 0 \\
0 & J_{\dv }
\end{pmatrix},
\end{equation}
which is a square matrix of dimension $n + 2k$, defining a 2-form on the space
\begin{equation}
\label{eq:h}
H^* := \mathbb{R}^p \times \mathbb{R}^{*p} \times \mathbb{R}^k \times \mathbb{R}^{*k},
\end{equation}
where $\mathbb{R}^*$ denotes the dual vector space of $\mathbb{R}$.

Next, define $J'$ to be the matrix obtained from $J$ by replacing all negative entries with zero. Then
\begin{equation}
J = J' - (J')^{\intercal}.
\end{equation}
For any two elements $x, y \in G$, define the bicharacter functions:
\begin{equation}
\beta(x, y) = e^{2\pi i \langle x, J' y\rangle}, \quad
\rho(x, y) = e^{2\pi i \langle x, J y\rangle}.
\end{equation}
Here, the inner product $\langle x, J y \rangle$ is defined via a natural covering map
\begin{equation}
\mathbb{R}^p \times \mathbb{R}^{*p} \times \mathbb{Z}^k \times \mathbb{Z}^{k} \to G,
\end{equation}
allowing us to view $x$ and $y$ as column vectors in $\mathbb{R}^{n+2k}$. While the expression $J' y$ depends on the choice of representative in the covering space, the values of $\beta(x, y)$ and $\rho(x, y)$ are independent of that choice.

\section{Displacement Operators}

Let $m = (m_{\cv}, m_{\dv}) \in M = \mathbb{R}^p \times \Z_{\dia c}$ and $s = (s_{\cv}, s_{\dv}) \in \widehat{M} = \widehat{\mathbb{R}^p} \times \widehat{\Z_{\dia c}}$. We will define the displacement operator $D(m, s)$ acting on $L^2(M)$ as the tensor product
\begin{equation}\label{eq:gen_displacement}
    D(m, s) = D(m_{\cv}, s_{\cv}) \otimes D(m_{\dv}, s_{\dv}),
\end{equation}
where each component acts on the corresponding subspace $L^2(\mathbb{R}^p)$ and $L^2(\Z_{\dia c})$, respectively.

To define $D(m_{\cv}, s_{\cv})$ on $L^2(\mathbb{R}^p)$, we first define the translation and modulation operators:
\begin{equation}
T_{m_{\cv}} \xi(t) = \xi(t - \sqrt{2\pi}m_{\cv}), \quad
M_{s_{\cv}} \xi(t) = e^{\sqrt{2\pi} i s_{\cv} \cdot t} \, \xi(t),
\end{equation}
for $\xi \in L^2(\mathbb{R}^p), t \in \mathbb{R}^p$. These correspond to the usual momentum and position operators:
\begin{equation}
T_{m_{\cv}}=e^{-\sqrt{2\pi}im_{\cv}\cdot\hat{\mathbf{p}}},\quad M_{s_{\cv}}=e^{\sqrt{2\pi}is_{\cv}\cdot\hat{\mathbf{x}}}.
\end{equation}
Here, \(\hat{\mathbf{x}} = (\hat{x}_1,\hat{x}_2,\cdots,\hat{x}_p)\) is a vector of position operators, and the same for the vector momentum operator \(\hat{\mathbf{p}}\).
The displacement operator on $\mathbb{R}^p$ is then defined as
\begin{equation}
D(x_{\cv})=D(m_{\cv}, s_{\cv}) := M_{s_{\cv}} T_{m_{\cv}},
\end{equation} $x_{\cv} = (m_{\cv}, s_{\cv})\in \R^p \times \widehat{\R^p}.$

These operators satisfy the following commutation relation:
\begin{equation}
\label{eq:ccr_semi_displacement_general}
D(m_{\cv}, s_{\cv}) \, D(l_{\cv}, t_{\cv}) =
e^{-2\pi i \langle x_{\cv}, J_{\cv}' y_{\cv} \rangle} \,
D(m_{\cv} + l_{\cv}, s_{\cv} + t_{\cv}),
\end{equation}
where  
\begin{equation}x_{\cv} = (m_{\cv}, s_{\cv}),~ y_{\cv} = (l_{\cv}, t_{\cv})\in \R^p \times \widehat{\R^p}, ~
J_{\cv}' = \begin{pmatrix}
0 & \dia{1}_p \\
0 & 0
\end{pmatrix}.
\end{equation}
Moreover, the standard Heisenberg commutation relation holds:
\begin{equation}\label{eq:ccr_displacement_R^n}
D(x_{\cv}) D(y_{\cv}) = e^{-2\pi i \langle x_{\cv}, J_{\cv} y_{\cv} \rangle} D(y_{\cv}) D(x_{\cv}),
\quad \text{for all } x_{\cv}, y_{\cv} \in \mathbb{R}^p \times \widehat{\mathbb{R}^p}.
\end{equation}

Now let us define $D(m_{\dv}, s_{\dv})$, which are the usual qudit Weyl displacement operators. Define the operators $T_{m_{\dv}}$ and $M_{s_{\dv}}$ on $L^2(\Z_{\dia c})$ by
\begin{equation}
T_{m_{\dv}} \, \xi(t) = \xi\left(m^{-1} t\right), \quad M_{s_{\dv}} \, \xi(t) = \langle t , s_{\dv} \rangle \, \xi(t),
\end{equation}
for $\xi \in L^2(\Z_{\dia c})$ and $t \in \Z_{\dia c}$. The pairing is given by
\begin{equation}
\langle t , s_{\dv} \rangle = e^{2\pi i \mathlarger{\sum_{i=1}^k} \frac{s_i t_i}{c_i}}, \quad \text{where } s = (s_i)_{i=1}^k, \; t = (t_i)_{i=1}^k.
\end{equation}

Then, for $x_{\dv} = (m_{\dv}, s_{\dv}) \in \Z_{\dia c} \times \widehat{\Z_{\dia c}}$, the displacement operator is defined as
\begin{equation}
D(m_{\dv}, s_{\dv}) := M_{s_{\dv}} T_{m_{\dv}}.
\end{equation}
These operators satisfy the discrete Heisenberg commutation relation:
\begin{equation}
\label{eq:ccr_displacement_qdit}
D(x_{\dv}) \, D(y_{\dv}) = e^{-2\pi i \langle x_{\dv}, J_{\dv} y_{\dv} \rangle} \, D(y_{\dv}) \, D(x_{\dv}).
\end{equation}

Finally, if we define the general displacement operators acting on $L^2(\mathbb{R}^p \times \Z_{\dia c})$ as
\begin{equation}
D(x) := D(m_{\cv}, s_{\cv}) \otimes D(m_{\dv}, s_{\dv}), \quad \text{for } x = (m, s) \in M \times \widehat{M},
\end{equation}
then the full Heisenberg commutation relation reads
\begin{equation}
\label{eq:ccr_displacement_general}
D(x) \, D(y) = e^{-2\pi i \langle x, J y \rangle} \, D(y) \, D(x), \quad \text{for all } x, y \in M \times \widehat{M}.
\end{equation}

\begin{remark}
  Note that in the case $M = \mathbb{R}^p$, Eq.~\eqref{eq:ccr_semi_displacement_general} does not exactly match the convention commonly used in the physics or mathematics literature (e.g.,~\cite[Equation 6]{Conrad2022gottesmankitaev}). To align with the standard convention, one needs to modify the definition of the displacement operator $D(m_{\cv}, s_{\cv})$ by including a phase factor $e^{-\pi i \langle s_{\cv}, m_{\cv} \rangle}$. This modification does not affect Eq.~\eqref{eq:ccr_displacement_R^n} or Eq.~\eqref{eq:ccr_displacement_general}. In fact, we incorporate this modification by introducing phases to the displacement operators later in Definition~\ref{def:stab_group}.

\end{remark}
 
\section{Lattices and associated LCA codes}\label{app:lca-intro}

We call a subgroup $H$ of $G$ co-compact if the quotient $G / H$ is compact. A subgroup $\mathscr{D}$ of $G$ that is both discrete and co-compact is called a lattice in $G$. Fix a bicharacter $\rho$ in $G$.  
Define the 
dual lattice for a lattice $\mathscr{D}$ as
\begin{equation}
\mathscr{D}^{\perp}=\{z \in G: \rho(z, y)=1, ~ \forall y \in \mathscr{D}\}.
\end{equation}
From now on, we only consider lattices inside $G=M \times \widehat{M},$  $M=\mathbb{R}^p  \times \Z_{\dia c}$, where $p \in \mathbb{Z}_{\geqq 0}$\footnote{\sch{$\mathbb{Z}_{\geqq 0}$ means all positive numbers including zero, i.e. $\N \cup \{0\}$.}}.  We have already defined a natural bicharacter $\rho$ on $G$.

Let $\mathscr{D}$ be a lattice inside $G$. Then the dual group for $G / \mathscr{D}$ will be discrete. This dual group can be identified~\cite[pg. 76]{rieffel} with $\mathscr{D}^{\perp}$. Thus $\mathscr{D}^{\perp}$ must be a discrete subgroup of $G$. In the same way, the dual of $G / \mathscr{D}^{\perp}$ can be identified with $\mathscr{D}$, which is discrete. It follows that $G / \mathscr{D}^{\perp}$ is compact. Thus $\mathscr{D}^{\perp}$ is also a lattice. 

We will now construct lattices inside $G$ using \textit{symplectic Gram matrices}. Let $n \geqq 2$ and $\mathscr{T}_n$ be the space of $n \times n$ real anti-symmetric matrices. We will call such a real anti-symmetric matrix a \textit{symplectic Gram matrix} also.

We follow~\cite[Def.~2.1]{li2004strong}, \cite[Def.~4.1]{rieffel} and related discussions therein for the following. Let $L = \R^n,$ and view $\Z^n$ as the standard lattice in $L^*.$ For a $\varTheta\in \mathscr{T}_n,$ view $\varTheta$ is in $\bigwedge^2(L).$ 
By an embedding map for $\varTheta \in \mathscr{T}_n$ we mean a linear map $T$ from $L^*$ to the dual space $H^*$~\eqref{eq:h} satisfying

\begin{enumerate}
   \item  $T\left(\mathbb{Z}^n\right) \subseteq \mathbb{R}^p \times \mathbb{R}^{* p}  \times \mathbb{Z}^k \times \mathbb{Z}^k$. Then we can think of $T\left(\mathbb{Z}^n\right)$ as in $G$ via composing $\left.T\right|_{\mathbb{Z}^n}$ with the natural covering map $\mathbb{R}^p \times \mathbb{R}^{* p}  \times \mathbb{Z}^k \times \mathbb{Z}^k \rightarrow G$.
   \item $T\left(\mathbb{Z}^n\right)$ is a lattice in $G$.
   \item The form $J$ on $H^*$ is pulled back by $T$ to the form $\varTheta$ on $L^*$, i.e. 
   \begin{equation}
   \label{eq:embedding-condition}
       T^{\intercal} J T= \varTheta~.
   \end{equation}
\end{enumerate}

Let $\mathscr{D}=T\left(\mathbb{Z}^n\right)$. Then by definition of $T$ we may think of $\mathscr{D}$ as a lattice in
$$G=\mathbb{R}^p \times \widehat{\mathbb{R}^{p}}  \times\left(\mathbb{Z}_{c_1} \times \cdots \times \mathbb{Z}_{c_k}\right) \times\left(\mathbb{Z}_{c_1} \times \cdots \times \mathbb{Z}_{c_k}\right).
$$ 

Given such a map $T$ viewed as $T\left(\mathbb{Z}^n\right) \subseteq \mathbb{R}^p \times \mathbb{R}^{* p}  \times \mathbb{Z}^k \times \mathbb{Z}^k,$ we define the phase map
\begin{equation}
    A(l):=\exp\left(2\pi i \langle-T(l),J^{\prime}T(l)/2\rangle\right)\quad\text{for}\quad l\in\mathbb{Z}^{n}\,.
\end{equation}

We are now ready to define LCA stabilizer group. For a given symplectic Gram matrix $\varTheta$, an embedding map $T$ for $\varTheta,$ and the standard generators $\{e_1, e_2, \cdots, e_n\}$ of $\mathbb{Z}^n,$ denote $T(e_i)$ by $\mathbf{e}_i$. If $\varTheta$ is an integer matrix, then using the condition (3) of the definition of an embedding map, from Equation~\eqref{eq:ccr_displacement_general} it readily follows that the displacement operators $\{D(\mathbf{e}_i)\}_{i=1}^n$ commute with each other inside the set of bounded operators on $L^2(M).$ This motivates the following definition.

\begin{definition}\label{def:stab_group}
    (LCA stabilizer group). Given an embedding map $T$ for an integer symplectic Gram matrix $\varTheta,$ the stabilizer group of an LCA code is given by a set of displacements (with phases)
\begin{equation}
\label{eq:stabgroup}
    \mathcal{S}:=\left\langle A(e_{1})D\left(\mathbf{e}_{1}\right),\ldots,A(e_{n})D\left(\mathbf{e}_{n}\right)\right\rangle \quad\quad\text{with}\quad\quad\mathbf{e}_{j}=T(e_{j})\,.
\end{equation}

\end{definition}

It will be convenient for us to denote $A(e_i)D\left(\mathbf{e}_i\right)$ by $U_i.$ In general, for $l\in \Z^n,$ denote $A(l)D\left(T(l)\right)$ by $U_l.$ 
The generators commute due to the following,
\begin{subequations}
    \begin{align}
     U_{i}U_{j}&=e^{-2\pi i\langle\mathbf{e}_{i},J\mathbf{e}_{j}\rangle}U_{j}U_{i}&\text{by Eq.~\eqref{eq:ccr_displacement_general}}\\&=e^{-2\pi i\langle T(e_{i}),JT(e_{j})\rangle}U_{j}U_{i}&\text{by Eq.~\eqref{eq:stabgroup}}\\&=e^{-2\pi i\langle e_{i},T^{\intercal}JT(e_{j})\rangle}U_{j}U_{i}&\\&=e^{-2\pi i\langle e_{i},\varTheta(e_{j})\rangle}U_{j}U_{i}&\text{by Eq.~\eqref{eq:embedding-condition}}\\&=e^{-2\pi i\varTheta_{ij}}U_{j}U_{i}&\\&=U_{j}U_{i}&\varTheta\text{\ is an integer matrix.}
    \end{align}
\end{subequations}

Given an embedding map $T$ associated to an integer symplectic Gram matrix $\varTheta=(\varTheta_{ij})$, the embdedding $\mathscr{D}=T(\Z^n)$ is a lattice in $G.$ We call such a lattice an \textit{LCA lattice}. 
The stabilizer code (i.e., the joint \(+1\)-eigenspace) corresponding to the given LCA stabilizer group is called the \textit{LCA code}.  

Before constructing a general class of embedding maps, and hence LCA codes, in the next section, we first note that the well-known GKP stabilizer groups (see~\cite{10.1103/physreva.64.012310,Conrad2022gottesmankitaev,10.1103/prxquantum.3.010335}) and their associated codes are examples of the general construction described above. These well-known examples include the scaled GKP groups which are recalled below.

\textbf{Scaled GKP groups}. Consider \( M = \mathbb{R}^p \). Recall that \( J = \begin{pmatrix} 0 & \dia{1}_p \\ -\dia{1}_p & 0 \end{pmatrix} \). A scaled GKP stabilizer group is obtained from the symplectic Gram matrix 
\begin{equation} 
\varTheta = \begin{pmatrix} 0 & \lambda \\ -\lambda & 0 \end{pmatrix} \otimes \dia{1}_p = \lambda J, \quad \text{where}~ \lambda \in \mathbb{N},
\end{equation} 
by taking an embedding map \( T \) such that \( T^{\intercal} J T = \varTheta \). Recall that a symplectic matrix is a matrix \( T \) such that \( T^{\intercal} J T = J \). The set of symplectic matrices is denoted by \( \text{Sp}(2n, \mathbb{R}) \). It is then clear that, if \( T_0 \) is a symplectic matrix, then we can set
$T = \sqrt{\lambda} T_0$ for any natural number \(\lambda\).
Such a matrix \( T \) satisfies the condition \( T^{\intercal} J T = \varTheta \).

A particularly important example of such scaled GKP stabilizer groups is the \textit{square lattice GKP code}, introduced in Ref.~\cite{10.1103/physreva.64.012310}, which is given by
\begin{equation} 
\varTheta_{\square} = \begin{pmatrix} 0 & 2 \\ -2 & 0 \end{pmatrix}, \quad T_{\square} = \sqrt{2} \dia{1}_2.
\end{equation} 
One could also consider
\begin{equation} 
\varTheta_{\hexago} = \begin{pmatrix} 0 & 2 \\ -2 & 0 \end{pmatrix}, \quad T_{\hexago} = \frac{1}{3^{1/4}} \begin{pmatrix} 2 & 1 \\ 0 & \sqrt{3} \end{pmatrix},
\end{equation} 
which is also known as the \textit{hexagonal GKP code}.

\section{The group $\mathrm{SO}(n, n \mid \mathbb{Z})$ and the corresponding lattices}
\label{app:general-lca}

Let $\mathrm{O}(n, n \mid \mathbb{R})$ be the group of linear transformations of the space $\mathbb{R}^{2 n}$ preserving the quadratic form $x_1 x_{n+1}+x_2 x_{n+2}+\cdots+x_n x_{2 n}$, and $\mathrm{SO}(n, n \mid \mathbb{Z})$ be the subgroup of $\mathrm{O}(n, n \mid \mathbb{R})$ consisting of matrices with integer entries and determinant 1. We can write an element $g$ of $\mathrm{O}(n, n \mid \mathbb{R})$ as
\begin{equation}
g=\left(\begin{array}{ll}
A & B \\
C & D
\end{array}\right),
\end{equation}
where $A, B, C, D$ are $n \times n$ matrices satisfying
\begin{equation}
A^{\intercal} C+C^{\intercal} A=0=B^{\intercal} D+D^{\intercal} B, \quad A^{\intercal} D+C^{\intercal} B=\dia{1}_n .
\end{equation}

We now define a (partial) action of $\mathrm{SO}(n, n \mid \mathbb{Z})$ on the space of symplectic Gram matrices $\mathscr{T}_n$ as
\begin{equation}
g \cdot \varTheta:=(A \varTheta+B)(C \varTheta+D)^{-1}
\end{equation}
whenever $C \varTheta+D$ is invertible\sch{over the real numbers}. It is known that for each $g \in \mathrm{SO}(n, n \mid \mathbb{Z}),$ this action is defined on a dense open subset of $\mathscr{T}_n$~\cite{rieffel1999morita,li2004strong}.

A particular \(g\) can be used to obtained the dual lattice \(-g\cdot \varTheta\) from a given LCA code defined by \(\varTheta\).
The columns of the Gram matrix \(S\) of that lattice can then be used to define logical Pauli operators.

\begin{theorem}[Theorem~\ref{thm:main_text}, detailed restatement] \label{thm:main_dual}
Let $Z$ be matrix in $\mathscr{T}_n$ such that its entries are all rational numbers. Then there exists a $g \in \mathrm{SO}(n, n \mid \mathbb{Z})$ and a finite abelian group $\Z_{\dia c},$ such that for any $\varTheta \in \mathscr{T}_n$ with $\varTheta-Z$ invertible there are embedding maps $T, S: L^* \rightarrow H^*,$ $H^*=\mathbb{R}^p \times \mathbb{R}^{* p}  \times \mathbb{R}^k \times \mathbb{R}^{* k}$ with $M=\mathbb{R}^p \times \Z_{\dia c},$ for the symplectic Gram matrices $\varTheta$ and $-g\cdot\varTheta$ respectively satisfying $S\left(\mathbb{Z}^n\right)=\left(T\left(\mathbb{Z}^n\right)\right)^{\perp}.$
\end{theorem}

\begin{proof}
We refer to the proof of Proposition 4.1 in~\cite{li2004strong}. However, we write down the explicit description of $T, S, g,$ and $g \cdot \varTheta.$
Since $Z$ is rational, there is some $m \in \mathbb{N}$ such that $m Z$ is integral. 
Thinking of $m Z$ as a bilinear alternating form on $\mathbb{Z}^n$, 
we can bring it to alternating Smith normal form.
Namely, we can find an $R \in \mathrm{GL}(n, \mathbb{Z})$, some integer $1 \leq k \leq p$ and unique positive integers $h_1, \ldots, h_k \geq 0$ satisfying \(h_1 | h_2 | \cdots | h_k\) such that
\begin{equation}
m Z=R^{\intercal}\left(\begin{array}{ccc}
0 & \dia{h} & 0 \\
-\dia{h} & 0 & 0 \\
0 & 0 & 0
\end{array}\right) R~,
\end{equation}
where $\dia{h} =\operatorname{diag}\left(h_1, \ldots, h_k\right)$. Let $d_j / c_j=h_j / m$ with $\textnormal{gcd}\left(d_j, c_j\right)=1$ and $c_j>0$ for each $1 \leq j \leq k$. Set $W=\mathbb{Z}_{c_1} \times \cdots \times \mathbb{Z}_{c_k}$.

Let $\varTheta \in \mathscr{T}_n$ with $\varTheta-Z$ is invertible.
So we can find a $T_{\cv} \in \mathrm{GL}(n, \mathbb{R})$ such that $T_{\cv}^{\intercal} J_{\cv} T_{\cv}=\varTheta-Z$ using the alternating Smith normal form of \(\varTheta-Z\). 
Let $\dia{d}=\operatorname{diag}\left(d_1, \ldots, d_k\right),$ $\dia{c}=\operatorname{diag}\left(c_1, \ldots, c_k\right),$ and let
\begin{equation}
T_{\dv }=\left(\begin{array}{ccc}
\dia{d} & 0 & 0 \\
0 & \dia{1}_k & 0
\end{array}\right)
R , \quad T=\binom{T_{\cv}}{T_{\dv }}
~.\end{equation}
 Then, $T_{\cv}, T_{\dv }$ and $T$ have sizes $n \times n, 2 k \times n$ and $(n+2 k) \times n$, respectively. 
One readily checks $T^{\intercal} J T=\varTheta$, where $J$ is defined in Eq.~\eqref{eq:symplectic-form}. Notice that as a linear map from $L^*=\mathbb{R}^{* n}$ to $H^*=\mathbb{R}^p \times \mathbb{R}^{* p}  \times \mathbb{R}^k \times \mathbb{R}^{* k}$, $T$ carries the lattice $\mathbb{Z}^n=\mathbb{Z}^{2 p}$ into $\mathbb{R}^p \times \mathbb{R}^{* p} \times \mathbb{Z}^k \times \mathbb{Z}^k$. Thus the conditions of an embedding map are satisfied and hence $T$ is an embedding map for $\varTheta$. 

For each $1 \leq j \leq k$ since $\textnormal{gcd}\left(d_j, c_j\right)=1$ we can find $a_j, b_j \in \mathbb{Z}$ such that $a_j d_j+b_j c_j=1$. Let
\begin{equation}
\dia{b}=\operatorname{diag}\left(b_1, \ldots, b_k\right), \quad \dia{a} =\operatorname{diag}\left(a_1, \ldots, a_k\right) .
\end{equation}
Now $S$ is given by
\begin{equation}
S=\left(\begin{array}{cc}
S_{\cv}  \\
\left(\begin{array}{ccc}
0 & -\dia{1}_k & 0 \\
\dia{a} & 0 & 0
\end{array}\right) 
\end{array}\right),
\end{equation}
where $S_{\cv}$ has size $n \times n$ and is given by:
\begin{equation}
S_{\cv}=J_{\cv}\left(T_{\cv}^{\intercal}\right)^{-1} R^{\intercal}\left(\begin{array}{ccc}\dia{c}^{-1} & 0 & 0 \\ 0 & \dia{c}^{-1} & 0 \\ 0 & 0 & -\dia{1}_{2 p-2 k}\end{array}\right).
\end{equation}

Then $S$ is an embedding map for
\begin{equation}
-\varTheta^{\perp}=S^{\intercal} J S
\end{equation}
where
\begin{equation}
\begin{aligned}
\varTheta^{\perp}= & \left(\begin{array}{ccc}
\dia{c}^{-1} & 0 & 0 \\
0 & \dia{c}^{-1} & 0 \\
0 & 0 & -\dia{1}_{2 p-2 k}
\end{array}\right) R (\varTheta-Z)^{-1} R^{\intercal}\left(\begin{array}{ccc}
\dia{c}^{-1} & 0 & 0 \\
0 & \dia{c}^{-1} & 0 \\
0 & 0 & -\dia{1}_{2 p-2 k}
\end{array}\right) 
 +\left(\begin{array}{ccc}
0 & -\dia{a} \dia{c}^{-1} & 0 \\
\dia{a} \dia{c}^{-1} & 0 & 0 \\
0 & 0 & 0
\end{array}\right).
\end{aligned}
\end{equation}

Now we have to find $g=\left(\begin{array}{ll}A^{\prime} & B^{\prime} \\ C^{\prime} & D^{\prime}\end{array}\right) \in \mathrm{SO}(n, n \mid \mathbb{Z})$ such that $\varTheta^{\perp}=g\cdot \varTheta$. We have
\begin{equation}
\begin{aligned}
& A^{\prime}=
\left(\begin{array}{ccc}
0 & -\dia{a} & 0 \\
\dia{a} & 0 & 0 \\
0 & 0 & 0
\end{array}\right)R^{-\intercal}, \quad B^{\prime}=
\left(\begin{array}{ccc}
\dia{b} & 0 & 0 \\
0 & \dia{b} & 0 \\
0 & 0 & -\dia{1}_{2 p-2 k}
\end{array}\right)R
\end{aligned}
\end{equation}

\begin{equation}
C^{\prime}=\left(\begin{array}{ccc}
\dia c & 0 & 0\\
0 & \dia c & 0\\
0 & 0 & -\dia 1_{2p-2k}
\end{array}\right)R^{-\intercal},\quad D^{\prime}=\left(\begin{array}{ccc}
0 & -\dia d & 0\\
\dia d & 0 & 0\\
0 & 0 & 0
\end{array}\right)R~.
\end{equation}
One can verify the above matrices satisfy \(g\cdot \varTheta = \varTheta^{\perp}\) by explicit calculation using the various relations between \(\dia{b}\), \(\dia{a}\), \(\dia{h}\), \(\dia{c}^{-1}\), and \(\dia{d}\). 
\end{proof}

Using the above theorem, we have constructed a large class of LCA codes. From now on, an LCA code will refer exclusively to the codes constructed above.

\subsection{Qudit symplectic transformations}
\label{app:general-lca-qudit-clifford}

The construction of codes in Theorem~\ref{thm:main_dual} can be generalized straightforwardly.  
Let \( L \in \mathrm{GL}(2k, \mathbb{Z}) \) be an integer symplectic matrix, meaning that it satisfies \( L^{\intercal}  J_{\dv} L = J_{\dv} \).  
Such a matrix can arise from symplectic transformations on the qudit side (see~\cite[Theorem 1]{Newman1964}), and any such matrix is automatically a modular symplectic matrix since satisfaction of the integer symplectic condition guarantees satisfaction of the modular one.
We then modify the definition of \( T \) in Theorem~\ref{thm:main_dual} as follows:
\begin{equation}
T^{L}_{\dv} = 
L \begin{pmatrix}
\dia{d} & 0 & 0 \\
0 & \dia{1}_k & 0
\end{pmatrix}, \qquad
R = L T_{\dv}, \qquad
T^{L} = 
\begin{pmatrix}
T_{\cv} \\
T^{L}_{\dv}
\end{pmatrix}.
\end{equation}
A direct calculation shows that \( (T^{L})^{\intercal}  J T^{L} = \varTheta \); hence \( T^{L} \) is an embedding of \( \varTheta \).

If we further assume
\begin{equation}
L \begin{pmatrix}
\dia{d} & 0 & 0 \\
0 & \dia{1}_k & 0
\end{pmatrix}
=
\begin{pmatrix}
\dia{d} & 0 & 0 \\
0 & \dia{1}_k & 0
\end{pmatrix} \tilde{L},
\end{equation}
for some \( \tilde{L} \in \mathrm{GL}(n, \mathbb{Z}) \), the quantities \( \varTheta^{\perp}, g \), and \( S \) can be obtained exactly as in Theorem~\ref{thm:main_dual} by absorbing \( \tilde{L} \) into \( R \); that is, replace \( R \) with \( \tilde{L}R \) in the arguments of Theorem~\ref{thm:main_dual}. Note that this assumption is not required to define the code; rather, it ensures a convenient description of the dual lattice as stated in Theorem~\ref{thm:main_dual}. The condition holds in many cases, for example, when \( \dia{d} = \dia{1}_k \).

\section{Logical dimension calculation}
\label{app:logical-dim}

\sch{Let us take $\Theta$ to be an integer matrix in Theorem~\ref{thm:main_dual}. Denote $\mathscr{D}=T(\Z^n)$ and $\mathscr{D}^\perp=S(\Z^n)$.} If $K$ denotes the dimension of the logical Hilbert space, then $K^2= |\mathscr{D}^\perp/\mathscr{D}|,$ since this is the number of distinct of logical operations.
We can now easily compute the above number using the formula $|\mathscr{D}^\perp/\mathscr{D}|=\frac{|G/\mathscr{D}|}{|G/\mathscr{D}^\perp|}$. But in the above situation, using Theorem 3.4 along with Proposition 4.3 of~\cite{rieffel}, $|G/\mathscr{D}|$ is $|\det(T_{\cv})|$  and $|G/\mathscr{D}^\perp|$ is  $|\det(S_{\cv})|.$ Hence $K^2= |\det(T_{\cv})|/|\det(S_{\cv})|.$

We now aim to justify why the numbers
\begin{equation}
K^2 = \frac{|\det(T_{\cv})|}{|\det(S_{\cv})|},
\end{equation}
and $K$ are integers. Recall that
\begin{equation}
S_{\cv} = J_{\cv}(T_{\cv}^{\intercal})^{-1} R^{\intercal} 
\begin{pmatrix}
\dia{c}^{-1} & 0 & 0 \\
0 & \dia{c}^{-1} & 0 \\
0 & 0 & -\dia{1}_{2p-2k}
\end{pmatrix},
\quad 
T_{\cv}^{\intercal} J_{\cv} T_{\cv} = \varTheta - Z,
\quad 
\dia{c}^{-1} = \mathrm{diag}\left( \frac{1}{c_1}, \ldots, \frac{1}{c_k} \right).
\end{equation}
Then,
\begin{equation}
\det(S_{\cv}) 
= \det\left( J_{\cv}(T_{\cv}^{\intercal})^{-1} R^{\intercal} 
\begin{pmatrix}
\dia{c}^{-1} & 0 & 0 \\
0 & \dia{c}^{-1} & 0 \\
0 & 0 & -\dia{1}_{2p - 2k}
\end{pmatrix} \right)
= \det(T_{\cv})^{-1} \det(R) \det(\dia{c}^{-1})^2.
\end{equation}
Hence,
\begin{equation}
\frac{\det(T_{\cv})}{\det(S_{\cv})} = \frac{\det(T_{\cv})^2}{\det(R) \det(\dia{c}^{-1})^2} 
= \frac{\det(\varTheta - Z)}{\det(R)\det(\dia{c}^{-1})^2}.
\end{equation}

Now consider the matrix
\begin{equation}
Z = R^{\intercal} 
\begin{pmatrix}
0 & \dia{d/c} & 0 \\
-\dia{d/c} & 0 & 0 \\
0 & 0 & 0
\end{pmatrix} R,
\quad 
\dia{d/c} = \mathrm{diag}\left( \frac{d_1}{c_1}, \ldots, \frac{d_k}{c_k} \right).
\end{equation}
We compute
\begin{equation}
\det(\varTheta - Z) 
= \det\left( \varTheta - R^{\intercal} 
\begin{pmatrix}
0 & \dia{d/c} & 0 \\
-\dia{d/c} & 0 & 0 \\
0 & 0 & 0
\end{pmatrix} R \right)
= \det\left( R^{-\intercal} \varTheta R^{-1} - 
\begin{pmatrix}
0 & \dia{d/c} & 0 \\
-\dia{d/c} & 0 & 0 \\
0 & 0 & 0
\end{pmatrix} \right) \cdot \det(R)^2.
\end{equation}

Let us denote
\begin{equation}
\widetilde{\varTheta} = R^{-\intercal} \varTheta R^{-1}, \quad 
\tilde{Z} = 
\begin{pmatrix}
0 & \dia{d/c} & 0 \\
-\dia{d/c} & 0 & 0 \\
0 & 0 & 0
\end{pmatrix}.
\end{equation}
Since $\det(R) = \pm 1$, and we are ultimately interested in the absolute value of $\det(T_{\cv}) / \det(S_{\cv})$, we may disregard $\det(R)$ and focus on the quantity
\begin{equation}
\frac{\det(\widetilde{\varTheta} - \tilde{Z})}{\det(\dia{c}^{-1})^2}.
\end{equation}
Noting that $\det(\dia{c}^{-1})^2 = \left( \frac{1}{c_1 c_2 \cdots c_k} \right)^2$, we have
\begin{equation}
\frac{\det(\widetilde{\varTheta} - \tilde{Z})}{\det(\dia{c}^{-1})^2}
= \left( c_1 c_2 \cdots c_k \cdot \mathrm{Pf}(\widetilde{\varTheta} - \tilde{Z}) \right)^2,
\end{equation}
where $\mathrm{Pf}(A)$ denotes the Pfaffian of the antisymmetric matrix $A$. Recall that the Pfaffian of a $2p \times 2p$ antisymmetric matrix $A = (a_{ij})$ is given by
\begin{equation}
\operatorname{Pf}(A) = \sum_{\pi \in S_{2p}} (-1)^{|\pi|} \prod_{s=1}^{p} a_{\pi(2s-1)\, \pi(2s)},
\end{equation}
where the sum runs over all permutations $\pi \in S_{2p}$ satisfying $\pi(2s-1) < \pi(2s)$ for all $1 \le s \le p$ and $\pi(1) < \pi(3) < \cdots < \pi(2p - 1)$; as usual, $|\pi|$ denotes the signature of $\pi$.

Let the entries of $\widetilde{\varTheta} - \tilde{Z}$ be $a_{ij}$. All $a_{ij}$ are integers except when $j = i + p$, in which case $a_{ij}$ is a rational number with denominator $c_i$. Consequently, each monomial appearing in $\operatorname{Pf}(\widetilde{\varTheta} - \tilde{Z})$ may have at most one factor with denominator $n_i$ for each $i$. Hence, every term in the Pfaffian has denominator of the form $c_1^{a_1} c_2^{a_2} \cdots c_k^{a_k}$, where each exponent $a_i \in \{0, 1\}$. Therefore, multiplying $\operatorname{Pf}(\widetilde{\varTheta} - \tilde{Z})$ by $c_1 c_2 \cdots c_k$ clears all denominators and yields an integer.

\section{Examples}

The construction above generalizes the\sch{usual construction of GKP-like codes, which are constructed as symplectic lattices inside $\R^{2p}.$} Here in this section, we want to sketch some special cases of the above construction that are themselves interesting.

\subsection{Lattices inside $\R^p \times \widehat{\R^p}$}\label{sec:rp}
This is the case of the usual GKP codes. In this case, $n = 2p$, and $Z=0.$ 
Let $\varTheta \in \mathcal{T}_n$ be such that $\det(\varTheta) \neq 0$. This means that $\varTheta -Z$ is invertible.  It follows from the proof of Theorem~\ref{thm:main_dual} that $g= \begin{pmatrix}
0 & \dia{1}_{2p} \\
\dia{1}_{2p} & 0
\end{pmatrix} \in \mathrm{SO}(n, n \mid \mathbb{Z}).
$ Hence
\begin{equation} 
\varTheta^{\perp} = g \cdot \varTheta = \varTheta^{-1},
\end{equation} 
and one obtains a GKP lattice within the group $\R^p \times \widehat{\R^p}$, assuming $\varTheta$ is an integer matrix. The logical dimension is given by the Pfaffian of the matrix $\varTheta$. We now describe how to obtain the logical Pauli $\hat{X}$ and $\hat{Z}$ operators. 

For any such integer Gram matrix $\varTheta$, there exists an invertible matrix $R \in \mathrm{GL}(2p, \mathbb{Z})$ such that 
\begin{equation} \label{eq:standard_R}
R^{\intercal} \left( \begin{pmatrix}
0 & 1 \\
-1 & 0
\end{pmatrix} \otimes \dia{d} \right) R = \varTheta,
\end{equation} 
where $\dia{d}$ is a diagonal matrix $\mathrm{diag}(d_1, d_2, \ldots, d_p)$ in alternating Smith normal form, satisfying $d_1 \mid d_2 \mid \cdots \mid d_p$ . 
Furthermore, the lattice is uniquely determined by $\dia{d}$; that is, if two diagonal matrices $\dia{d}_1$ and $\dia{d}_2$ both satisfy this condition, then $\dia{d}_1 = \dia{d}_2$. The matrix
$
\begin{pmatrix}
0 & 1 \\
-1 & 0
\end{pmatrix} \otimes \dia{d}
$ is called the \textit{standard form} of $\varTheta$ in the context of GKP codes.

Start with a Gram symplectic matrix in standard form:
$
\varTheta = \begin{pmatrix}
0 & 1 \\
-1 & 0
\end{pmatrix} \otimes \dia{d}.
$ Then the Gram matrix of the dual lattice is given by
$
\varTheta^{\perp} = \begin{pmatrix}
0 & 1 \\
-1 & 0
\end{pmatrix} \otimes \dia{d}^{-1}.
$

The corresponding lattice generator is given by $S = J (T^{\intercal})^{-1}$, where $T^{\intercal} J_{\cv} T = \varTheta$. Let us choose $T$ to be the standard lattice matrix $T_{\dia{d}}$ associated with $\dia{d}$, given by
\begin{equation} \label{osc_sta_T} 
T_{\dia{d}} = \dia{1}_2\otimes\dia{\sqrt{d}}.
\end{equation} 

Any other $T$ of lattice type $\dia{d}$ differs from $T_{\dia{d}}$ by a symplectic transformation; that is, $T = W T_{\dia{d}}$ where $W$ is symplectic. Define
$S_{\dia{d}} := J (T_{\dia{d}}^{\intercal})^{-1}, \quad S := J (T^{\intercal})^{-1}.$
Then we have:
$S = J (T^{\intercal})^{-1} = J ((W T_{\dia{d}})^{\intercal})^{-1} = J W^{-\intercal} T_{\dia{d}}^{-\intercal} = W S_{\dia{d}}.
$

Denote by $\mathbf{e}_i^{\perp}$ the element $S_{\dia{d}}(e_i)$ for $i = 1, 2, \ldots, 2p$, and define
\begin{equation} 
\overline{Z}_i = D\left(\mathbf{e}_i^{\perp}\right), \quad \overline{X}_i = D\left(\mathbf{e}_{i+p}^{\perp}\right),
\end{equation} 
for $i = 1, 2, \ldots, p$. Then, using Eq.~\eqref{eq:ccr_displacement_R^n}, we obtain
\begin{equation} 
\overline{X}_i \overline{Z}_i = \exp\left(2\pi i / d_i\right) \overline{Z}_i \overline{X}_i.
\end{equation} 
A simple calculation shows that $\overline{X}_i^{d_i}$ and $\overline{Z}_i^{d_i}$ belong to the stabilizer group. Hence, the operators $\overline{X}_i$ and $\overline{Z}_i$ for $i = 1, 2, \ldots, p$ define qudit systems.

Let us now consider a simple example: the square lattice GKP code. Take
$\varTheta = \begin{pmatrix}
0 & 2 \\
-2 & 0
\end{pmatrix}.$
We may choose
\begin{equation} 
T = \begin{pmatrix}
\sqrt{2} & 0 \\
0 & \sqrt{2}
\end{pmatrix}, \quad S = \begin{pmatrix}
0 & \frac{1}{\sqrt{2}} \\
-\frac{1}{\sqrt{2}} & 0
\end{pmatrix}.
\end{equation} 
This yields logicals $\overline{Z} = D(S(e_1))$ and $\overline{X} = D(S(e_2))$.

\subsection{Lattices inside $(\mathbb{R} \times \mathbb{Z}_c) \times \widehat{(\mathbb{R} \times \mathbb{Z}_c)}$}
\label{app:general-single-mode-integer-case}

In this section, we classify all the LCA codes for $n = 2$.
We know that the entries of the antisymmetric matrix $Z$ are rational. Hence, for $n = 2$, the matrix $Z$ is determined by a single rational number, which we denote by $-d/c =: z$, where $c, d \in \mathbb{Z}$, $c > 0$, and $\textnormal{gcd}(c,d) = 1$. 
Therefore, 
\begin{equation}
Z = \begin{pmatrix} 0 & z \\ -z & 0 \end{pmatrix}.
\end{equation}

Let \(\varTheta = \begin{pmatrix} 0 & \theta \\ -\theta & 0 \end{pmatrix}\) be an integer matrix with $\theta \neq z$ so that $\varTheta-Z$ is invertible.
We are free to pick the oscillator part of the encoding, \(T_{\cv}\), as long as
\begin{equation}
T_{\cv}^{\intercal} J_{\cv}T_{\cv}=\varTheta-Z=\left(\begin{array}{cc}
0 & \theta+\frac{d}{c}\\
-\left(\theta+\frac{d}{c}\right) & 0
\end{array}\right)~.
\end{equation}
Letting \(R = \left(\begin{array}{cc} 1 & 0 \\ 0 & 1 \end{array}\right)\) in the proof of Theorem~\ref{thm:main_dual}, we have \(T_{\dv } = \left(\begin{array}{cc} -d & 0 \\ 0 & 1 \end{array}\right)\) which yields the full encoding map \(T = \left(\begin{array}{cc} T_{\cv} \\ T_{\dv } \end{array}\right)\). Then the code is generated by the stabilizers $D(T(e_1))$ and $D(T(e_2))$ (with phases) which act on the Hilbert space $L^2(\R \times \Z_c).$  A direct computation shows that the dimension of the code is given by $|c\theta + d|$. For a general $R,$ which is necessary in $\mathrm{Sp}(2, \mathbb{Z}),$ the map \(T_{\dv } = \left(\begin{array}{cc} -d & 0 \\ 0 & 1 \end{array}\right)R.\) In this part, we will assume \(R = \left(\begin{array}{cc} 1 & 0 \\ 0 & 1 \end{array}\right)\). The general case is similar. 

In the case $\theta+\frac{d}{c}>0,$ we call \(T_{\cv} = \sqrt{\theta+\frac{d}{c}}\dia{1}_{2}\), $R=\dia{1}_{2}$ the standard lattice matrices. If $\theta+\frac{d}{c}<0,$ then by conjugating $\varTheta-Z$ by \(R=\left(\begin{array}{cc} 0 & 1 \\ 1 & 0 \end{array}\right)\) as in Eq.~\eqref{eq:standard_R}, we can assume that $\theta+\frac{d}{c}>0.$

We will first show that for a fixed $\theta$, different values of $z$ lead to distinct codes.

When $c = 1$, a quick computation shows that the code corresponds to the usual 2-dimensional GKP code (see previous subsection) with logical dimension $|\theta + d|$. Hence, all such codes are distinct. When $c \neq 1$, Theorem~\ref{thm:main_dual} yields codes inside the Hilbert space $L^2(\mathbb{R} \times \mathbb{Z}_c)$. If two rational numbers $-d_1/c_1$ and $-d_2/c_2$ yield the same code, then $c_1 = c_2$. Comparing the logical dimensions, we get $d_1 = d_2$. Hence, we conclude that each rational number $-d/c$ gives a different code for a fixed $\theta$.

We now determine the logical operators of the code corresponding to a rational number $z = -d/c$. To do this, we compute the map $S$ and the operators $V_1 = D(S(e_1))$ and $V_2 = D(S(e_2))$. Since $-d$ and $c$ are coprime, there exist integers $a$ and $b$ such that $-ad + bc = 1$. From Theorem~\ref{thm:main_dual}, we compute $g$ and $S$ for $z = -d/c$ ($d \neq 0$) as:
\begin{equation}\label{eq:g'forz}
    g := \begin{pmatrix} A & B \\ C & D \end{pmatrix} \in \mathrm{SO}(2, 2 \mid \mathbb{Z}),~\text{for}~ A = \begin{pmatrix} 0 & -a \\ a & 0 \end{pmatrix}, \quad
B = \begin{pmatrix} b & 0 \\ 0 & b \end{pmatrix}, \quad
C = \begin{pmatrix} c & 0 \\ 0 & c \end{pmatrix}, \quad
D = \begin{pmatrix} 0 & d \\ -d & 0 \end{pmatrix}~,
\end{equation}
and with
\begin{equation}
S = \begin{pmatrix}
S_{\cv} \\
\begin{pmatrix} 0 & -1 \\ a & 0 \end{pmatrix}
\end{pmatrix},
\end{equation}
where $S_{\cv}$ is a $2 \times 2$ matrix given by
\begin{equation}
S_{\cv} = J_{\cv} (T_{\cv}^\intercal)^{-1} \begin{pmatrix} 1/c & 0 \\ 0 & 1/c \end{pmatrix}.
\end{equation}
Then $S$ is an embedding map such that
\begin{equation}
-\varTheta^{\perp} = S^{\intercal} J S = \begin{pmatrix} 0 & \frac{a\theta + b}{c\theta + d} \\ -\frac{a\theta + b}{c\theta + d} & 0 \end{pmatrix}.
\end{equation}

It follows that the operators $V_1 = D(S(e_1))$ and $V_2 = D(S(e_2))$ satisfy the commutation relation:
\begin{equation}
V_1 V_2 = e^{-2\pi i \frac{a\theta + b}{c\theta + d}} V_2 V_1.
\end{equation}
Now, since the logical dimension is $|c\theta + d|$, if $a\theta + b = 1$, then we can take
\begin{equation}
\overline{Z} = V_1 = D(S(e_1)), \quad \overline{X} = V_2 = D(S(e_2)).
\end{equation}
If $a\theta + b \neq 1$, we first show that $a\theta + b$ and $c\theta + d$ are coprime. Define:
\begin{equation}
x := a\theta + b, \quad y := c\theta + d.
\end{equation}
Suppose a prime $p$ divides both $x$ and $y$:
\begin{equation}
p \mid x, \quad p \mid y.
\end{equation}
Then:
\begin{equation}
a\theta \equiv -b \mod p, \quad c\theta \equiv -d \mod p.
\end{equation}
Multiply the first congruence by $d$ and the second by $b$:
\begin{equation}
ad\theta \equiv -bd \mod p, \quad bc\theta \equiv -bd \mod p.
\end{equation}
Subtracting yields:
\begin{equation}
(ad - bc)\theta \equiv 0 \mod p.
\end{equation}
Since $ad - bc = -1$, we obtain $\theta \equiv 0 \mod p$. Substituting back,
\begin{equation}
x \equiv b \mod p, \quad y \equiv d \mod p.
\end{equation}
Thus, $p \mid b$ and $p \mid d$, and so $p \mid ad - bc = -1$, which implies $p = 1$ --- a contradiction. Hence $\gcd(x, y) = 1$.

Therefore, there exists an integer $k$ such that $k(a\theta + b) \equiv 1 \mod (c\theta + d)$. In this case, $(V_1)^k$ and $V_2$ satisfy the correct commutation relation. Hence, we can take
\begin{equation}
\overline{Z} = (V_1)^k = (D(S(e_1)))^k, \quad \overline{X} = V_2 = D(S(e_2)).
\end{equation}

\begin{remark}
    Note that in the above construction $g$ is determined by the matrix $\mathfrak{g}\in \begin{pmatrix} -a & -b \\ c & d \end{pmatrix} \in \rm{Sp}(2,\Z).$ 
\end{remark}

\section{Clifford Gates of LCA codes} 
\label{app:gates}

In this section, we will discuss the Clifford gates for the LCA codes in the cases \( M = \mathbb{R}^p \) and \( M = \mathbb{R} \times \mathbb{Z}_c \). These Clifford gates arise from certain automorphisms of noncommutative tori. The reader is referred to~\cite{elpw2010structure, jeong2015finite, chakluef2019metaplectic, chaktracing2023tracing, chaksymmet2025symmetrized} for a comprehensive study of these automorphisms in the setting of noncommutative tori.

\subsection{The Case \( M = \mathbb{R}^p \)}

Recall from Subsection~\ref{sec:rp} that, in the case when \( M = \mathbb{R}^p \), we have \( g = \begin{pmatrix}
0 & \dia{1}_{2p} \\
\dia{1}_{2p} & 0
\end{pmatrix} \), and the embedding map \( T = T_{\cv} \) is invertible. For an integer symplectic Gram matrix \( \varTheta \) with \( \text{det}(\varTheta) \neq 0 \), we define the following group:
\begin{equation} 
\operatorname{Aut}(\varTheta) = \left\{ W_{\varTheta} \in \mathrm{GL}(n, \mathbb{Z}) \mid W_{\varTheta}^{\intercal} \varTheta W_{\varTheta} = \varTheta \right\}
\end{equation} 

\begin{lemma}
If \( W_{\varTheta} \in \operatorname{Aut}(\varTheta) \), then \( W = T W_{\varTheta} T^{-1} \) is a symplectic matrix inside \( \operatorname{Sp}(2p, \mathbb{R}) \).
\end{lemma}

\begin{proof}
This is straightforwardly done by verifying the symplectic condition \(W^{\intercal}J_{\cv}W = J_{\cv}\).
\end{proof}

For any element \( W \in \mathrm{Sp}(2p, \mathbb{R}) \), we can choose an operator \( \widehat{W} \) from the complex metaplectic group \( \mathrm{Mp}(2p, \mathbb{R}) \) up to a phase. It is well known (see~\cite[Thm. 128]{de2011symplectic}) that
\begin{equation} \label{eq:meta_prp_R}
\widehat{W} A(l) D(T(l)) \widehat{W}^{\dagger} = A(WT(l)) D(WT(l)), \quad l \in \mathbb{Z}^{2p}.
\end{equation} 
Using $WT=TW_\varTheta$ we get
\begin{equation}\label{eq:metaplectic_property_1}
    \widehat{W} U_l \widehat{W}^{\dagger} = U_{W_{\varTheta} l}, \quad l \in \mathbb{Z}^{2p}.
    \end{equation}

Given an element in $\mathrm{Aut}(\varTheta)$, we now define an action on the logical operators. First, we define the subgroup
\begin{equation}\label{eq:group}
\mathcal{R} := \left\langle \rho(R) = \begin{pmatrix} R & 0 \\ 0 & (R^{-1})^{\intercal} \end{pmatrix} ~\bigg |~ R \in \mathrm{GL}(n, \mathbb{Z}) \right\rangle \subseteq \mathrm{SO}(n, n \mid \mathbb{Z}).
\end{equation}

Let $W_{\varTheta} \in \mathrm{Aut}(\varTheta)$ and assume that 
\begin{equation}
g \rho\left(W_{\varTheta}^{\intercal}\right) g^{-1} \subset \mathcal{R}
\end{equation}
inside $\mathrm{SO}(n, n \mid \mathbb{Z})$. Then $W_{\varTheta}$ defines an element in $\mathrm{Aut}(-g \cdot \varTheta)$. Indeed, since $g \rho\left(W_{\varTheta}^{\intercal}\right) g^{-1} \subset \mathcal{R}$, we have 
\begin{equation}
g \rho\left(W_{\varTheta}^{\intercal}\right) g^{-1} = \begin{pmatrix} S & 0 \\ 0 & (S^{-1})^{\intercal} \end{pmatrix}
\end{equation}
for some $S \in \mathrm{GL}(n, \mathbb{Z})$. Then $S^{\intercal}$ is an element of $\mathrm{Aut}(-g \cdot \varTheta)$: 
\begin{equation}
(S^{\intercal})^{\intercal} (g \cdot \varTheta) S^{\intercal} = S (g \cdot \varTheta) S^{\intercal} = (g \rho\left(W_{\varTheta}^{\intercal}\right) g^{-1}) \cdot (g \cdot \varTheta) = g \rho\left(W_{\varTheta}^{\intercal}\right) \cdot \varTheta = g \cdot \varTheta.
\end{equation}

Now, if $W_\varTheta \in \operatorname{Aut}(\varTheta)$ and 
$g = \begin{pmatrix}
0 & \dia{1}_{2p} \\
\dia{1}_{2p} & 0
\end{pmatrix}$, 
one can easily verify that $S^{\intercal} = (W_{\varTheta}^{-1})^{\intercal}$. Hence, $(W_{\varTheta}^{-1})^{\intercal} \in \operatorname{Aut}(-\varTheta^{\perp}).$  Therefore, $(W_{\varTheta}^{-1})^{\intercal}$ implements the desired action on the logical side. More explicitly, if \( \boldsymbol{\sigma} \in (\mathbb{Z}_{d_1} \times \cdots \times \mathbb{Z}_{d_p})^{\times 2} \) is the symplectic representation of the GKP logical Pauli operator \( P(\boldsymbol{\sigma}) = \prod_i \bar{Z}_i^{\hat{\sigma}_i} \bar{X}_i^{\hat{\sigma}_{i+p}} \) up to a phase (as described in Subsection~\ref{sec:rp}), then \( \boldsymbol{\sigma} \to (W_{\varTheta}^{-1})^{\intercal} \boldsymbol{\sigma} \mod \dia{d} \), where ``mod \( \dia{d} \)" means that the \( i \)-th and \( p+i \)-th rows of a vector are taken modulo \( d_i \).

The implementing unitary for the above logical action on the code space is given by the unitary operator \( \widehat{W} \).
One may need to apply a displacement operator in addition to achieve the correct action (see the issue mentioned on Page 9 of Ref.~\cite{fiber-bundle-ft}).

We now outline how to prepare the codewords for the code with respect to the symplectic Gram matrix \( \varTheta \) and a given \( T \) such that \( T^{\intercal} J T = \varTheta \). Suppose we can prepare a codeword \( \ket{\psi} \) for \( \varTheta_{\dia{d}} \) with the embedding map \( T_{\dia{d}} \), where \( \varTheta_{\dia{d}} \) is the standard form of \( \varTheta \), and \( T_{\dia{d}} \) is the standard lattice matrix. Then we know that \( (T_{\dia{d}})^{\intercal} J T_{\dia{d}} = \varTheta_{\dia{d}} \), and \( V^{\intercal} \varTheta V = \varTheta_{\dia{d}} \), for some matrix \( V \in \mathrm{GL}(2p, \mathbb{Z}) \). The matrix \( S = T V T_{\dia{d}}^{-1} \) is a symplectic matrix, and the corresponding metaplectic operator \( \widehat{S} \) transforms \( \ket{\psi} \) into the desired codeword. 

To prove the above, first note that $D(T_{\dia{d}}(l))\ket{\psi}=\ket{\psi},$ $l\in \Z^n.$ We have to show that \( \widehat{S}\ket{\psi} \) is a stabilized by $D(T(l))$ for all $l\in \Z^n.$ Now \begin{equation}\widehat{S}^\dagger D(T(l))\widehat{S}\ket{\psi} \propto  D(S^{-1}T(l))\ket{\psi}\propto D(T_{\dia{d}}V^{-1}(l))\ket{\psi}\propto \ket{\psi},\end{equation} since $V^{-1}(l)\in \Z^n.$ This proves our claim.

\subsection{The Case \( M = \mathbb{R} \times \Z_c \)}
\label{app:single-mode-symplectic}

Recall that in this case, the code is determined by a rational number $z = -d/c$. Now, start with a parameter $\theta$ such that $\theta \neq z$, so that the code is well-defined.

Assume that $W_\theta \in \mathrm{Aut}(\varTheta)$, where
\begin{equation}
\varTheta = \begin{pmatrix} 0 & \theta \\ -\theta & 0 \end{pmatrix}, \quad W_\theta = \begin{pmatrix} p & q \\ r & s \end{pmatrix}.
\end{equation}
By definition, $W_\theta$ is in the integer symplectic group $\mathrm{Sp}(2, \mathbb{Z})$. Using Eq.~\eqref{eq:g'forz}, one checks that
\begin{equation}
g' \rho\left(W_{\theta}^{\intercal}\right) g'^{-1} \subset \mathcal{R},
\end{equation}
where the group \(\cal R\) is defined in Eq.~\eqref{eq:group}, and where the corresponding action on the logical side is given by $(W_{\theta}^{-1})^{\intercal}$, as in the previous subsection.

We now aim to determine the unitary operator $W$ on $L^2(\mathbb{R} \times \mathbb{Z}_c)$ that implements the above action. To this end, we seek an analogue of Eq.~\eqref{eq:meta_prp_R} for this setting, and want to find the corresponding matrix $W$. Recall that in this case, $T = \begin{pmatrix} T_{\cv} \\ T_{\dv } \end{pmatrix}$, where $T_{\cv}$ satisfies $T_{\cv}^{\intercal} J_{\cv} T_{\cv} = \varTheta - Z$, and $T_{\dv } = \begin{pmatrix} -d & 0 \\ 0 & 1 \end{pmatrix}R$.

Now we claim that such a $W$ exists. Indeed, motivated by~\cite[Section 3]{2505.19869}, $W$ can be taken as
\begin{equation}\label{eq:symp_hybrid}
W = \begin{pmatrix} T_{\cv} W_\theta T_{\cv}^{-1} & 0 \\ 0 & T_{\dv} W_\theta T_{\dv}^{-1} \end{pmatrix},
\end{equation}
where
we have assumed that $T_{\dv} W_\theta T_{\dv}^{-1}=:W_c$ is an integer matrix. This condition assures that the entries of $W_c$ can be taken modulo $c$, so that $W_c \in \mathrm{Sp}(2, \mathbb{Z}_c).$ For $R=\dia{1}_{2},$ the integrality of $W_c$ boils down to the condition  $d \mid r.$ 
Then $W$ preserves the symplectic form $J$, i.e.,
\begin{equation}
W^{\intercal} J W = J, \quad \text{where} \quad J = \begin{pmatrix} J_{\cv} & 0 \\ 0 & J_{\dv } \end{pmatrix}.
\end{equation}

It is also straightforward to verify that
\begin{equation}
W T(l) = T(W_\theta(l)), \quad \text{for all } l \in \mathbb{Z}^2.
\end{equation}

To obtain an analogue of Eq.~\eqref{eq:meta_prp_R}, note that our $W$ does not entangle the qudit and the oscillator, so one can apply Eq.~\eqref{eq:meta_prp_R} separately to both parts.
For the qudits, the analogous equation is given in~\cite{10.1093/qmath/ham023} or in~\cite[Lemma 2.1, (1)]{zbMATH07374495}. It should be noted that due to the phase factors appearing in~\cite[Lemma 2.1, (1)]{zbMATH07374495}, a suitable ``Pauli correction"  is necessary (see, for example,~\cite[Page 9]{fiber-bundle-ft}).

Finally, to obtain an analogue of Eq.~\eqref{eq:metaplectic_property_1} in this setting, one can use the identity
\begin{equation}
W T(l) = T(W_\theta(l)),
\end{equation}
just as in the case of $\mathbb{R}^p$.

A standard form in this case, along with the corresponding metaplectic operator, can be obtained similarly to the case $M = \mathbb{R}^p$, by replacing $\varTheta$ with $\varTheta - Z$.

Consider the following subgroup of $\operatorname{Sp}(2,\mathbb{Z})$:
\begin{equation}
\Gamma_0(d) := \left\{
W = \begin{pmatrix}
p & q \\
r & s
\end{pmatrix}
\;\middle|\;
W \in \operatorname{Sp}(2,\mathbb{Z}), \; d \mid r
\right\}.
\end{equation}

\noindent This is known as the \textit{congruence subgroup} of $\operatorname{Sp}(2,\mathbb{Z})$ (see~\cite[Definition~2.1.6]{ZerbesModularForms2022}). From~\cite[Lemma~3.2.1]{ZerbesModularForms2022}, it follows that the reduction map (taking mod $c$)
\[
\Gamma_0(d) \longrightarrow \mathrm{Sp}(2,\mathbb{Z}_c)
\]
is surjective. Note that the proof of~\cite[Lemma~3.2.1]{ZerbesModularForms2022} also holds when $c$ and $d$ are coprime; it is not necessary to assume that $c$ is prime.
 This has the following consequence:

Suppose we want to implement a logical gate $\widetilde{W} \in \mathrm{Sp}(2,\mathbb{Z}_c)$ at the qubit level. We can always lift $\widetilde{W}$ to an element of $\Gamma_0(d)$ (which we again denote by $\widetilde{W}$). Now, define
\[
W_\theta := (\widetilde{W}^{-1})^{\intercal}.
\]
Then the metaplectic operator corresponding to the matrix $W$, as in formula~\eqref{eq:symp_hybrid}, provides the physical operation that implements $\widetilde{W}$ at the logical level. A similar analysis can be carried out for a general $R.$

The above analysis can be easily extended to this case:
\begin{equation}
T_{\dv } = 
\begin{pmatrix}
\dia{1}_k & 0 & 0 \\
0 & \dia{1}_k & 0
\end{pmatrix}
R,~.
\end{equation}
Here, we do not require any condition such as $d \mid r$ as in the previous case. Note that metaplectic operators for general finite abelian groups have been studied in~\cite{KAIBLINGER2009233}, where one can find an analogue of Equation~\eqref{eq:meta_prp_R} for the case of multiple qudits.

\prg{Oscillator-qudit symplectic transformations}
In the following, we prove that 
\begin{equation}
\operatorname{Sp}(\mathbb{R} \times \mathbb{Z}_c) \cong 
\left\{
\begin{pmatrix}
W & 0 \\
0 & V
\end{pmatrix}
\;\middle|\;
W \in \operatorname{Sp}(2, \mathbb{R}),\;
V \in \operatorname{Sp}(2, \mathbb{Z}_c)
\right\}
\cong \operatorname{Sp}(2, \mathbb{R}) \times \operatorname{Sp}(2, \mathbb{Z}_c).
\end{equation}

First, we note that
\begin{equation}
\operatorname{Sp}(\mathbb{R} \times \mathbb{Z}_c) \subseteq \operatorname{Aut}(\mathbb{R}^2 \times \mathbb{Z}_c^2).
\end{equation}
Since the groups involved are abelian, any $f \in \operatorname{Aut}(\mathbb{R}^2 \times \mathbb{Z}_c^2)$ can be written in block matrix form as
\begin{equation}
f = \begin{pmatrix}
f_{11} & f_{12} \\
f_{21} & f_{22}
\end{pmatrix},
\end{equation}
where each element is a group homomorphism,
\begin{equation}
f_{11} \in \operatorname{Hom}(\mathbb{R}^2, \mathbb{R}^2), \quad
f_{12} \in \operatorname{Hom}(\mathbb{R}^2, \mathbb{Z}_c^2), \quad
f_{21} \in \operatorname{Hom}(\mathbb{Z}_c^2, \mathbb{R}^2), \quad
f_{22} \in \operatorname{Hom}(\mathbb{Z}_c^2, \mathbb{Z}_c^2).
\end{equation}

Using the standard fact that a finite group cannot be mapped into the real line and vice versa,
\begin{equation}
\operatorname{Hom}(\mathbb{R}^2, \mathbb{Z}_c^2) = 0 = \operatorname{Hom}(\mathbb{Z}_c^2, \mathbb{R}^2),
\end{equation}
we conclude that $f_{12} = 0$ and $f_{21} = 0$. 
Thus, $f$ is of the form
\begin{equation}
f = \begin{pmatrix}
f_{11} & 0 \\
0 & f_{22}
\end{pmatrix},
\end{equation}
with $f_{11} \in \operatorname{Aut}(\mathbb{R}^2)$ and $f_{22} \in \operatorname{Aut}(\mathbb{Z}_c^2)$.
Moreover, if $f \in \operatorname{Sp}(\mathbb{R} \times \mathbb{Z}_c)$, then $f$ must preserve the symplectic form, which implies
\begin{equation}
f_{11} \in \operatorname{Sp}(\mathbb{R}^2) = \operatorname{Sp}(2, \mathbb{R}), \quad
f_{22} \in \operatorname{Sp}(\mathbb{Z}_c^2) = \operatorname{Sp}(2, \mathbb{Z}_c).
\end{equation}
This completes the proof.

\section{Standard form of $Z$ and codeword preparation}
\label{app:state-prep}

Let \( T \) be such that \( T_{\cv}^{\intercal} J T_{\cv} = Z \). Suppose we can prepare a codeword \( \ket{\psi} \) for \( Z_{\dia{d}} \) with the embedding map \( T_{\dia{d}}=\binom{T_{{\dia{d}},\cv}}{T_{\dia{d},\dv }} \), where \( Z_{\dia{d}} \) is the standard form of \( Z \), \( T_{\dia{d},\cv} \) is the standard lattice matrix and $T_{\dia{d},\dv }= \left(\begin{array}{ccc}
\dia{d} & 0 & 0 \\
0 & \dia{1}_k & 0
\end{array}\right).$ Then we know that \( (T_{\dia{d},\cv})^{\intercal} J_{\cv} T_{\dia{d},\cv} = Z_{\dia{d}} \), and \( R^{\intercal}  Z R= Z_{\dia{d}} \), for some matrix \( R \in \mathrm{GL}(2p, \mathbb{Z}) \). 
Also recall 
\begin{equation}
T^L=\binom{T_{\cv}}{LT_{\dia{d},\dv }R^{-1}} ~,    
\end{equation}
for some symplectic matrix $L$ on the qudit side.  The matrix 
\begin{equation}\label{eq:standard_form_symplectic}
  S = \left(\begin{array}{ccc}
S_{\cv} & 0  \\
0 & L 
\end{array}\right) = \left(\begin{array}{ccc}
T_{\cv} R T_{\dia{d},\cv}^{-1} & 0  \\
0 & L
\end{array}\right)    
\end{equation}
is a symplectic matrix in the hybrid system that relates \(T_{\dia{d},\cv}\) to \(T_{\cv}\) as
\begin{subequations}
\begin{align}
    (T_{\dia{d},\cv})^{\intercal} S_{\cv}^{\intercal} J_{\cv}S_{\cv}T_{D,\cv}&=(T_{\dia{d},\cv})^{\intercal} (T_{\dia{d},\cv}^{-1})^{\intercal} R^{\intercal} (T_{\cv})^{\intercal} J_{\cv}T_{\cv}RT_{\dia{d},\cv}^{-1}T_{\dia{d},\cv}\\&=R^{\intercal} T_{\cv}^{\intercal} J_{\cv}T_{\cv}R\\&=R^{\intercal} ZR\\&=Z_{\dia{d}}~.
\end{align}
\end{subequations}
In Eq.~\eqref{eq:standard_form_symplectic}, $L$ is viewed as an element of the symplectic group of $\Z_{\dia{c}}.$ 

The corresponding metaplectic operator \( \widehat{S} \) then transforms \( \ket{\psi} \) into the desired codeword. 
To prove this, first note that $D(T_{\dia{d}}(l))\ket{\psi}=\ket{\psi},$ $l\in \Z^n.$ We have to show that \( \widehat{S}\ket{\psi} \) is a stabilized by $D(T^L(l))$ for all $l\in \Z^n.$ Now \begin{equation}
\widehat{S}^\dagger D(T^L(l))\widehat{S}\ket{\psi} \propto  D(S^{-1}T^L(l))\ket{\psi}\propto D(T_{\dia{d}}R^{-1}(l))\ket{\psi}\propto \ket{\psi},
\end{equation}
since $R^{-1}(l)\in \Z^n.$

\end{document}